\documentclass[conference]{IEEEtran}
\IEEEoverridecommandlockouts

\usepackage{cite}
\usepackage{amsmath,amssymb,amsfonts}
\usepackage{algorithm}
\usepackage[noend]{algpseudocode}
\usepackage{graphicx}
\usepackage{textcomp}
\usepackage{xcolor}
\usepackage{url}
\usepackage{subfigure}
\usepackage{threeparttable}
\usepackage{multirow}
\usepackage{array}
\newtheorem{definition}{Definition}
\newtheorem{theorem}{Theorem}
\newtheorem{proof}{Proof}

\def\BibTeX{{\rm B\kern-.05em{\sc i\kern-.025em b}\kern-.08em
    T\kern-.1667em\lower.7ex\hbox{E}\kern-.125emX}}
\begin{document}

\title{RISK: Efficiently processing rich spatial-keyword queries on encrypted geo-textual data
\thanks{This paper is accepted by 42nd IEEE International Conference on Data Engineering (ICDE 2026). The source code is available at https://github.com/ygpeng-xd/RISK. \\
This work is supported by the National Natural Science Foundation of China (No. 62572378, 62272358, 62372352). We thank Huaxing Tong and Yilin Zhang for supplementary experiments, and the reviewers and Area Chair for their insightful comments.\\
$^*$Yanguo Peng is the corresponding author.}
}

\author{\IEEEauthorblockN{Zhen~Lv$^{1,2}$, Cong~Cao$^{2}$, Hongwei~Huo$^{2}$, Jiangtao~Cui$^{2,3}$, Yanguo~Peng$^{2,*}$, Hui~Li$^{2}$, Yingfan~Liu$^{2}$}
\IEEEauthorblockA{\textit{$^{1}$School of Artificial Intelligence, Xi'an International Studies University, Xi'an, China} \\
\textit{$^{2}$School of Computer Science and Technology, Xidian University, Xi'an, China}\\
\textit{$^{3}$Xi'an University of Posts \& Telecommunications, Xi'an, China}
\\
lvzhen@xisu.edu.cn, xducaoc@stu.xidian.edu.cn, hwhuo@mail.xidian.edu.cn, \{cuijt,ygpeng,hli,liuyingfan\}@xidian.edu.cn}
}

\maketitle

\begin{abstract}

Symmetric searchable encryption (SSE) for geo-textual data has attracted significant attention. However, existing schemes rely on task-specific, incompatible indices for isolated specific secure queries (e.g., range or \emph{k}-nearest neighbor spatial-keyword queries), limiting practicality due to prohibitive multi-index overhead. To address this, we propose RISK, a model for \underline{ri}ch \underline{s}patial-\underline{k}eyword queries on encrypted geo-textual data. In a textual-first-then-spatial manner, RISK is built on a novel \emph{k}-nearest neighbor quadtree (\emph{k}Q-tree) that embeds representative and regional nearest neighbors, with the \emph{k}Q-tree further encrypted using standard cryptographic tools (e.g., keyed hash functions and symmetric encryption). Overall, RISK seamlessly supports both secure range and \emph{k}-nearest neighbor queries, is provably secure under IND-CKA2 model, and extensible to multi-party scenarios and dynamic updates. Experiments on three real-world and one synthetic datasets show that RISK outperforms state-of-the-art methods by at least $0.5$ and $4$ orders of magnitude in response time for $1\%$ range queries and $10$-nearest neighbor queries, respectively. 
\end{abstract}


\section{Introduction}
\label{RISK:Intro}
In the cloud computing era, organizations increasingly leverage public clouds to build  high-quality, cost-effective data-driven service capabilities. IBM, for instance, reduced its co-selling lifecycle by 90\% and raised sales opportunities with AWS by 117\% through the Labra Platform~\cite{URLIBM}, while Uber modernized batch data infrastructure on Google Cloud Platform to drive substantial gains in user productivity and engineering velocity~\cite{URLUber}. Nevertheless, 94\% of organizations voice concerns over cloud security risks, a challenge mitigated by symmetric searchable encryption (SSE) solutions tailored to heterogeneous data that support direct queries on encrypted data. Across the diverse SSE landscape, secure spatial-keyword queries~\cite{Chen2020,Chen2021} have pioneered a novel pathway to balance security and availability, a progress driven by the rapid proliferation of social and location-based applications. These queries enable organizations including Twitter, Uber and Lyft to deliver on-the-fly complex data analysis capabilities to consumers while protecting encrypted geo-textual data from unauthorized leakage and malicious exploitation.

Existing solutions for secure spatial-keyword queries in the ciphertext domain focus on either range spatial-keyword (RSK) queries~\cite{Wang2021TIFS,Lv2023} or \emph{k}-nearest neighbor spatial-keyword (\emph{k}SK) queries~\cite{Song2024IoTsJ} with tailored indexing structures. Cui et al.~\cite{Cui2019ICDE} pioneered encrypted Bloom filters in 2019 to support secure boolean RSK queries, which return only identifiers of target results residing in a specific region and matching given keywords. Subsequent works have focused on efficiency improvement~\cite{Wang2020INFOCOM} and security enhancement~\cite{Wang2021TIFS,Yang2022ICDCS}. Lv et al.~\cite{Lv2023} proposed a RSK query solution in 2023 that preserves equivalent security while returning both identifiers and content of target results, which distinguishes it from prior methods. Song et al.~\cite{Song2024IoTsJ} extended this encrypted geo-textual query research one year later with a scheme for boolean \emph{k}SK queries.

Unfortunately, these solutions lack generalization to support multiple types of secure spatial-keyword queries, among which RSK and \emph{k}SK queries are equally critical for location-related applications. Specifically, \emph{k}-nearest neighbors (\emph{k}NNs) of a given query do not  lie within predefined spatial ranges, rendering indexing structures tailored for secure RSK queries incompatible with their \emph{k}SK counterparts. Secure RSK indices inherently encode range-specific semantics inapplicable to \emph{k}SK queries, and vice versa. Consequently, organizations requiring both query services must deploy two fully orthogonal indexing structures, a configuration that doubles storage and maintenance overheads. In contrast, a unified secure index natively supporting both RSK and \emph{k}SK queries would substantially mitigate such overheads.

\begin{figure}[t!]
	\centerline{\includegraphics[width=85mm]{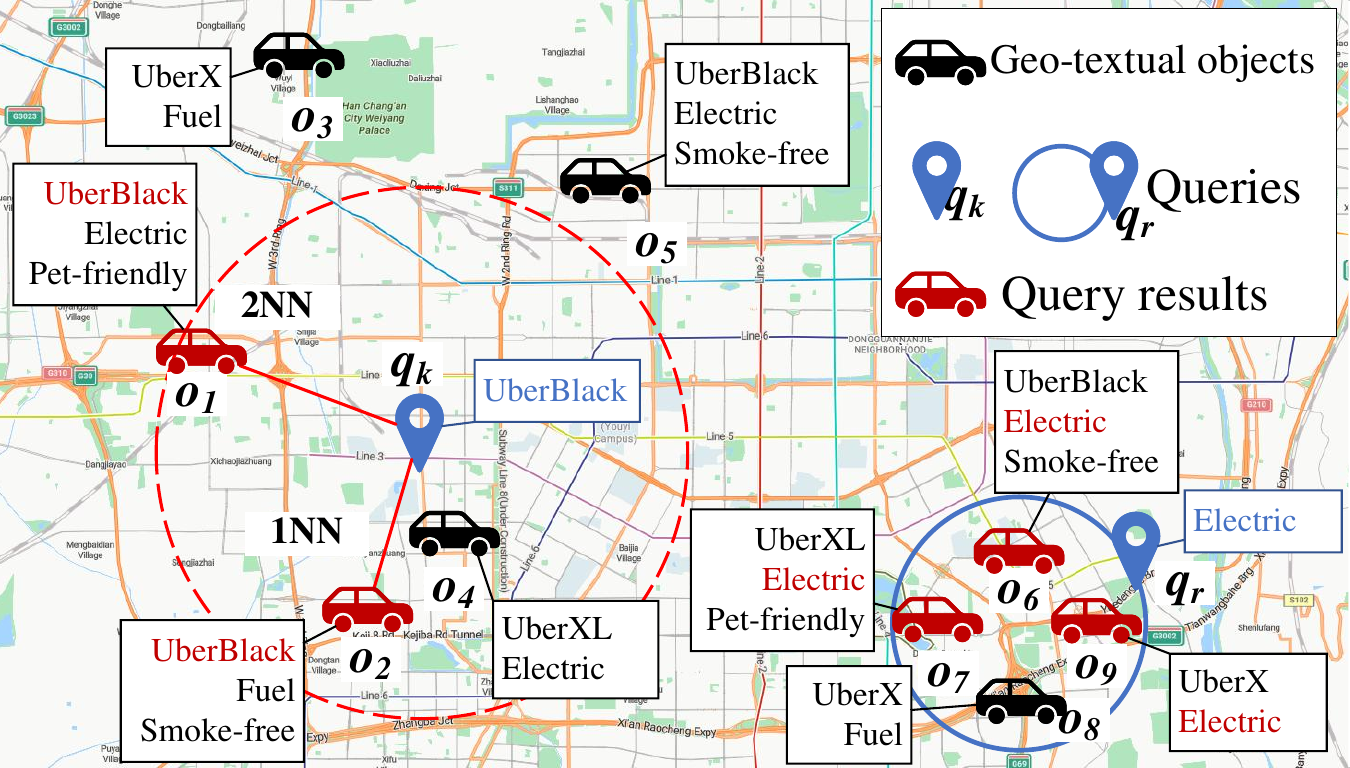}}
	\caption{Application scenarios for range and $2$NN spatial-keyword queries.}
	\label{RISK:fig:Scenario}
\end{figure}


\emph{Application scenarios}. Ride-hailing platforms (e.g., Uber and Lyft) often outsource full vehicle and passenger datasets to public clouds as data owners, minimizing hardware investment and maintenance overheads. Fig.~\ref{RISK:fig:Scenario} illustrates a ride-hailing scenario enabled by spatial-keyword queries. After outsourcing, these platforms act as clients to access the encrypted data stored on the cloud. Car icons and pin markers in the figure denote vehicles and passengers respectively, with both entities tagged with textual attributes specifying vehicle models, features, passenger preferences and related details. The platform continuously refines order allocation through two secure primitives. For the secure RSK query $q_r$, the platform retrieves all vehicles within a given range, analyzes the distribution of vehicles matching a specific model, and  dispatches these vehicles to enable rapid passenger pickups. Here, the platform identifies a high concentration of electric vehicles within the blue continuous circled area and subsequently reallocates several of these vehicles to other regions. For the secure \emph{k}SK query $q_k$, the platform directs the \emph{k} nearest vehicles to serve passengers requesting UberBlack, with either $o_1$ or $o_2$ assigned the order given that both are among the nearest options to the passenger. However, existing solutions for secure spatial-keyword queries cannot simultaneously support these two tasks. Designing a framework for secure rich spatial-keyword queries remains challenging.


To this end, we propose a secure \underline{ri}ch \underline{s}patial \underline{k}eyword query (RISK) model, built on a novel \emph{k}-nearest neighbor quadtree (\emph{k}Q-tree). The main contributions are as follows.

\begin{itemize}
  \item A hybrid index \emph{k}Q-tree embedding neighborhood and range semantics is constructed following a textual-first-then-spatial paradigm, each subset is expanded to incorporate the nearest neighbors (NNs) of selected representative objects, preserving the semantics  required for effective \emph{k}SK queries.
  
  \item To ensure security, all \emph{k}Q-trees corresponding to distinct keywords are converted into a secure \emph{k}Q-tree (S\emph{k}Q-tree) without property-preserving encryptions (PPEs). Structured as a key-value list, the S\emph{k}Q-tree conceals both keys and values via keyed hash functions and symmetric encryption to guarantee confidentiality.

  \item Built on the S\emph{k}Q-tree, the concrete RISK model supports both secure RSK and \emph{k}SK queries, achieving indistinguishability under adaptive chosen-keyword attack (IND-CKA2). RISK is further extended to multi-user scenarios and dynamic data updates for broader applicability.
  
  \item On three real-world and one synthetic datasets, RISK is evaluated to exhibit its superiority. It reduces query response time by $0.5$ to $4$ orders of magnitude and cuts cloud storage overhead by $1$ to $3$ orders of magnitude compared with state-of-the-art (SOTA) solutions. 
\end{itemize}

In the following, section~\ref{RISK:RelatedWorks} surveys SOTA related work. Section~\ref{RISK:Preliminaries} elaborates on necessary preliminaries. Section~\ref{RISK:Index} presents the detailed design of the \emph{k}Q-tree and S\emph{k}Q-tree. The concrete RISK model is described in section~\ref{RISK:Cons}. Section~\ref{RISK:Analyses} provides analyses of security, computational complexity, and theoretical comparisons. Section~\ref{RISK:Extensions} discusses extensions to multi-party scenarios and data updates. Section~\ref{RISK:Experi} conducts experimental evaluations and comparisons of RISK on four datasets. Finally, Section~\ref{RISK:Conclusion} concludes the paper.

\section{Related Works}
\label{RISK:RelatedWorks}
In 2005, Vaid et al.~\cite{Vaid2005STD} and Zhou et al.~\cite{Zhou2005} independently initiated spatial-keyword query models. Since then, driven by rising demand for  geo-textual data access, this problem has become predominant in social and location-based applications. By widely accepted consensus~\cite{Chen2020,Chen2021,Chen2013VLDB}, both RSK and \emph{k}SK queries are fundamental to such applications. 


\subsection{Progress in range spatial-keyword query}
\label{RISK:RelatedWorks:Range}
Building on the pioneering work~\cite{Vaid2005STD,Zhou2005}, RSK query research has yielded numerous approaches leveraging R-tree~\cite{Gobel2009CIKM}, R*-tree~\cite{Zhou2005}, space filling curve~\cite{Chen2006SIGMOD,Christoforaki2011CIKM,Ma2013WAIM}, and grid division~\cite{Vaid2005STD}. These solutions focus on either index size reducing or query efficiency improvement. As reported in~\cite{Christoforaki2011CIKM}, their performance asymptotically approaches optimality. However, none of the above work accounts for security or privacy.

In 2019, Cui et al.~\cite{Cui2019ICDE} first constructed ELCBFR+, a secure RSK query model built on Bloom filter encrypted via asymmetric scalar product-preserving encryption (ASPE). Unfortunately, Li et al.~\cite{Li2019ICDE} later demonstrated that ASPE itself is insecure against even ciphertext-only attacks, rendering ELCBFR+ fundamentally vulnerable and invalidating its security claims. To improve security, Wang et al.~\cite{Wang2020INFOCOM,Wang2021TIFS} proposed PBRQ and SKSE with symmetric-key hidden vector encryption (SHVE). Subsequently, Yang et al.~\cite{Yang2022ICDCS} introduced SKQ and LSKQ based on enhanced ASPE (EASPE), which resolve ASPE's security issues. However, all the above solutions depend on PPEs (i.e., ASPE, SHVE, and EASPE) that are unsuitable for industrial applications due to inherent security risks and encryption constraints. Recently, Lv et al.~\cite{Lv2023} presented RASK and RASK+, which address secure RSK queries using only symmetric encryption and keyed hash function, delivering significant gains in query response time.


\subsection{Progress in k-nearest neighbor spatial-keyword query}
\label{RISK:RelatedWorks:kNN}
The \emph{k}SK query, which performs \emph{k}NN search with textual constraints, was first addressed by De Felipe et al.~\cite{DeFelipe2008ICDE} in 2008 using an R-tree. Since then, various approaches leveraging R-Tree~\cite{Wu2012TKDE,Tao2014TKDE} and quard-tree~\cite{Zhang2013EDBT,Zhang2016TKDE} have been proposed to boost query efficiency. Notably, Wu et al.~\cite{Tao2014TKDE} designed a spatial inverted index for multidimensional data and developed a \emph{k}SK query model. In 2016, Zhang et al.~\cite{Zhang2016TKDE} proposed ILQ, a signature quadtree in memory for the same query type. To date, \emph{k}SK query efficiency has been greatly improved. Recently, in 2024, Song et al.~\cite{Song2024IoTsJ} presented the first model based on EASPE for direct \emph{k}SK querying over encrypted geo-textual data.


\subsection{Summary of related works}
\label{RISK:RelatedWorks:Summary}
Existing solutions for secure spatial-keyword queries suffer from the following limitations.
\begin{itemize}
  \item \emph{Tailored indices}. All the aforementioned indices are tailored exclusively for a single type of secure spatial-keyword query. Such indices are impractical for practical applications as they require deploying multiple indices, and incur excessive storage and maintenance overheads.

  \item \emph{Property-preserving encryption}. Most existing solutions except RASK rely on PPEs, incompatible with commercial products. First, commercial entities are only allowed to employ international or national standard encryptions, excluding the aforementioned PPEs. Second, PPEs lack long-term security verification. For instance, ASPE was first proposed by Wong et al.~\cite{Wong2009SIGMOD} in 2009 yet was broken by Li et al.~\cite{Li2019ICDE} merely 10 years later. Its short lifecycle renders it unsuitable for practical deployment.
\end{itemize}

Overall, limited progress has been made toward enabling rich spatial-keyword queries on encrypted geo-textual data.

\section{System model and preliminaries}
\label{RISK:Preliminaries}
In this section, we first present the system model, problem formalization and security model.  Next, we briefly introduce the quadtree that we adopt for index construction.

\subsection{System model and problem formalization}
\label{RISK:Preliminaries:SystemModel}
There are three entities in RISK designed for answering secure rich spatial-keyword queries on a geo-textual dataset $D=\{o_i\}_{i=1}^{n}$ containing $m$ distinct keywords. Each object $o \in D$ is modeled as a tuple $o=(p,\psi)$, where $p=(x,y)$ denotes the geographic coordinate of $o$ in a two-dimensional Euclidean space and $\psi=\{w_i\}_{i=1}^{m}$ represents the set of keywords describing the object $o$. The system architecture is illustrated in Fig.~\ref{RISK:fig:Architecture}.

\begin{figure}[t!]
	\centerline{\includegraphics[width=85mm]{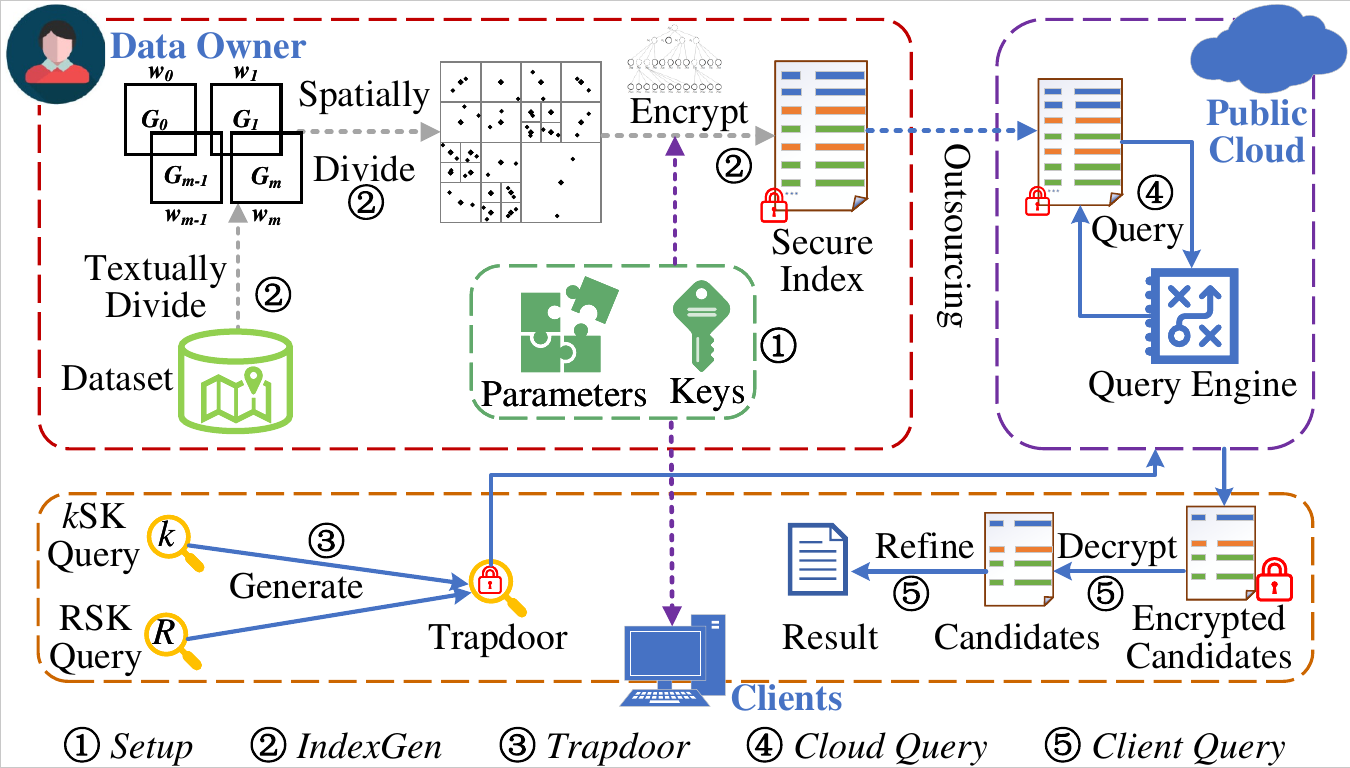}}
	\caption{The architecture of RISK. Here, the dotted arrows indicate off-line processes, and the continuous arrows are on-line processes which directly affects query time and is therefore critical to user experience. Thus, RISK mainly aims at reducing on-line time.}
	\label{RISK:fig:Architecture}
\end{figure}

\begin{itemize}
  \item A \emph{data owner} first invokes a \emph{Setup} primitive to initialize the whole system. Then it invokes an \emph{IndexGen} primitive to build a secure index for the dataset $D$. Subsequently, the secure index is further outsourced to a \emph{public cloud}.
  \item The \emph{public cloud} maintains the secure index. Upon receiving the trapdoor corresponding to either a RSK query $q_r$ or a \emph{k}SK query $q_k$ from a \emph{client}, the \emph{public cloud} invokes a \emph{Cloud Query} primitive to access the secure index and retrieve an encrypted candidate set $\hat{C}$, which is then returned to the \emph{client}.
  \item The \emph{client} invokes a \emph{Trapdoor} primitive to generate query trapdoors and submits them to the \emph{public cloud}. Subsequently, the \emph{client} executes a \emph{Client Query} primitive to decrypt the returned encrypted candidate set $\hat{C}$ and refine it to derive final query results.
\end{itemize}

The problems of secure RSK and \emph{k}SK queries are defined as follows, both constrained to the Euclidean-distance space.

\begin{definition}[Secure range spatial-keyword query problem] Given a secure index built on a geo-textual dataset $D$ and an encrypted trapdoor $T$ of a RSK query $q_r=(p,\psi,r)$, where $r$ is the range radius, the goal is to retrieve a result set $R=\{o | o \in D\}$ satisfying $\{dist(o.{p},q_r.{p})\leq r\} \land \{q_r.{\psi}\subseteq o.{\psi}\}$, where $dist(\cdot,\cdot)$ denotes the Euclidean distance.
  \label{RISK:Def:SRSKQ}
\end{definition}

\begin{definition}[Secure k-nearest neighbor spatial-keyword query problem]  Given a secure index built on a geo-textual dataset $D$ and an encrypted trapdoor $T$
of a \emph{k}SK query $q_k=(p,\psi,k)$. The problem retrieves a result set $R=\{o_i | (1 \leq i \leq k) \land (o_i \in D)\}$ such that $\forall o'\in D \setminus R,\forall o\in R:$ $dist(o'.p,q_k.p) \geq dist(o.p,q_k.p) \land q_k.{\psi}\subseteq o.{\psi}$, where $dist(\cdot,\cdot)$ denotes the Euclidean distance.
  \label{RISK:Def:SkSKQ}
\end{definition}

\subsection{Threat and security models}
\label{RISK:Preliminaries:ThreatModel}
In RISK, the data owner and the clients are honest, while the public cloud is honest-but-curious, which is widely accepted in both database community~\cite{Chen2025SIGMOD,Ahmed2025VLDB} and security community~\cite{perez2025USENIX,Sen2025CCS,Yao2025TIFS}. Furthermore, such model has also been deployed in industrial products, such as Goole's Mayfly~\cite{Bian2024mayfly}, Alibaba's SecretFlow~\cite{Fang2024VLDB}, etc. In such a threat model, the public cloud will follow the specified protocol exactly but may run any polynomial-time analysis to learn extra information about the data and query.

RISK adheres to IND-CKA2~\cite{Curtmola2006CCS} security model, which denotes indistinguishability under adaptive chosen keyword attack and serves as the security target for SOTA SSE solutions~\cite{Jiang2022TIFS,Tong2023TKDE}. This model assumes an adversary powerful enough to capture all encrypted indices, complete access patterns, and even continuously issue adaptive queries. Building on this, the security proof aims at demonstrating that such an adversary cannot distinguish between the real game $\mathsf{Game}_{\mathcal{R}}$ and simulated game $\mathsf{Game}_{\mathcal{S}}$ with a negligible probability by observing their respective outputs. Should this proof hold, the RISK scheme achieves security against such adversaries.

\begin{definition}[IND-CKA2 security] For any SSE instance, and any probabilistic polynomial-time (PPT) algorithm adversary $\mathcal{A}$ bounded by the security parameter $\lambda$, if there exists an efficient simulator such that $\mathcal{A}$ can distinguish between $\mathsf{Game}_{\mathcal{S}}$ and $\mathsf{Game}_{\mathcal{R}}$ only with  a negligible probability, then the SSE scheme is IND-CKA2 secure. That is to say,
$$|\mathsf{Pr}[\mathsf{Game}_{\mathcal{R}}=1]-\mathsf{Pr}[\mathsf{Game}_{\mathcal{S}}=1]|\leq \mathsf{negl}(\lambda).$$
Herein $\mathsf{negl}(\lambda)$ is a $\lambda$-bounded negligible function, where $\lambda$ is a security parameter.
  \label{RISK:Def:Security}
  \end{definition}

  \subsection{Preliminaries}
\label{RISK:Preliminaries:Preliminaries}
The \emph{quadtree} (Q-tree), a hierarchical spatial structure where each internal node has exactly four children, was proposed by Finkel et al.~\cite{Finkel1974Acta}. It is most widely adopted for partitioning two-dimensional space via recursive subdivision into four congruent cells~\cite{Samet1984CSUR,Li2024TKDE}. All nodes reside at distinct levels. Leaf nodes store actual data objects, whereas other nodes maintain the tree structure without holding any actual data objects. We thus refer to the former as real nodes and the latter as virtual nodes as illustrated in Fig.~\ref{RISK:fig:HybridIndex:kQ-tree}. A Q-tree satisfies the following features.

\begin{enumerate}
	\item Each cell maintains an object set $\Delta$ that comprises all objects falling in that cell.
	\item All cells share an identical maximum capacity. When this capacity threshold is reached, the cell is further subdivided into four new sub-cells.
	\item Each real node $N$ and its corresponding cell share a unique path $N.path$ which is generated in a top-down manner from the root to this node. Specifically, for each division, all cells are assigned path $0$ to $3$ left-to-right, bottom-to-top. As illustrated in Fig.~\ref{RISK:fig:HybridIndex:kQ-tree}, the path of node $N_{300}$ is ``$300$'', while that of node $N_{21}$ is ``$21$''. 
\end{enumerate}

In practical scenarios, inappropriate encryption schemes applied to the Q-tree will not only incur privacy leakage risks but also fail to support the efficient execution of \emph{k}NN queries due to their restriction on in-place computations and comparisons over ciphertext. To resolve these issues, we propose a well-designed secure hybrid index structure  that ensures privacy preservation while enabling both \emph{k}NN queries and range queries in subsequent sections.

\section{Secure hybrid index}
\label{RISK:Index}
As discussed in Section~\ref{RISK:RelatedWorks:Summary}, existing secure indices are tailored for specific types of spatial-keyword queries, thus rendering them incompatible with other query types. To address this limitation, a secure hybrid index is carefully designed by leveraging standard cryptographic primitives to simultaneously support both RSK and \emph{k}SK queries.

\subsection{Construction of kNN quadtree}
\label{RISK:Index:kQ-tree}
Both spatial-first-then-textual and textual-first-then-spatial paradigms are widely adopted for efficient geo-textual data indexing~\cite{Chen2020,Chen2013VLDB}. The former filters spatially first then textually, performs keyword queries within a fixed spatial range, and natively supports secure RSK queries~\cite{Lv2023}. However, encryption invalidates direct ciphertext comparisons, prevents dynamic spatial range expansion, and thus creates an inherent conflict with the core operations of $k$-nearest neighbor spatial keyword (\emph{k}SK) queries. As a result, this paradigm is unsuitable for \emph{k}SK queries. Conversely, the latter index type is better suited for \emph{k}SK queries, as it first rapidly filters out textually unmatched objects in the dataset and then performs spatial range retrieval. The remaining challenge is to embed neighborhood semantics into the textual-first-then-spatial index to support both RSK and \emph{k}SK queries. To address this challenge, we propose a novel \emph{k}NN quadtree (\emph{k}Q-tree) that effectively incorporates both range-related and neighborhood semantics.

\begin{figure}
	\centering
	\subfigure[Textual division and index.]{\includegraphics[width=85mm]{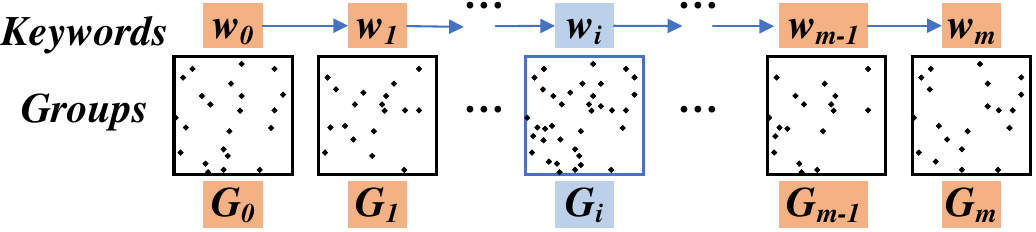}\label{RISK:fig:HybridIndex:Text}}\\
	\subfigure[\emph{k}Q-tree for group $G_i$.]{\includegraphics[width=85mm]{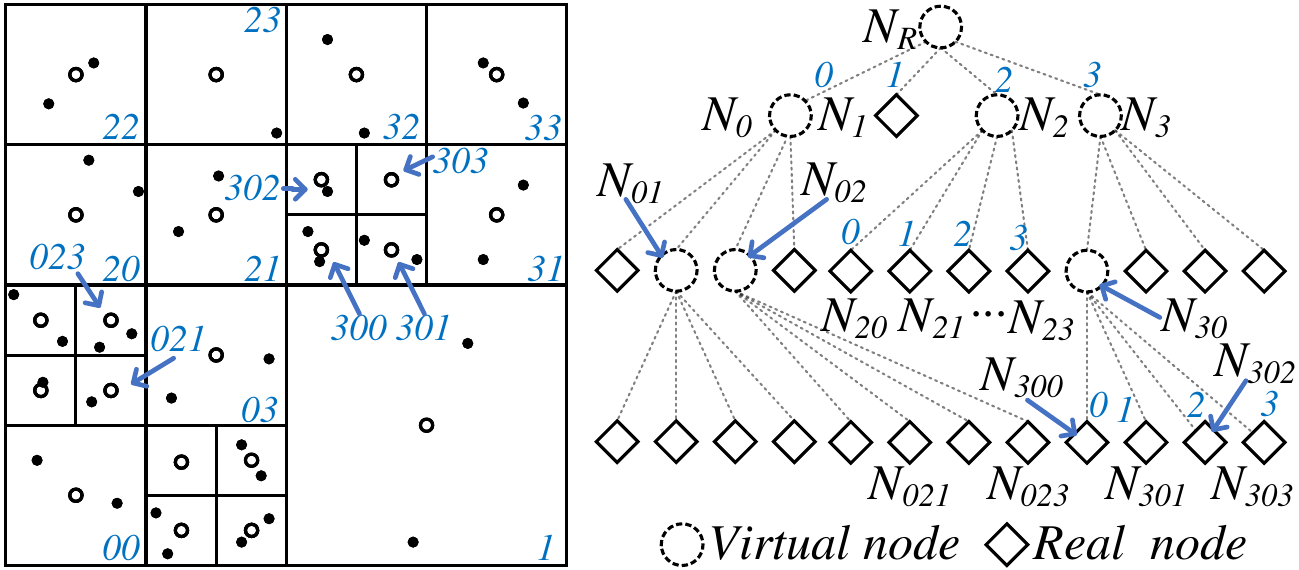}\label{RISK:fig:HybridIndex:kQ-tree}}
	\caption{Generic indexing concept ($k_{\max}=2$ and $d=3$). Virtual nodes (i.e., inner nodes), which only store structural semantics, are denoted by circles; real nodes (i.e., leaf nodes), which store objects, are represented by diamonds. All light blue numbers in the cells are unique identifiers, identical to the path of their corresponding nodes.}
	\label{RISK:fig:HybridIndex}
\end{figure}

Specifically, as illustrated in Fig.~\ref{RISK:fig:HybridIndex}, the entire geo-textual dataset is partitioned textually into $m$ groups, each corresponding to a unique keyword, where $m$ is the number of distinct keywords in $D$. For each group, we construct a \emph{k}Q-tree to accelerate spatial-keyword queries. To embed the neighborhood semantics of \emph{k}SK queries into the quadtree and convert the structure into a key-value tuple list, the \emph{k}Q-tree further incorporates three specialized features while retaining the inherent properties of a standard Q-tree.

\begin{enumerate}
	\addtocounter{enumi}{+3}
	\item All cells have a maximum capacity of $k_{\max}$, where $k_{\max}$ exceeds the number of returned NNs.
	
	\item For each real node $N$, a representative object $\dot{o}$ is selected for its corresponding cell. NN set $\Delta^k$ comprising $k_{\max}$ NNs to $\dot{o}$ in the same group is also inserted into real nodes.
	
	\item The whole tree for each keyword is converted into a key-value tuple list where the key $e$ is ``$keyword||N.path$'' and the value is $v=(e, \Delta, \Delta^k)$. As illustrated in Fig.~\ref{RISK:fig:HybridIndex:kQ-tree}, the key for node $N_{300}$ is ``$w_i||300$'' and the key for node $N_{21}$ is ``$w_i||21$'' in group $G_i$. 
\end{enumerate}

Given a \emph{k}SK query $q_k$, the cloud locates the cell containing $q_k.p$. Feature 4 mitigates, to a certain extent, the scenario where a cell cannot respond to $q_k$ due to an insufficient number of objects (such as cells $303$, $201$, etc. in Fig.~\ref{RISK:fig:HybridIndex:kQ-tree}). However, the true \emph{k}NNs may still fail to be retrieved. For example, if the query's NNs lie outside the located cell and are distributed across adjacent cells, the cell cannot provide valid candidates (such as a query locates in cell $1$'s left-top corner in Fig.~\ref{RISK:fig:HybridIndex:kQ-tree}). Feature 5 addresses this issue by ensuring each cell always maintains a sufficient number of candidates to cover the neighborhood range required by the query. Moreover, feature 6 converts the \emph{k}Q-tree into a key-value list for facilitating subsequent encryption.

\begin{algorithm}[!t]
	\footnotesize
	\caption{Construction of \emph{k}Q-tree.}
	\label{RISK:alg:kQ-tree}
	\begin{algorithmic}[1]
		\Require A group $G_i$ with a keyword $w$ and a maximum capacity $k_{\max}$.
		\Ensure A \emph{k}Q-tree $\mathbb{T}$ for group $G_i$ and the tree height $d_i$.
		\State Build a Q-tree $\mathbb{T}_i$ on $G_i$ with $k_{\max}$ and $d_i$ is the height of $\mathbb{T}_i$;
		\State $\mathbb{T} \leftarrow \emptyset$;
		\State $\mathbb{N} \leftarrow \cup \{\mathbb{T}_i.root.child_j\}_{j=1}^{3}$; \Comment{Initialize a node candidate queue.}
		\For{each $N \in \mathbb{N}$} \Comment{Process every node until the queue is empty.}
		\State Pop $N$ from $\mathbb{N}$;
		\If{$N$ is a virtual node}
		\State $\mathbb{N} \leftarrow \mathbb{N} \cup \{N.child_j\}_{j=1}^{3}$; \Comment{Insert child nodes into the queue.}
		\Else
		\State $e \leftarrow ``w||N.path"$; \Comment{Set key for node $N$.}
		\State Calculate $\Delta$ consisting all objects that fall in the node;
		\State Calculate the representative object $\dot{o}$ for $N$;
		\State Calculate NN set $\Delta^k$ to $\dot{o}$ in the same group;
		\State $v \leftarrow (e, \Delta, \Delta^k)$; \Comment{Set value for node $N$.}
		\State $\mathbb{T} \leftarrow \mathbb{T} \cup (e,v)$; \Comment{Insert key-value pair into the \emph{k}Q-tree.}
		\EndIf
		\EndFor
		\State \Return $\mathbb{T}_i=\mathbb{T}$ as the \emph{k}Q-tree and $d_i$ as the tree height;
	\end{algorithmic}
\end{algorithm}

\emph{Construction details.} For each group, the client executes Algorithm~\ref{RISK:alg:kQ-tree} to generate a \emph{k}Q-tree indexing objects. Totally, $m$ \emph{k}Q-trees will be built. First, a standard Q-tree is constructed (Line 1) following~\cite{Finkel1974Acta,Samet1984CSUR,Li2024TKDE}, and the tree height is derived simultaneously. The Q-tree is then traversed top-down to embed semantics and support both RSK and \emph{k}SK queries (Lines 4-14). Specifically, for each virtual node, only its child nodes are inserted into the candidate queue $\mathbb{N}$ (Line 7). For each real node, a key $e$  is formed by concatenating $w$ and the node's path (Line 9), and the NNs to its representative object $\dot{o}$ are computed in $G_i$ (Lines 11-12).  Deriving NNs is not the focus of this paper, and relevant details are referred to in~\cite{Hjaltason1995ASD}. The final value $v$ is the concatenation of $e$, $\Delta$ and $\Delta^k$ (Line 13), and the key-value pair $(e,v)$ is added to the \emph{k}Q-tree $\mathbb{T}$ (Line 14). After processing all real nodes, the complete \emph{k}Q-tree $\mathbb{T}_i$ and the height $d_i$ are obtained (Line 15). 

\subsection{Construction of secure kNN quadtree}
\label{RISK:Index:SkQ-tree}
To safeguard index privacy and confidentiality, the \emph{k}Q-tree requires encryption prior to cloud outsourcing. In contrast to existing works in literatures~\cite{Cui2019ICDE,Wang2021TIFS,Yang2022ICDCS,Song2024IoTsJ}, this paper leverages standardized cryptographic primitives to obfuscate the complete hierarchical structure of the \emph{k}Q-tree, conceal its associated geo-textual objects, and finally construct a secure \emph{k}NN quadtree (S\emph{k}Q-tree) by integrating $m$ \emph{k}Q-trees (one per distinct keyword). This design precludes the cloud from inferring the index topology or the underlying data distribution.

The client executes Algorithm~\ref{RISK:alg:SkQ-tree} to generate a S\emph{k}Q-tree. For a geo-textual dataset $D$, $m$ keyword-based groups are formed to construct $m$ \emph{k}Q-trees (Lines 1-2), which are then integrated into a single tree $\hat{\mathbb{T}}$ (Line 3). To enforce uniform value lengths, the maximum length $len$ is computed by scanning all values in $\mathbb{T}$ (Line 4), which is critical for ciphertext indistinguishability, consistent with prior works~\cite{Wang2021TIFS,Lv2023}. Next, through a loop (Lines 5-9), all keys are hashed, values are padded to the uniform length, and encrypted with a symmetric encryption scheme. Finally, $\hat{\mathbb{T}}$ is returned as the S\emph{k}Q-tree, with the maximum tree height $d$ returned simultaneously.

\begin{algorithm}[!t]
	\footnotesize
	\caption{Construction of S\emph{k}Q-tree.}
	\label{RISK:alg:SkQ-tree}
	\begin{algorithmic}[1]
		\Require A dataset $D$, a capacity $k_{\max}$, and secret keys $sk_H, sk_E$. 
		\Ensure A S\emph{k}Q-tree $\hat{\mathbb{T}}$ as the secure index and a maximum tree height $d$.
		\State Reorganize $m$ groups $\{G_i\}_{i=1}^{m}$ such that $\forall o \in G_i, w_i \in o.\psi$.
		\State Derive \emph{k}Q-trees $\{\mathbb{T}_i\}_{i=1}^{m}$ and tree heights $\{d_i\}_{i=1}^{m}$ by Algorithm~\ref{RISK:alg:kQ-tree} with $k_{\max}$;
		\State $\hat{\mathbb{T}} \leftarrow \cup{\{\mathbb{T}_i\}_{i=1}^{m}}$; \Comment{Integrate all \emph{k}Q-trees as a single tree.}
		\State Calculate maximum length $len$ for all values $v \in \hat{\mathbb{T}}$ by linear scanning;
		\For{each $(e,v)\in \hat{\mathbb{T}}$} \Comment{Process every key-value pair in $\hat{\mathbb{T}}$.}
		\State $e \leftarrow H_{sk_H}(e)$; \Comment{Hash the key $e$.}
		\State $v$ is padded to length $len$; \Comment{pad $v$ to a fixed length $len$.}
		\State $v \leftarrow Enc_{sk_E}(v)$; \Comment{Encrypt the value $v$.}
		\EndFor
		\State $d \leftarrow \max_{1 \leq i \leq m} \{d_i\}$; \Comment{Calculate the maximum tree height.}
		\State \Return $\hat{\mathbb{T}}$ as the secure index and $d$ as the maximum tree height;
	\end{algorithmic}
\end{algorithm}

The employed keyed hash function $H$ and symmetric encryption scheme $Enc$ collectively conceal the \emph{k}Q-tree structure and all  its values. Thus, by deploying a $SkQ$-tree, geo-textual data can be securely outsourced to the public cloud for rich spatial-keyword queries. In fact, to efficient process a S\emph{k}Q-tree, a perfect hash function~\cite{Fox1992SIGIR,Pibiri2021SIGIR} can be introduced on the public cloud to accelerate hashed key matching.

\section{Construction of RISK}
\label{RISK:Cons}
In RISK, both secure RSK and \emph{k}SK queries are supported over a S\emph{k}Q-tree. All primitives for RISK construction, as discussed in Section~\ref{RISK:Preliminaries:SystemModel}, are detailed below.

\subsection{Setup}
\label{RISK:Cons:Setup}
This primitive initializes prerequisite cryptographic tools and is executed by the data owner. The RISK system is parameterized by a security parameter $\lambda$ quantifying security strength. In practice, $\lambda$ usually denotes the length of a symmetric encryption key or a keyed hash function output. A larger $\lambda$  implies stronger security, selected per application-specific requirements. The National Institute of Standards and Technology (NIST)~\cite{Barker2020NIST} has explicitly deprecated $112$-bit symmetric security since 2020, setting the minimum accepted $\lambda$ to $128$ bits.

Given a security parameter $\lambda$, a keyed hash function $H_{sk_H}(\cdot) :\{0,1\}^{\lambda} \times \{0,1\}^{*} \rightarrow \{0,1\}^{\lambda}$ and a symmetric encryption $Enc_{sk_E}( \cdot):\{0,1\}^{\lambda} \times \{0,1\}^{*} \rightarrow \{0,1\}^{*}$ are appropriately selected. The  associated symmetric decryption is defined as $Dec_{sk_E}(\cdot):\{0,1\}^{\lambda} \times \{0,1\}^{*} \rightarrow \{0,1\}^{*}$ such that $Dec_{sk_E}(Enc_{sk_E}(m))=m$. Both secret keys $sk_H$ and $sk_E$ are also randomly selected and bounded by $\lambda$.

For the entire system, these secret keys are $sk_H$ and $sk_E$.

\subsection{IndexGen}
\label{RISK:Cons:IndexGen}
The S\emph{k}Q-tree serves as RISK's secure index, generated by the data owner with this primitive. For a geo-textual dataset $D$ and predefined $k_{\max}$, the secure index $\hat{\mathbb{T}}$ is constructed by executing Algorithms~\ref{RISK:alg:kQ-tree} and ~\ref{RISK:alg:SkQ-tree} sequentially, and then outsourced to the public cloud for client query processing. 

Concurrently, the algorithm derives the maximum tree height $d$ and calculates the width $W$ of dataset $D$, both of which are incorporated into the system parameters to underpin trapdoor generation. Here,
$$W=\max_{o \in D} \{o.p.x,o.p.y\}-\min_{o \in D} \{o.p.x,o.p.y\}.$$
Thus, system parameters are $SP=(d, W, H, Enc, Dec)$, publicly available to all participants.

\subsection{Trapdoor}
\label{RISK:Cons:Trapdoor}
For different secure spatial-keyword query types, trapdoor primitive adopts tailored strategies to generate trapdoors for the public cloud's candidate location filtering. Such primitives are executed by the client. Trapdoors for secure RSK and \emph{k}SK spatial-keyword queries are presented sequentially below.

\textbf{\emph{RTrapdoor}} is the primitive for trapdoor generation in secure RSK query. Since the client does not store the partition details of the S\emph{k}Q-tree, it cannot directly locate candidate tuples. Instead, it infers the coverage of a given query range based on the maximum tree height $d$.

To generate identifiers for the potentially covered query region, a false positive-tolerant approach is proposed. Since these false positive identifiers map to non-existent virtual nodes in the S\emph{k}Q-tree, the public cloud can efficiently filter them out. As illustrated in Fig.~\ref{RISK:fig:TrapdoorGen}, the detailed procedure illustrated in Algorithm~\ref{RISK:alg:CoverIDs} is executed by the client initiating a secure RSK query.

\begin{figure}[t!]
	\centerline{\includegraphics[width=85mm]{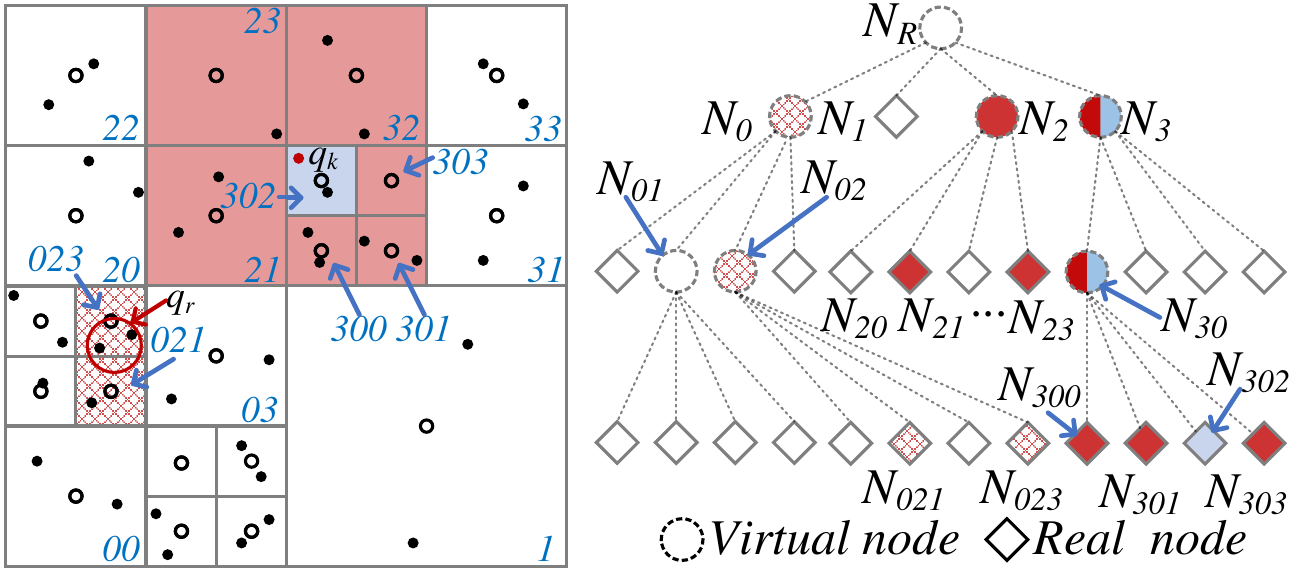}}
	\caption{An example for generating trapdoors for RISK. Here, in the left region, the red circle denotes a RSK query $q_r$ and the object $q_k$ colored by red dot is the center for both NSK and \emph{k}SK queries. Here, $k=2$ is for \emph{k}SK query. For both regions, the patterned, light-colored, and dark-colored cells indicate touched cells and nodes for RSK, NSK, and \emph{k}SK queries, respectively.}
	\label{RISK:fig:TrapdoorGen}
\end{figure}

Initially, the bottom-left and top-right coordinates of the circular query range's circumscribed square are computed to facilitate subsequent trapdoor generation (Line 1). The covered identifier set $\mathbb{ID}$ is initialized as empty, and the minimum inter-cell interval $w$ is derived (Line 2). A nested loop is then executed to compute identifiers of cells intersecting the range query by traversing the locating point from bottom-left to top-right in steps of $w$ (Lines 3-10). Leveraging  the uniform partitioning of the \emph{k}Q-tree, the  identifiers of all intersecting cells are directly computable (Lines 7-9). However, unaware of the tree structure and unable to distinguish real from virtual nodes in the path from the root to leaf, the client cannot determine the exact cell containing the locating point. Thus, all node identifiers along this path are inserted into $\mathbb{ID}$ (Lines 6, 10). Finally, $\mathbb{ID}$ is returned (Line 11).

\begin{algorithm}[!t]
	\footnotesize
	\caption{Generation of potential covered identifiers.}
	\label{RISK:alg:CoverIDs}
	\begin{algorithmic}[1]
		\Require A center $p$, a radius $r$, and the public parameter $SP=(d, W)$. 
		\Ensure A set $\mathbb{ID}$ of covered identifiers.
		\State $x_{\min} \leftarrow p.x-r$, $x_{\max} \leftarrow p.x+r$, $y_{\min} \leftarrow p.y-r$ and $y_{\max} \leftarrow p.y+r$;
		\State $w \leftarrow W \cdot 2^{-d}$, $\mathbb{ID} \leftarrow \emptyset$; \Comment{Initialize the minimum interval and $\mathbb{ID}$.}
		\For{$x=x_{\min}$ to $x_{\max}$ with step $w$} \Comment{Traverse along with X-axis.}
		\For{$y=y_{\min}$ to $y_{\max}$ with step $w$} \Comment{Traverse along with Y-axis.}
		\State $ID \leftarrow \emptyset$; \Comment{Initialize a covered identifier.}
		\For{$i=1$ to $d$ with step $1$} \Comment{Recursively compute $ID$.}
		\State $code \leftarrow 2\cdot \lfloor (\lfloor y/w \rfloor \mod 2^{d-i+1}) \cdot 2^{-(d-i)}\rfloor $;
		\State $code \leftarrow code + \lfloor (\lfloor x/w \rfloor \mod 2^{d-i+1}) \cdot 2^{-(d-i)}\rfloor$
		\State $ID \leftarrow ID || code$;
		\State If $\{ID\} \not\subseteq \mathbb{ID}$ then $\mathbb{ID} \leftarrow \mathbb{ID} \cup \{ID\}$;
		\EndFor
		\EndFor
		\EndFor
		\State \Return $\mathbb{ID}$;
	\end{algorithmic}
\end{algorithm}



The trapdoor $T_R$ for $q_r$ is generated as follows. 
\begin{equation}
	T_R=\{h_i=H_{sk_H}(w_i||ID_j) | w_i\in q_r.\psi, ID_j \in \mathbb{ID}\}.
	\label{RISK:Equa:RTrapdoor}
\end{equation}
Here $\mathbb{ID}$ is derived by Algorithm~\ref{RISK:alg:CoverIDs} with $p=q_r.p$ and $r=q_r.r$.

\emph{Example}. For the RSK query $q_r$ marked by the red circle in Fig.~\ref{RISK:fig:TrapdoorGen}, the two red-patterned cells $021$ and $023$ both intersect with $q_r$. However, since the client has no knowledge of the tree structure, it cannot determine the hierarchical levels of these two cells. Fortunately, the identifiers of all virtual and real nodes covered by the query range can be deduced by the client from the tree height. Thus, cells $0$ and $02$ are also considered as potential intersected cells. The final set of covered identifiers is $\mathbb{ID}=\{0, 02, 021, 023\}$. Furthermore, if the query's description satisfies $\{w_i\} \in q_r.\psi$, the trapdoor for $q_r$ is formulated as follows.
\begin{align*}
	T_R=\{& H_{sk_H}(w_i||0), H_{sk_H}(w_i||02), \\ &H_{sk_H}(w_i||021), H_{sk_H}(w_i||023)\}.
\end{align*}

\textbf{\emph{NTrapdoor}} is the primitive for trapdoor generation in secure nearest neighbor spatial-keyword (NSK) queries. By integrating the NN set $\Delta^{k}$ into the value field $v$ of each \emph{k}Q-tree tuple (Algorithm~\ref{RISK:alg:kQ-tree}, Line 13), we fully exploit secure RSK queries and reduce the NN query $q_k=(p, \psi, 1)$ to a secure RSK query with $r=0$. The trapdoor $T_N$ is calculated as follows. 
\begin{equation}
	T_N=\{H_{sk_H}(w_i||ID_j) | w_i\in q_k.\psi, ID_j \in \mathbb{ID}\}.
	\label{RISK:Equa:NTrapdoor}
\end{equation}
Where $\mathbb{ID}$ is derived by Algorithm~\ref{RISK:alg:CoverIDs}, with $p=q_k.p$ and $r=0$.

\emph{Example}. For the NSK query $q_k$ marked by the red dot in Fig.~\ref{RISK:fig:TrapdoorGen}, the light-colored cell $302$ is exactly the target cell containing this query. Following the same rationale for RSK queries, two additional cells $3$ and $30$ are also considered as intersected cells. The final set of covered identifiers is $\mathbb{ID}=\{3, 30, 302\}$. Furthermore, if the query keyword satisfies $\{w_i\} \in q_k.\psi$, the trapdoor for $q_k$ is constructed as follows.
\begin{equation*}
	T_N=\{ H_{sk_H}(w_i||3), H_{sk_H}(w_i||30), H_{sk_H}(w_i||302)\}.
\end{equation*}

\textbf{\emph{kTrapdoor}} is the primitive for generating trapdoors tailored to secure \emph{k}SK queries. During the construction of a S\emph{k}Q-tree, the NN set $\Delta^k$ is integrated into value $v$ (Algorithm~\ref{RISK:alg:kQ-tree}, Lines 10–11), whereby the results of \emph{k}NN queries are highly likely to fall within this set. However, the distance between $q_k.p$ and the representative object $\dot{o}$ introduces a non-negligible probability that the $k$ objects in $\Delta^k$ are not the accurate \emph{k}NN matches. To improve query accuracy, a two-phase strategy is proposed to achieve secure \emph{k}SK queries. 

In the first phase, a secure RSK query is issued with $r=0$ to spatially locate the cell containing the query. Thus, 
\begin{equation}
	T_k^{(1)}=T_N.
	\label{RISK:Equa:kTrapdoor1}
\end{equation}

Then, $T_k^{(1)}$ is sent to the public cloud, which returns an encrypted candidate set $\hat{C}$. The plaintext candidate set is derived as $C=\{(e, v)=Dec_{sk_E}(\hat{c}) | \hat{c}\in\hat{C}\}$. Next, $v$ is parsed into $e$, $\Delta$ and $\Delta^{k}$, with $e$ further decomposed into $w$ and $path$. The length of $path$ indicates the located level of the cell that $q_k.p$ falls in. To enhance the query's accuracy, all cells around the located cell are fetched back for more candidates. 

To this end, we generate an appropriate radius $r^k$ to cover all directly adjacent cells. For any key-value pair $(e_i, v_i) \in C$, the query radius for the second phase is as follows.
\begin{equation}
	r_i^k=\sqrt{2} \cdot W \cdot 2^{-|v_i.e.path|-1} + 0.1 \cdot W \cdot 2^{-d}.
	\label{RISK:Equa:rk}
\end{equation}
Here, the first component constructs the circumscribed circle enclosing the target cell, and the second slightly extends the radius to encompass all immediate neighboring cells.The touched cells and nodes are shaded in dark color in Fig.~\ref{RISK:fig:TrapdoorGen}.

The trapdoor $T_k^{(2)}$ for the second phase is formulated as,
\begin{equation}
	T_k^{(2)}=\{H_{sk_H}(w_i||ID_j) | w_i\in q_k.\psi, ID_j \in \mathbb{ID}_i\}.
	\label{RISK:Equa:kTrapdoor2}
\end{equation}
Here $\mathbb{ID}_i$ is derived by Algorithm~\ref{RISK:alg:CoverIDs}, with $p=q_k.p$ and $r=r_i^k$.


Thus, the final trapdoor for a secure \emph{k}SK query $q_k$ is
\begin{equation}
	T_k=(T_k^{(1)}, T_k^{(2)}).
	\label{RISK:Equa:kTrapdoor}
\end{equation}

\emph{Example}. For the aforementioned \emph{k}SK query $q_k$ (denoted by the red dot in Fig.~\ref{RISK:fig:TrapdoorGen}, where $k=2$), the light-colored cell $302$ is retrieved in first phase. The final set of covered identifiers for $T_k^{(1)}$ is $\{3, 30, 302\}$. In the second phase, the dark-colored cells are located, the final set of covered identifiers is $\mathbb{ID}=\{300, 301, 303, 32, 320, 321, 2, 23, 231, 21, 211, 213\}$. Moreover, for any query keyword $\{w_i\} \in q_r.\psi$, the trapdoors for $q_k$ are as follows.
\begin{align*}
	T_k^{(1)}=\{ &  H_{sk_H}(w_i||3), H_{sk_H}(w_i||30), H_{sk_H}(w_i||302)\}.\\	
	T_k^{(2)}=\{
	&H_{sk_H}(w_i||300), H_{sk_H}(w_i||301), H_{sk_H}(w_i||303), \\
	&H_{sk_H}(w_i||32), H_{sk_H}(w_i||320), H_{sk_H}(w_i||321),\\	
	&H_{sk_H}(w_i||2), H_{sk_H}(w_i||23), H_{sk_H}(w_i||231),\\	
	&H_{sk_H}(w_i||21), H_{sk_H}(w_i||211), H_{sk_H}(w_i||213)\}.
\end{align*}

In both secure NSK and \emph{k}SK queries, the intersection results across distinct keywords may fail to satisfy the required cardinality of $1$ (for NSK queries)  or $k$ (for \emph{k}SK queries) objects, respectively. This limitation exposes the two-phase strategy to the risk of overlooking valid matches, thereby necessitating supplementary query phases for mitigation. 

\textbf{Additional \emph{Trapdoor$^+$}} is the primitive for trapdoor generation in supplementary phases, addressing the aforementioned issues. Its generated trapdoor is similar to $T_R$ primitive for secure RSK queries. The radiuses for secure NSK and \emph{k}SK queries are initialized to be $r=0$ and $r=r_i^k$ respectively. The radius is then stepwise increased by $W \cdot 2^{-d}$ until sufficient results are derived. Hence, the supplementary radiuses for both queries are as follows.
\begin{equation}
	\begin{split}
		\left\{
		\begin{aligned}
		r=&0 .\\
		r_i^k=&\sqrt{2} \cdot W \cdot 2^{-|v_i.e.path|-1} + (\theta + 0.1) \cdot W \cdot 2^{-d}.
		\end{aligned} 
		\right.
	\end{split}
	\label{RISK:Equa:r+}
\end{equation}

For secure NSK and \emph{k}SK queries, the additional trapdoors remain unchanged when Equations~\ref{RISK:Equa:NTrapdoor} and~\ref{RISK:Equa:kTrapdoor2} are computed with the above radiuses. Here, positive integer $\theta$ counts the number of additional phases performed, increasing by $1$ in each round. These phases proceed iteratively until enough results are generated, guaranteeing fully accurate final outcomes.

\subsection{Query}
\label{RISK:Cons:Query}
While the generated trapdoors for secure RSK, NSK, and \emph{k}SK queries are distinct, the initial cloud query procedure remains uniform. Core differences exist in the client query processes, presented in sequence herein.

\textbf{\emph{CQuery}} is the primitive for a public cloud generating encrypted candidates for a received trapdoor $T$, which is one of $T_R$, $T_N$, $T_k^{(1)}$, and $T_k^{(2)}$. The cloud generates encrypted candidates as follows.
\begin{equation}
	\hat{C}=\{(\hat{e}, \hat{v}) | \hat{e}=T_i \land (\hat{e}, \hat{v}) \in \hat{\mathbb{T}} \land T_i \in T\}.
	\label{RISK:Equa:CCandidate}
\end{equation}

When the client receives the encrypted candidate $\hat{C}$, it can first obtain the plaintext candidate $C$ as follows:
\begin{equation}
	C=\{(e, v)=Dec_{sk_E}(\hat{c}) | \hat{c}\in\hat{C}\}.
	\label{RISK:Equa:Candidate}
\end{equation}

Subsequently, distinct primitives are invoked for different query types.

\textbf{\emph{RQuery}} is the primitive for a client to derive results for a secure RSK query $q_r$, with its workflow delineated in Algorithm~\ref{RISK:alg:RQuery}.  Initially, the result set $R$ is initialized as empty (Line 1). A loop is then executed to traverse all elements in the candidate set $C$ (Lines 2-3). Specifically, each element $v$ is parsed, and only the component $v.\Delta$ is leveraged for result refinement. If an object $o$ satisfies the condition specified in Line 3, it is added to $R$. Finally, $R$ is returned to the client.  

\begin{algorithm}[!t]
	\footnotesize
	\caption{Client-side query for a secure RSK query.}
	\label{RISK:alg:RQuery}
	\begin{algorithmic}[1]
		\Require A query $q_r$, an encrypted candidate $C$. 
		\Ensure A result set $R$.
		\State $R \leftarrow \emptyset$;
		\For{each $o \in v.\Delta$ and $(e, v) \in C$} \Comment{Traverse all objects in $C$.}
		\State If $o \not\in R \land dist(o, q_r.p) \leq q_r.r \land q_r.\psi \subseteq o.\psi$ then insert $o$ into $R$;
		\EndFor
		\State \Return $R$;
	\end{algorithmic}
\end{algorithm}

\textbf{\emph{NQuery}} is the primitive for a client to derive result for a secure NSK query $q_k$ in which $k=1$. The query result should be one of the objects in $v.\Delta^k$ or $v.\Delta$. Hence, the result $R$ is
\begin{equation}
	R=\{o | dist(o, q_k.p) \leq dist(o', q_k.p)\}.
	\label{RISK:Equa:NQuery}
\end{equation}
Where, $o, o' \in (v.\Delta^k \cup v.\Delta)$, $o' \neq o$, and $q_k.\psi \subseteq (o.\psi \cap o'.\psi)$.

\textbf{\emph{kQuery}} is the primitive for a client to derive result for a secure \emph{k}SK query $q_k$ with $k>1$. Upon completing two-phase queries on the public cloud, the client executes Algorithm~\ref{RISK:alg:kQuery} to derive the \emph{k}SK query result. Specifically, the result is derived from all NN sets. First, objects in these sets are aggregated with duplicates removed (Lines 1-2). Then, a linear scan is performed to generate the exact query result $R$ (Lines 4-7).

\begin{algorithm}[!t]
	\footnotesize
	\caption{Client-side query for a secure \emph{k}SK query.}
	\label{RISK:alg:kQuery}
	\begin{algorithmic}[1]
		\Require A secure \emph{k}SK query $q_k$, encrypted candidates $C^{(1)}$ and $C^{(2)}$ for trapdoor $T_k^{(1)}$ and $T_k^{(2)}$,  and a secret key  $sk_E$. 
		\Ensure A result set $R$.
		\State $V \leftarrow \{o | (o \in v.\Delta \cup  v.\Delta^k) \land ((\hat{e}, v)\in C^{(i)}) \land (i\in\{1,2\})\}$;
		\State Repeat objects in $V$ are removed;
		\State $R \leftarrow \emptyset$, $\delta  \leftarrow 0$; \Comment{Initialize result set.}
		\For{each $o \in V$} \Comment{Traverse all objects in $V$.}
		\If{$(|R|<k) \lor (q_k.\psi \subseteq o.\psi \land dist(q_k.p, o.p)<\delta)$} 
		\State $o$ is inserted into $R$; \Comment{$o$ is a valid result.}
		\State $\delta \leftarrow \max\{\delta,dist(q_k.p, o.p)\}$; \Comment{Reset the maximum distance.}
		\EndIf
		\EndFor
		\State \Return $R$;
	\end{algorithmic}
\end{algorithm}

\section{Theoretical analysis}
\label{RISK:Analyses}
RISK is analyzed in this section from aspects of security and complexity, followed by theoretical comparisons.

\subsection{Security analysis}
\label{RISK:Analyses:Security}
The security of RISK is proved by the following theorem.

\begin{theorem}[IND-CKA2 security for RISK]
	RISK is IND-CKA2 $(\mathcal{L}_1,\mathcal{L}_2)$-secure in the random oracle model. For any probabilistic polynomial-time adversary $\mathcal{A}$ attempting to break RISK, its advantage is bounded by
	\begin{equation*}
	  \begin{split}
		&|\mathsf{Pr}[\mathsf{Game}_{\mathcal{R}}=1]-\mathsf{Pr}[\mathsf{Game}_{\mathcal{S}}=1]|\\
		\leq & \mathsf{negl}(\lambda)=\epsilon (\sum_{i=1}^{m}2^{d_i-1}+ s \cdot d) \frac{poly(\lambda)}{2^{\lambda}}.
	  \end{split}
	\end{equation*}
	Herein, $\mathsf{Game}_{\mathcal{R}}$ and $\mathsf{Game}_{\mathcal{S}}$ denote the real and simulated games, respectively, and $\epsilon$ is an adjustable polynomial expansion factor for range queries. The leakage functions are $\mathcal{L}_1=\emptyset$ and $\mathcal{L}_2=\{SP(Q), Hist(Q)\}$. Both the real game $\mathsf{Game}_{\mathcal{R}}$ and the simulated game $\mathsf{Game}_{\mathcal{S}}$ are defined as follows,
	\begin{itemize}
		\item $\mathsf{Game}_{\mathcal{R}}$. In the real game, $\mathcal{A}$ acquires the outputs of the real primitives except query. Then, $\mathcal{A}$ adaptively performs query and obtains the real transcripts generated by these primitives. Finally, $\mathcal{A}$ observes the real transcripts and outputs a bit $b\in\{0,1\}$.

		\item $\mathsf{Game}_{\mathcal{S}}$. In the simulated game, all hash functions are replaced by random oracles, and symmetric encryption scheme generates random values as outputs. $\mathcal{A}$ obtains the outputs of the simulated primitives. Then, $\mathcal{A}$ adaptively performs simulated queries and gets the simulated transcripts generated by these primitives. Finally, $\mathcal{A}$ observes the simulated transcripts and outputs a bit $b\in\{0,1\}$.
	\end{itemize}
  	\label{RISK:Theorem:Security}
\end{theorem}

\begin{proof}
	In RISK, there are three indispensable primitives (i.e., \textbf{RTrapdoor}, \textbf{CQuery}, and \textbf{RQuery}) for achieving secure queries. By composing these primitives, RISK can be regarded as a secure range spatial-keyword (RSK) query system, which also supports secure \emph{k}-nearest neighbor spatial-keyword (\emph{k}SK) queries. If these three primitives are proven IND-CKA2 secure, other primitives can be proved due to the combinable employments of the foregoing primitives.

    In the following, we first prove IND-CKA2 security for secure RSK queries based on the above three primitives. Then, we further promote the IND-CKA2 security for secure rich spatial-keyword queries in RISK.

    \emph{IND-CKA2 security of RSK queries falling in a single cell} is proved in a game hopping manner, in which a sequence of games approaching $\mathsf{Game}_{\mathcal{S}}$ from $\mathsf{Game}_{\mathcal{R}}$ is carefully crafted. The proof is carried out by attempting to distinguish the transcripts of two adjacent games in the sequence. If no polynomial-time adversary can distinguish two adjacent games, the security of $\mathsf{Game}_{\mathcal{R}}$ is equivalent to that of $\mathsf{Game}_{\mathcal{S}}$.

    Throughout the following analysis of the indistinguishability between two adjacent games, the adopted keyed hash function and symmetric encryption algorithm are replaced by random oracles, as the random oracle model is adopted in the proof. A random oracle maintains a triple list $\mathbb{O}=\{(in, sk, out)\}$ to answer a query $in$ with secret key $sk$ and output a randomly generated $out$. Three elements $(in, sk, out)$ are just peremptorily linked with no semantic correlation. Thus, the advantage for any probabilistic polynomial-time adversary in breaking a random oracle is $poly(\lambda)/(2^{\lambda})$ which is negligible~\cite{Canetti2004JACM}. 

	\emph{Game $\mathsf{G}_{0}$} is exactly the real game $\mathsf{Game}_{\mathcal{R}}$. They both output the same transcript, and hence it holds that
	\begin{equation*}
	  \begin{split}
		|\mathsf{Pr}[\mathsf{G}_{0}=1]-\mathsf{Pr}[\mathsf{Game}_{\mathcal{R}}=1]|=0.
	  \end{split}
	\end{equation*}

	\emph{Game $\mathsf{G}_{1}$} is identical to $\mathsf{G}_{0}$ except that $H$ and $Enc$ adopted in lines 6 and 8 of Algorithm~\ref{RISK:alg:SkQ-tree} (from the original paper) are replaced by two orthogonal random oracles $O_H$ and $O_E$. For distinct inputs $in_0$ and $in_1$, both random oracles output random values $out_0$ and $out_1$ respectively. Obviously, $\mathsf{Pr}[out_0=out_1] \leq poly(\lambda)/(2^{\lambda})$ holds by the security of random oracles. Since $H$ and $Enc$ are each called at most $\sum_{i=1}^{m}(2^{d_i})$ times, the advantage for any probabilistic polynomial-time adversary in distinguishing the games is,
	\begin{equation*}
	  \begin{split}
		|\mathsf{Pr}[\mathsf{G}_{1}=1]-\mathsf{Pr}[\mathsf{G}_{0}=1]|\leq \sum_{i=1}^{m}\frac{2^{d_i}\cdot poly(\lambda)}{2^{(\lambda -1 )}}.
	  \end{split}
	\end{equation*}

	\emph{Game $\mathsf{G}_{2}$} precisely follows $\mathsf{G}_{1}$ except that the keyed hash function $H$ adopted in Equation~(\ref{RISK:Equa:RTrapdoor}) (from the original paper) is replaced by random oracle $O_H$ defined in $\mathsf{G}_{1}$. Assume the given query falls in a single cell, there are $d$ elements in $\mathbb{ID}$.  Since the query involves $s$ distinct keywords, $\mathbb{ID}$ thus has a total of $s \cdot d$ elements. As $H$ is executed $s \cdot d$ times, the advantage for any probabilistic polynomial-time adversary in distinguishing the two games is,
	\begin{equation*}
	  \begin{split}
		|\mathsf{Pr}[\mathsf{G}_{2}=1]-\mathsf{Pr}[\mathsf{G}_{1}=1]|\leq \frac{s \cdot d \cdot poly(\lambda)}{2^{\lambda}}.
	  \end{split}
	\end{equation*}
	
	\emph{Game $\mathsf{G}_{2}$} is exactly the simulated game $\mathsf{Game}_{\mathcal{S}}$, since all keyed hash functions and symmetric encryption algorithms in $\mathsf{Game}_{\mathcal{R}}$ are fully replaced by random oracles. Hence, we get that $|\mathsf{Pr}[\mathsf{Game}_{\mathcal{S}}=1]-\mathsf{Pr}[\mathsf{G}_{2}=1]|=0$ holds.

	By generalization, there exists a probabilistic polynomial-time  adversary distinguishing $\mathsf{Game}_{\mathcal{R}}$ from $\mathsf{Game}_{\mathcal{S}}$ with a negligible probability. Hence, we have
	\begin{equation*}
	  \begin{split}
		&|\mathsf{Pr}[\mathsf{Game}_{\mathcal{R}}=1]-\mathsf{Pr}[\mathsf{Game}_{\mathcal{S}}=1]|\\
		\leq & \mathsf{negl}_0(\lambda)=(\sum_{i=1}^{m}2^{d_i-1}+ s \cdot d) \frac{poly(\lambda)}{2^{\lambda}}.
	  \end{split}
	\end{equation*}
	
	\emph{IND-CKA2 security for RSK queries falling in multiple cells}. If a range query falls in multiple cells, it can be viewed as the client packaging multiple queries falling in a single cell. Further, the repeated generating trapdoors will be removed to conserve both bandwidth and computation resources. Hence, by introducing an expansion factor $\epsilon_1$, we have
		\begin{equation*}
	  \begin{split}
		&|\mathsf{Pr}[\mathsf{Game}_{\mathcal{R}}=1]-\mathsf{Pr}[\mathsf{Game}_{\mathcal{S}}=1]|\\
		\leq & \mathsf{negl}_1(\lambda)= \epsilon_1 \cdot \mathsf{negl}_0(\lambda).
	  \end{split}
	\end{equation*}


	\emph{IND-CKA2 security for nearest neighbor spatial-keyword (NSK) queries}. The advantage of any probabilistic polynomial-time adversary in breaking NSK queries is equivalent to that in breaking range queries. The only difference between the two queries is that the radius for the former is fixed to $0$ while the radius for the latter is variable according to the given query. The advantages analyses are identical to those above, since NSK queries must fall in a single cell. Hence, we have
	\begin{equation*}
	  \begin{split}
		|\mathsf{Pr}[\mathsf{Game}_{\mathcal{R}}=1]-\mathsf{Pr}[\mathsf{Game}_{\mathcal{S}}=1]|\leq  \mathsf{negl}_2(\lambda)=\mathsf{negl}_0(\lambda).
	  \end{split}
	\end{equation*}

	\emph{IND-CKA2 security for k\emph{SK} queries}. In RISK, a \emph{k}SK query is accomplished by sequentially executing a NSK query and a RSK query. Where the expansion factor for the range query is $\epsilon_2$. Thus, we have 
	\begin{equation*}
		\begin{split}
		  &|\mathsf{Pr}[\mathsf{Game}_{\mathcal{R}}=1]-\mathsf{Pr}[\mathsf{Game}_{\mathcal{S}}=1]|\\
		  \leq & \mathsf{negl}_3(\lambda)= (\epsilon_2+1) \cdot \mathsf{negl}_0(\lambda).
		\end{split}
	  \end{equation*}

	  \emph{Analysis of leakage functions}. Leakages occur only at \textbf{Trapdoor} and \textbf{Query}, as analyzed below. For a keyed hash function (resp. a symmetric encryption), the probability of producing the same output for two distinct keys (resp. values) in a \emph{k}Q-tree is negligible. Thus, $\mathcal{L}_1=\emptyset$ holds. During \textbf{Trapdoor} and \textbf{Query}, a series $Q=\{q_i\}$ of queries are issued. The query pattern is defined as $SP(Q)={sp_{ij}}$ where $sp_{ij}=sp_{ji}=1$ if $q_i=q_j$ holds or $sp_{ij}=sp_{ji}=0$. It is apparent that $SP(Q)$ is inevitably leaked. $Hist(Q)$ maintains all pairs $<ID, ts>$ for all touched objects when answering $q_i \in Q$. Herein, $ts$ is a timestamp for the time when the object is touched. $Hist(Q)$ is also evidently leaked because the public cloud can observe it by monitoring the stored S\emph{k}Q-tree. The leakage function is $\mathcal{L}_2=\{SP(Q), Hist(Q)\}$.
\end{proof}

\emph{Resistance to Access Pattern Leakage Attacks}. Numerous attacks exploit accurate access patterns~\cite{Lacharite2018SP,Zhang2016USENIXSecurity,Li2023SCUR}, among which the file-injection attack is particularly notable. This attack assumes adversaries can inject forged files embedded with predefined queries and leverage observed access patterns to launch attacks. However, such attacks are inherently ineffective against RISK for two key reasons. First, file-injection attacks depend on the critical assumption that keyword spaces are limited and enumerable, a condition inapplicable to RISK as its query location space is continuous and non-enumerable. Second, RISK leaks imprecise access patterns, as the cloud server returns candidates with false positives. Adversaries thus obtain erroneous patterns when injecting forged files and cannot distinguish between valid and invalid ones, results in failed file-injection attacks. 

\subsection{Complexity analysis}
\label{RISK:Analyses:Complexity}
This section analyzes the complexity of RISK in trapdoor generation and query processing, in terms of query latency and transmission overhead. We consider the baseline scenario where a range query is confined to a single cell with $k_{\max}$ objects (the maximum per cell). For cross-cell queries, the corresponding complexity can be derived via multiplication by a scaling factor, which is omitted here for brevity.

\begin{table*}
	\centering
	\caption{Theoretical comparisons with existing constructions.}
	\label{RISK:table:TheoreticalComparison}
	\begin{threeparttable}
		\begin{tabular}{c|c|c|c|c|c}
			\hline
			Schemes  & Type & Security & Encryption tools\tnote{$\ast$} & Query time\tnote{$\dagger$} & Trapdoor overhead\tnote{$\dagger$}\\
			\hline
			ELCBFR+~\cite{Cui2019ICDE}  & Range & Broken & ASPE & $O(\beta_1 k\log_k (n)) t_{\mathrm{ASPE}}$ & $O(\beta_1 \ell_{B} \lambda)$\\
			PBRQ~\cite{Wang2020INFOCOM}  & Range &  IND-CPA &  HVE & $O(\log_4 (n)(2^{|log_2{4^{d}}|}+\ell_{B}))t_{\mathrm{HVE}}$ & $O(3\ell_{B}\lambda)$\\
			SKSE~\cite{Wang2021TIFS}  & Range &  IND-SCPA &  HVE & $O(\beta_2 k\log_k (n)) t_{\mathrm{HVE}}$ & $O(\beta_2 f \lambda)$\\
			SKQ~\cite{Yang2022ICDCS}  & Range &  IND-SCPA &  EASPE & $O(16n(36+m)) t_{\mathrm{IP}}$ & $O(8(36+m)^2)$\\
			LSKQ~\cite{Yang2022ICDCS}  & Range &  IND-SCPA &  EASPE & $O(2n(44+m))t_{\mathrm{IP}}$ & $(44+m)^2$\\
			RASK~\cite{Lv2023}  & Range & IND-CKA2 & KH \& SE & $O(s)(t_D+t_H) + O(\frac{n}{2^{kd}})t_C$ & $O(s\lambda)$\\
			PBKQ~\cite{Song2024IoTsJ}  & \emph{k}NN & IND-CKA2 & EASPE & $O(|DB(w)|+1) t_{IP} + O(M)t_H + O(t n) t_{E}$ & $O(s\lambda)$\\
			\hline
			\multirow{2}{*}{RISK}  & Range & IND-CKA2 & KH \& SE & $O(s)(t_H + t_L) + O(s)t_L+O(s(2k_{\max}+1))(t_D+t_C)$ & $O(s\lambda)$\\
			  & \emph{k}NN & IND-CKA2 & KH \& SE & $O(s(\eta+1))(t_H + t_L) + O((\eta+1) s(2k_{\max}+1))(t_D+t_C)$ & $O(s(\eta+2)\lambda)$\\
			\hline
		\end{tabular}
		\begin{tablenotes}
			\footnotesize
			\item[$\ast$] SE is symmetric encryption, KH is a keyed hash function, ASPE is an asymmetric scalar product-preserving encryption~\cite{Wong2009SIGMOD}, EASPE is an enhanced ASPE presented in~\cite{Yang2022ICDCS}, and HVE is hidden vector encryption~\cite{Lai2018CCS}.
			\item[$\dagger$] $\beta_1$ and $\beta_2$ are the number of bloom filters in~\cite{Cui2019ICDE} and~\cite{Wang2021TIFS} respectively. $\ell_{B}$ is the size of a bloom filter. $t_{\mathrm{ASPE}}$, $t_{\mathrm{HVE}}$, and $t_{\mathrm{IP}}$ are the running times for a single comparison on encrypted bloom filters in ASPE, HVE, and EASPE, respectively. $f$ is the size of the bitmap in PBRQ.
		\end{tablenotes}
	\end{threeparttable}
\end{table*}

\emph{Complexity of trapdoor generating}. First, for an RSK query, $s$ trapdoors are generated, one for each of its $s$ keywords. For each keyword, there is only a single cell covering the query range, the trapdoor generation time is $O(s) t_H$ and the transmission overhead is $O(s\lambda)$, where $\lambda$ is the size of hash value. Second, for a NSK query, only one cell is involved due to the corresponding node contains all potential NNs. Thus, its trapdoor generation time and transmission overhead match those of the RSK query. Finally, for a \emph{k}SK query, the combined deployment of \textbf{NTrapdoor} and \textbf{\emph{k}Trapdoor} yields a total trapdoor generation time of $O(s) t_H + O(\eta \cdot s) t_H = O(s(\eta+1)) t_H$. Here, $\eta$ is the average number of adjacent cells to the cell containing all query keywords. Transmission overhead comprises $O(s(\eta+1))$ for the aforementioned trapdoors and $O(s)$ for the target cell identifier, leading to a total overhead of $O(s(\eta+2)\lambda)$.

\emph{Query complexity analysis}. First, for a RSK query, the public cloud traverses all $s$ keywords and locates the corresponding cells in the \emph{k}Q-tree. Leveraging perfect hash during \emph{k}Q-tree construction, the time complexity of cell localization is $O(s)t_L$, where $t_L$ is the single localization latency. The located cells are then transmitted to the client for decryption and refinement. The transmission overhead for delivering encrypted candidate objects amounts to $O(s(2k_{\max}+1)\lambda)$, since each located node contains at most $k_{\max}$ objects and $k$ NNs to the representative object (cf. Lines 10-12 of Algorithm~\ref{RISK:alg:kQ-tree}). The time complexity for decryption and refinement at the client is $O(s(2k_{\max}+1))(t_D+t_C)$, where $t_D$ and $t_C$ denote the latency per decryption and comparison operation, respectively. The total time complexity thus equals $O(s)t_L+O(s(2k_{\max}+1))(t_D+t_C)$. Furthermore, the NSK query reduces to the RSK query with $r=0$, and its complexity analysis is consistent with that of the RSK query.

A \emph{k}SK query involves two phases on the public cloud. The first phase executes a NSK query, with time complexity $O(s)t_L+O(s(2k_{\max}+1))(t_D+t_C)$ and transmission overhead $O(s(2k_{\max}+1)\lambda)$. The second phase, which executes $\eta$ consecutive NSK queries, achieves time complexity $O(\eta s)t_L+O(\eta s(2k_{\max}+1))(t_D+t_C)$ and transmission overhead $O(\eta s(2k_{\max}+1)\lambda)$. Overall, the total time complexity of a \emph{k}SK query is $O((\eta+1) s)t_L+O((\eta+1) s(2k_{\max}+1))(t_D+t_C)$ with corresponding transmission overhead $O((\eta+1) s(2k_{\max}+1)\lambda)$.

\subsection{Theoretical comparisons}
\label{RISK:Analyses:TheoreticalComparisons}
This section compares RISK with SOTA schemes for secure spatial-keyword queries, as summarized in Table~\ref{RISK:table:TheoreticalComparison}. 


Notably, all models listed in Table~\ref{RISK:table:TheoreticalComparison} except ELCBFR+ are provably secure under equivalent security level (i.e., IND-SCPA and IND-CKA2). Most of the schemes, however, rely on PPE tools for direct querying over encrypted geo-textual data. Neither HVE nor EASPE has industrial recognition, with no viable industrial-grade encryption alternatives available. In contrast, RISK as well as RASK exclusively adopts industrially standardized cryptographic tools, rendering it better aligned with industrial deployment than SOTA schemes.

Both RASK and RISK outperform SOTA schemes by leveraging highly efficient cryptographic primitives (i.e., KH and SE), whose computational overheads are orders of magnitude lower than that of HVE and EASPE. For secure RSK queries, RISK achieves minimal trapdoor complexity, on par with PBKQ~\cite{Song2024IoTsJ} and RASK~\cite{Lv2023}, and substantially lower than that of SOTA schemes. For secure \emph{k}SK queries, RISK's trapdoor size increases to accommodate access to neighboring cells, and this expansion is bounded by the finite number of adjacent cells. In summary, RISK's efficiency and trapdoor complexity establish it as the top-tier schemes.

\section{Practical extensions}\label{RISK:Extensions}

In this section, we extend RISK to natively support both capabilities.

	

\subsection{Extensions to multi-party scenarios}
\label{RISK:Extensions:MultiUser}
In RISK, the employed cryptographic primitives (i.e., symmetric encryption and keyed hash function) natively preclude multi-user or multi-owner scenarios support. Directly extending RISK to such scenarios poses a critical challenge rooted in the key sharing problem. Fortunately, trusted execution environments (TEEs)~\cite{Wu2025TDSC,Lv2025TKDE,Wang2025TIFS} provide a robust security framework enforced by hardware for sensitive data and programs hosted on the cloud. Endowed with secure enclave and remote attestation capabilities, TEEs effectively resolve this issue.

To extend RISK for multi-user scenarios, secret keys $sk_H$, $sk_E$, and the system parameter $SP$ are securely embedded into the TEE of the public cloud.  Leveraging the TEE's secure enclave, external adversaries are denied access to these secret keys. When multiple clients initiate spatial-keyword queries, each negotiates a session key $sk_S$ with the TEE. The client encrypts the original query using $sk_S$ and transmits it to the TEE, which then recovers the query, acts as a proxy client, and further  execute RISK's concrete protocol to retrieve the query result. Finally, the TEE encrypts the result with $sk_S$ and returns it to the client. This mechanism enables secure spatial-keyword querying for multiple clients.

Multi-owner scenarios are supported similarly. Specifically, all data owners securely transmit their geo-textual objects to the TEE, which acts a proxy owner and constructs the secure index as specified in Section~\ref{RISK:Cons}. Furthermore, the secure index is stored outside the TEE. Authorized clients can then initiate spatial-keyword queries over this index.

\emph{Security analysis of multi-party extensions.} Each client can only access its own query and corresponding result, as the communication between the client and TEE is encrypted using an independent session key $sk_S$ and the former cannot touch $sk_H$ and $sk_E$. So, malicious clients are isolated from the security index, thereby preventing collusion among them. While the TEE can access all queries and results across multiple clients, its full trustworthiness ensures no leakage of sensitive information. Accordingly, the security guarantees of RISK in multi-user scenarios are equivalent to those in single-user scenarios. The same security assurance applies to multi-owner scenarios similarly.

\subsection{Extensions for data updates}
\label{RISK:Extensions:DataUpdates}
A further critical challenge for RISK in practical deployment is data updates. Since data owners often need to modify their outsourced objects, we extend RISK to support dynamic data updates based on the primitives introduced in Section~\ref{RISK:Cons}.

In RISK, small-scale data updates are implemented via interactions between the data owner and public cloud. The core challenges lie in locating all existing cells for which the newly inserted object $o^*$ serves as a nearest neighbor, and updating the secure index to support both secure RSK and \emph{k}SK queries. For each keyword $w^* \in o^*.\psi$, the data owner first issues a \emph{k}SK query with $q_k^*=(o^*.p, w^*, k_{\max})$. This query retrieves all objects $\tilde{\Delta}$ within and adjacent to the cells containing $o^*.p$. The retrieval radius is then expanded, and all objects $\tilde{\Delta}^*$ in the newly covered cells are derived using additional \emph{Trapdoor$^+$}. This expansion repeats until $o^*.p$ no longer qualifies as a nearest neighbor for any object in $\tilde{\Delta}^*$. If expansion proceeds, $\tilde{\Delta}$ is updated as $\tilde{\Delta}=\tilde{\Delta}\cup \tilde{\Delta}^*$. Finally, the NN set for each touched cell in the S\emph{k}Q-tree is recalculated, and $o^*$ is inserted. For the deletion of an existing object $o'$, the process is similar, the only distinction is that $o'$ is removed after updating NN sets for all touched cells.

Notably, the S\emph{k}Q-tree incurs no splits during updates, making it efficient only for small-scale data updates. Large-scale updates may induce data skew, and improving efficiency for such updates remains an interesting open problem. From a security perspective, the update operations only use the primitives in Section~\ref{RISK:Cons}, whose security is proved in Theorem~\ref{RISK:Theorem:Security}. Thus, the security of RISK remains unchanged under data updates, and its security guarantee holds.

\section{Experimental evaluations}
\label{RISK:Experi}
Both RISK and SOTA schemes are evaluated on the same datasets to demonstrate the superiority of RISK.

\subsection{Experimental environments and datasets}
\label{RISK:Experi:Environments}
Three real-world datasets are collected from social and map applications, differing in spatial distribution, keyword frequency, and other characteristics as summarized in Table~\ref{RISK:table:Datasets}. Specifically, \emph{Twitter} comprises geotagged tweets posted during the $2020$ U.S. election. \emph{NewYork} and \emph{Paris} contain locations with description in New York City and Paris City, respectively. All the three real-world datasets exhibit skewed keyword frequency distributions with large mean squared error (MSE). For a more comprehensive evaluation, we additionally introduced a non-skewed synthetic dataset denoted as \emph{Gaussian}, whose spatial coordinates conform to a Gaussian distribution. 

\begin{table}[!t]
	\centering
	\caption{Overview of the employed real datasets.}
	\label{RISK:table:Datasets}
	\begin{tabular}{crrrr}
		\hline
		\multirow{2}{*}{Dataset}  & \multirow{2}{*}{Total objects}   & \multicolumn{1}{c}{Total distinct}       & \multicolumn{2}{c}{Keyword}\\
		\cline{4-5}
		                          &                                  & \multicolumn{1}{c}{keywords}             & Freq.      & MSE \\
		\hline
		\emph{Twitter}~\cite{URLTwitter} & $800,437$       & $1,185,558$           & $13$            &    $671$  \\
		\emph{NewYork}~\cite{URLNewYork} & $4,181,477$       & $63,854$              & $622$        &   $31,759$   \\
		\emph{Paris}~\cite{URLParis}     & $5,711,822$       & $103,110$           & $739$        &  $45,969$ \\
		\emph{Gaussian}   & $10,000,000$ & $9,074,395$ & $13.77$ & $11$ \\
		\hline
	\end{tabular}
\end{table}

All experiments are conducted on a machine equipped with two Intel E5-2630v4 3.1 GHz CPUs and 384 GB of DDR4 RAM. A query set of $100$ objects is randomly sampled from each dataset. All performance comparisons between RISK and SOTA schemes are reported as the mean across $5$ independent runs per query instance. The representative object for each cell is its center object as widely adopted in clustering~\cite{Xie2022CSUR,Zhou2024CSUR}. 

\subsection{Parameter tuning}
\label{RISK:Experi:Parameters}
In RISK, $k_{\max}$ denotes the maximum number of objects per representative object in a cell. According to \emph{k}Q-tree's characteristics, $k_{\max}$ directly affects the tree depth, which in turn dictates trapdoor count and query response time. To optimize RISK's performance, $k_{\max}$ is tuned from $2$ to $100$ during dataset partitioning. For this process, the RSK query range is fixed at $1\%$  and the number of returned NNs is set to $k_{\max}$. Tuning results are summarized in Fig.~\ref{RISK:fig:Parameter}.

\begin{figure}
	\centering
	\includegraphics[width=85mm]{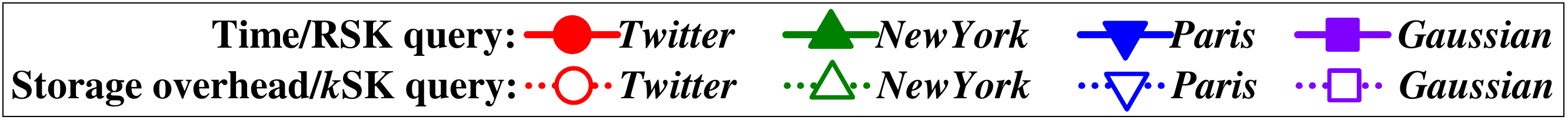}

	\subfigure[Index building performances.]{\includegraphics[width=42.5mm]{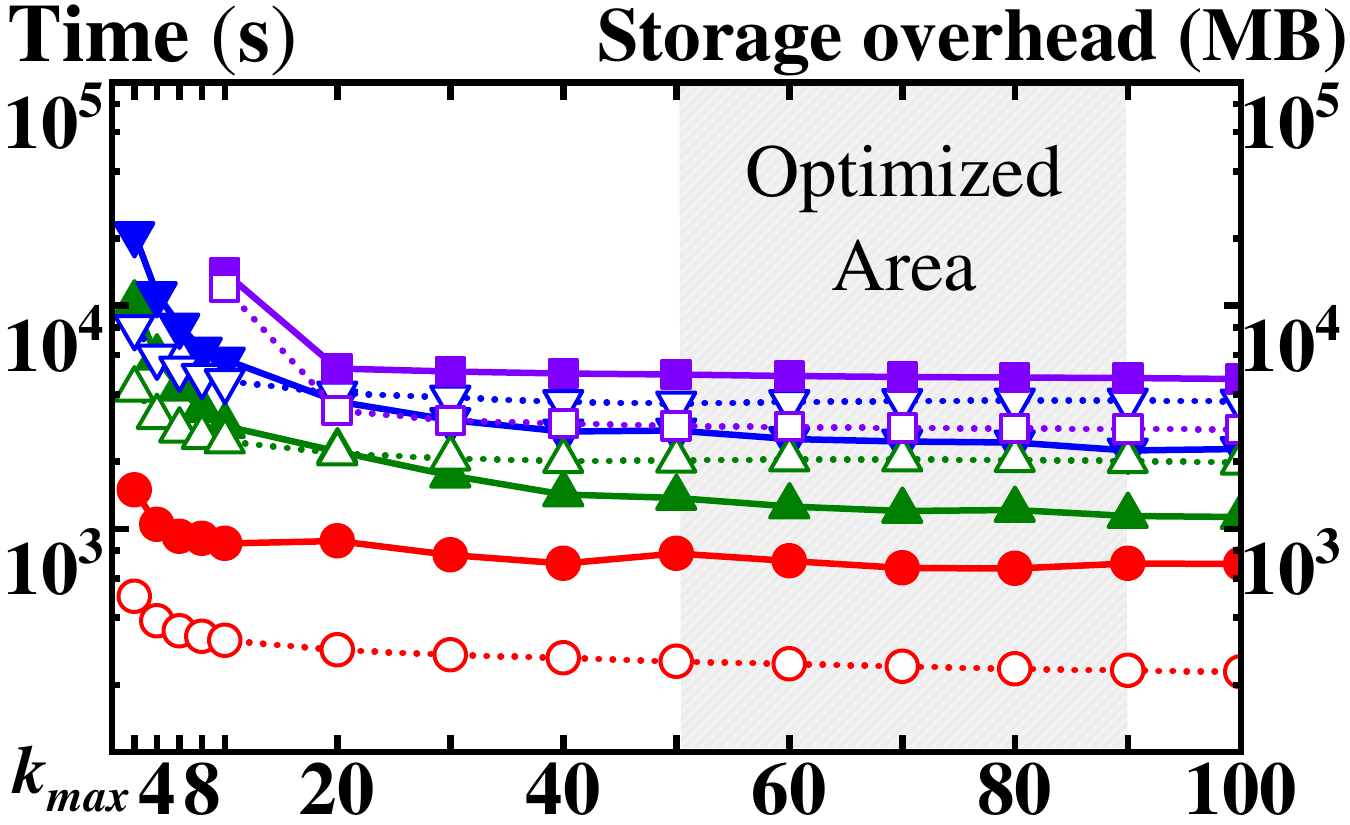}\label{RISK:fig:Parameter:IndexBuilding}}
	\subfigure[Communication performance.]{\includegraphics[width=42.5mm]{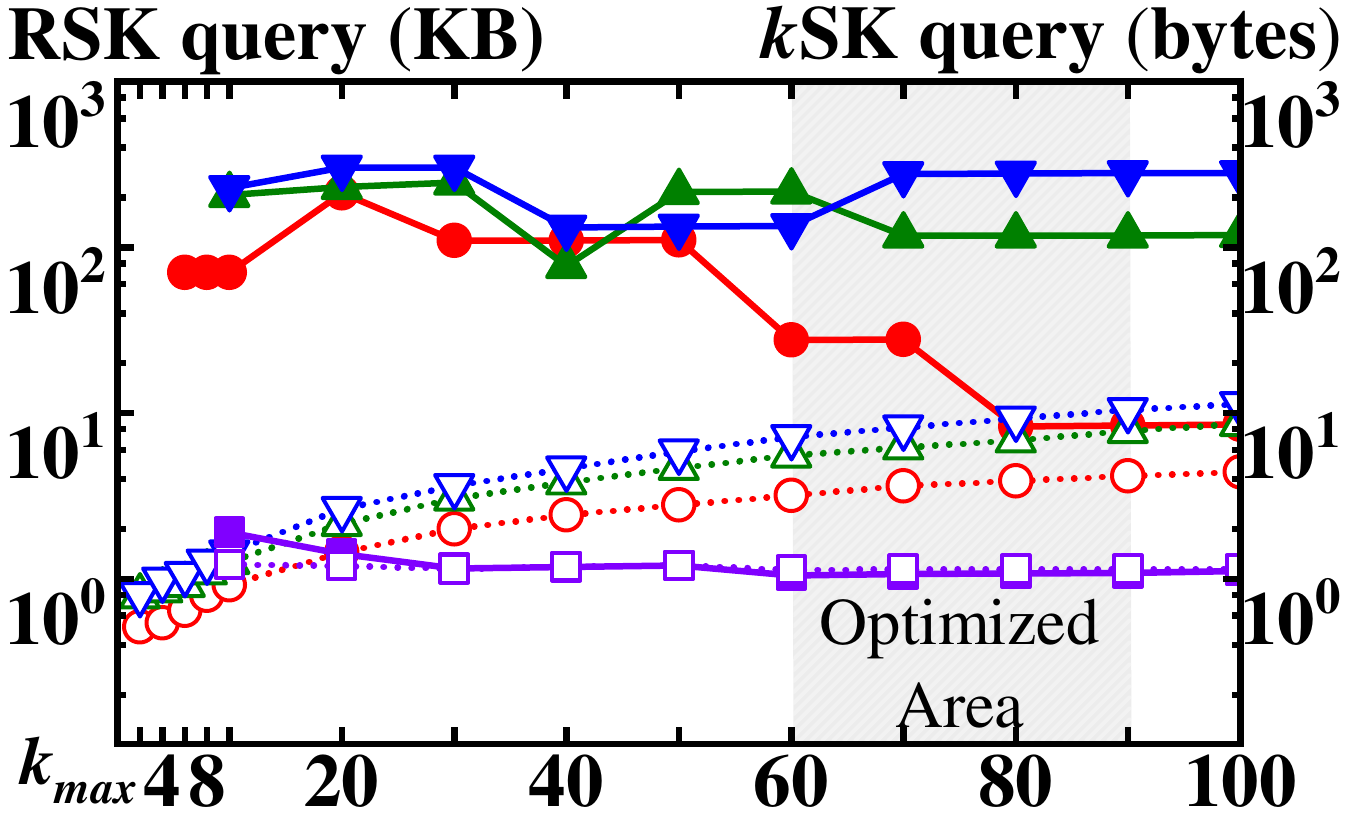}\label{RISK:fig:Parameter:Communication}}\\
	\subfigure[Client response time.]{\includegraphics[width=42.5mm]{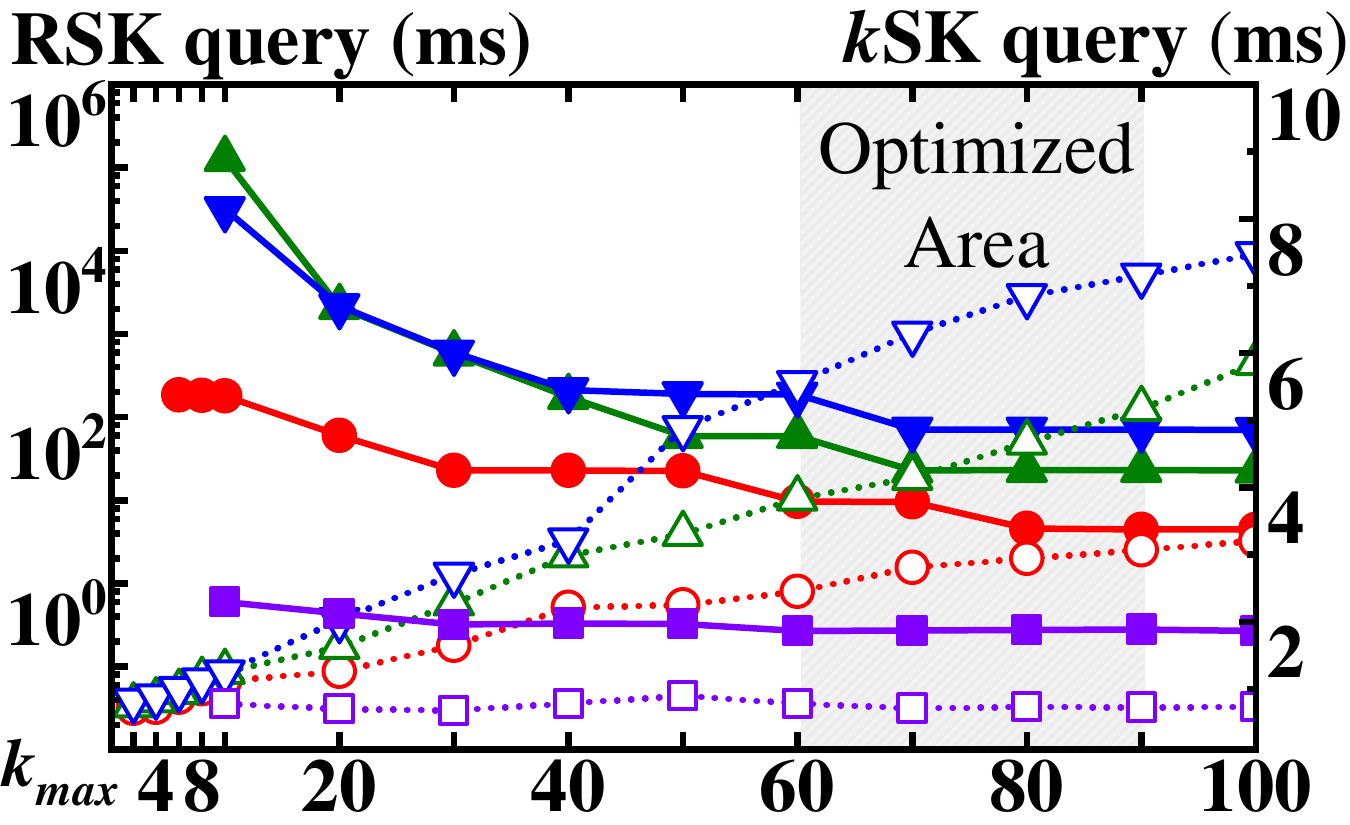}\label{RISK:fig:Parameter:ClientTime}}
	\subfigure[Cloud response time.]{\includegraphics[width=42.5mm]{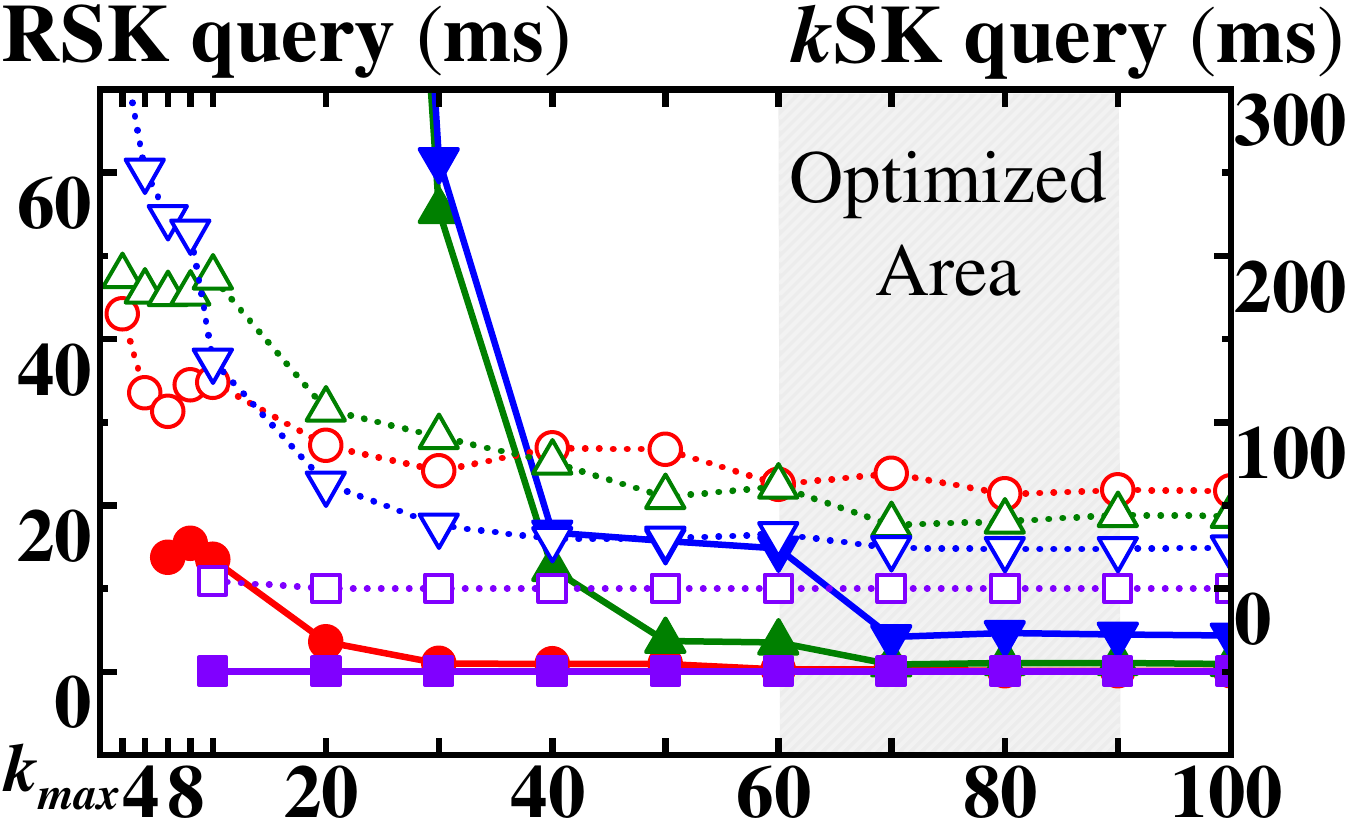}\label{RISK:fig:Parameter:CloudTime}}
	\caption{The parameter tuning results for RISK. }
	\label{RISK:fig:Parameter}
\end{figure}


As shown in Fig.~\ref{RISK:fig:Parameter:IndexBuilding}, as $k_{\max}$ increases, the number of touched nodes decreases, and both the index construction time and cloud storage cost decrease. Communications cost consists of trapdoor overhead and candidate overhead. The \emph{k}Q-tree height is inversely proportional to the number of generated trapdoors and directly proportional to the number of located candidates. As illustrated in Fig.~\ref{RISK:fig:Parameter:Communication}, for both RSK and \emph{k}SK queries, communication performance asymptotically stabilizes when $k_{\max}\geq 60$.

Client and cloud response times are illustrated in Fig.~\ref{RISK:fig:Parameter:ClientTime} and Fig.~\ref{RISK:fig:Parameter:CloudTime}, respectively. For RSK queries, increasing $k_{\max}$ significantly reduces response time, as the number of trapdoors required to cover the target area drops sharply. Thus, RSK query response times for both sides decrease dramatically with increasing $k_{\max}$. For \emph{k}SK queries, larger $k_{\max}$ values greatly reduce trapdoor matching time on the cloud, which dominates overall \emph{k}SK query response time, thus decreasing cloud response time. Moreover, increasing $k_{\max}$ sharply elevates the number of generated trapdoors, which in turn significantly increases client response time. This increment, however, is limited to several milliseconds. 

Based on the above performance observations, optimized areas are marked in Fig.~\ref{RISK:fig:Parameter} to highlight the optimum $k_{\max}$ values. For performance optimization, $k_{\max}$ is fixed at $80$ across all three datasets in the subsequent experiments. The query range and number of returned NNs are tuned within $[1\%,5\%]$ and $[2,10]$, respectively, with the number of keywords per query fixed at $2$. To further reduce cloud storage costs, the number of NNs per representative object is set to $10$.

\subsection{Experimental comparisons}
\label{RISK:Experi:Comparisons}
Herein, RISK is compared with SOTA schemes for RSK and \emph{k}SK queries on the \emph{Twitter} dataset. PBRQ~\cite{Wang2020INFOCOM}, SKQ~\cite{Yang2022ICDCS}, LSKQ~\cite{Yang2022ICDCS}, SKSE~\cite{Wang2021TIFS} and PBKQ~\cite{Song2024IoTsJ} are excluded from evaluations on the \emph{NewYork}, \emph{Paris}, and \emph{Gaussian} datasets, as all these schemes suffer from either excessive query latency ($>100$ s) or exorbitant storage overhead ($>500$ GB). Such limitations render them impractical for large-scale deployment, with their performance metrics on such datasets unavailable.

\begin{figure}
	\centering
	\includegraphics[width=85mm]{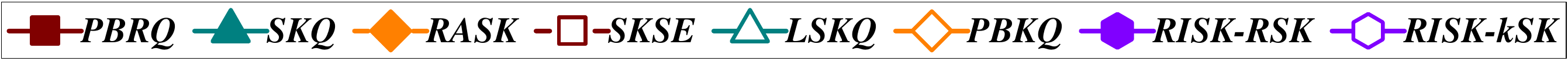}
	\subfigure[Cloud storage overhead.]{\includegraphics[width=42.5mm]{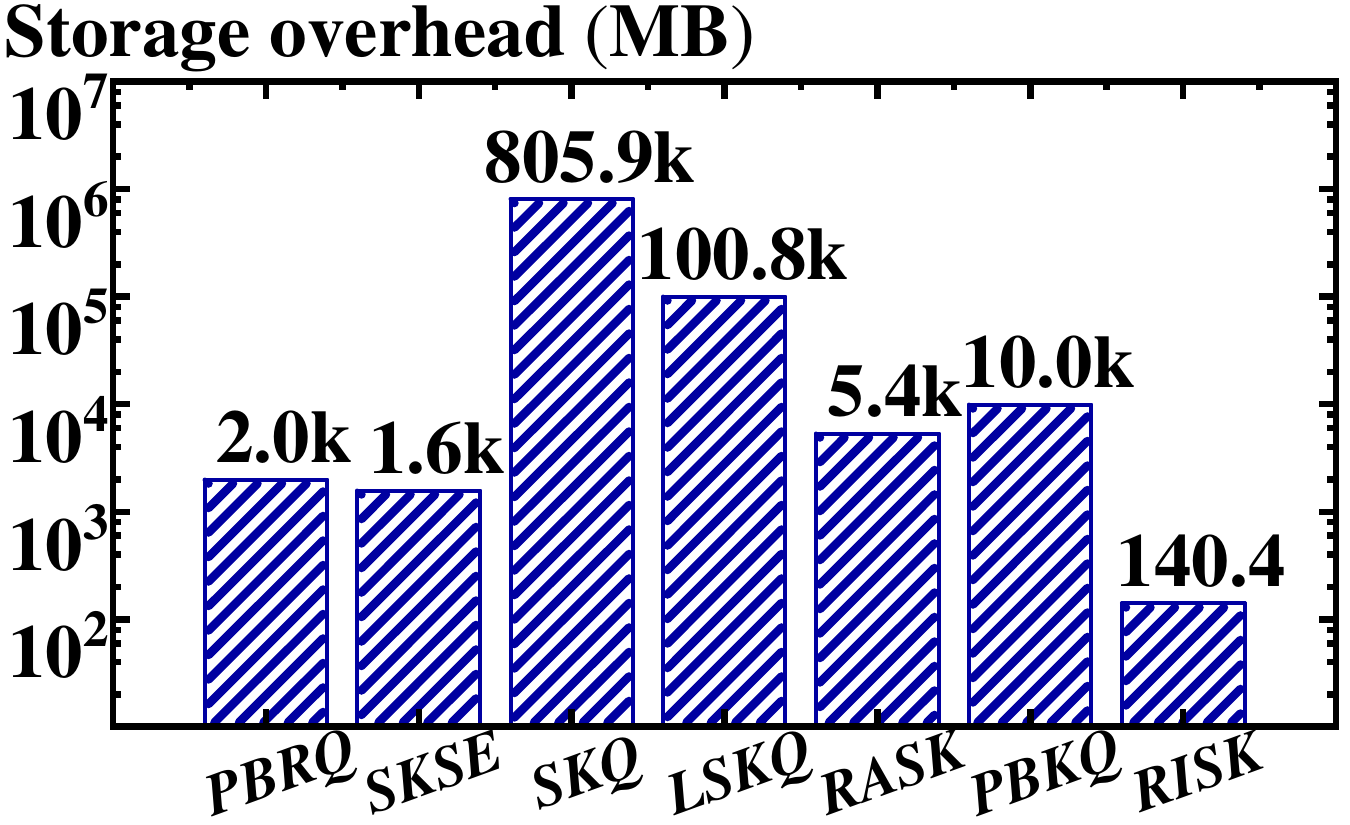}\label{RISK:fig:Comparisons:Storage:Twitter}}
	\subfigure[Total response time.]{\includegraphics[width=42.5mm]{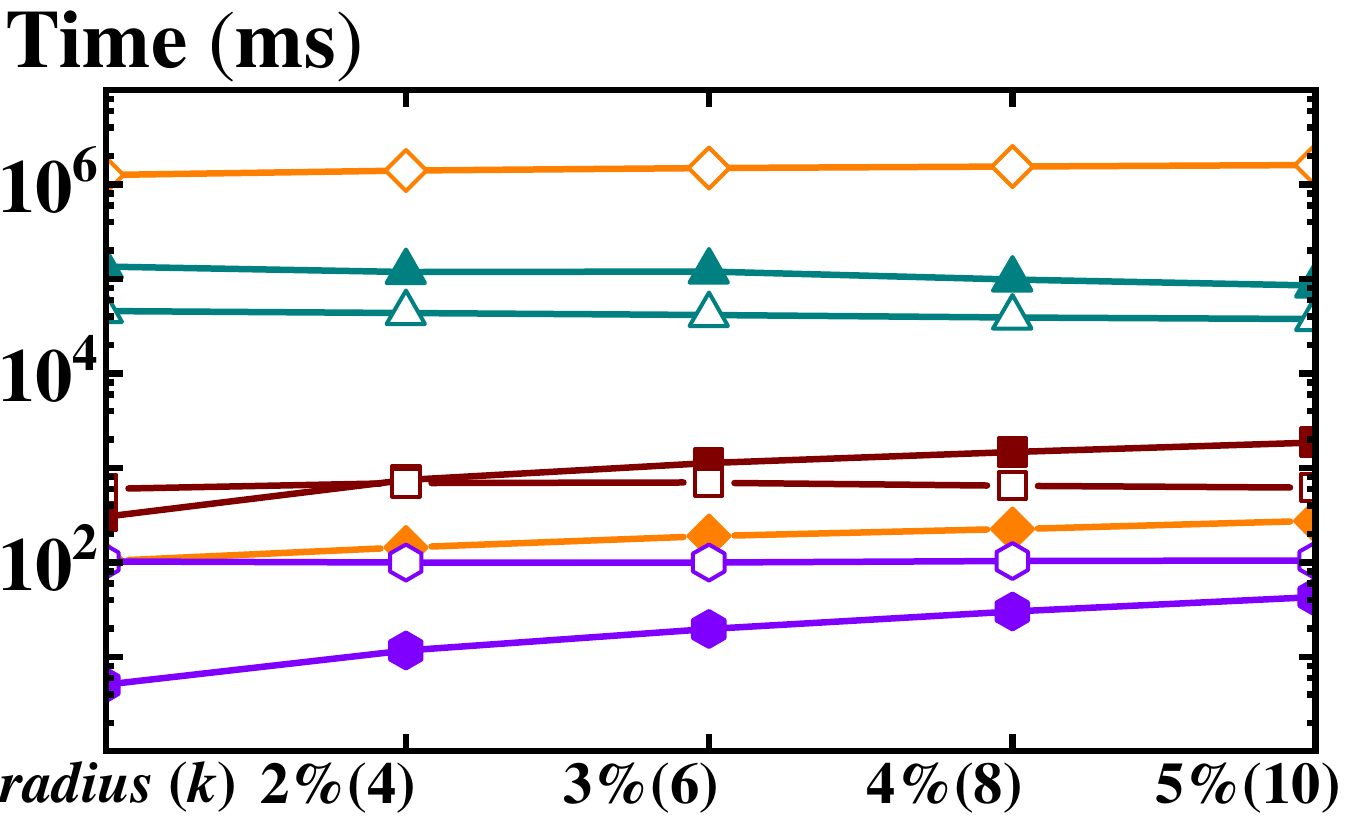}\label{RISK:fig:Comparisons:Response:Twitter}}\\
	\subfigure[Trapdoor overhead.]{\includegraphics[width=42.5mm]{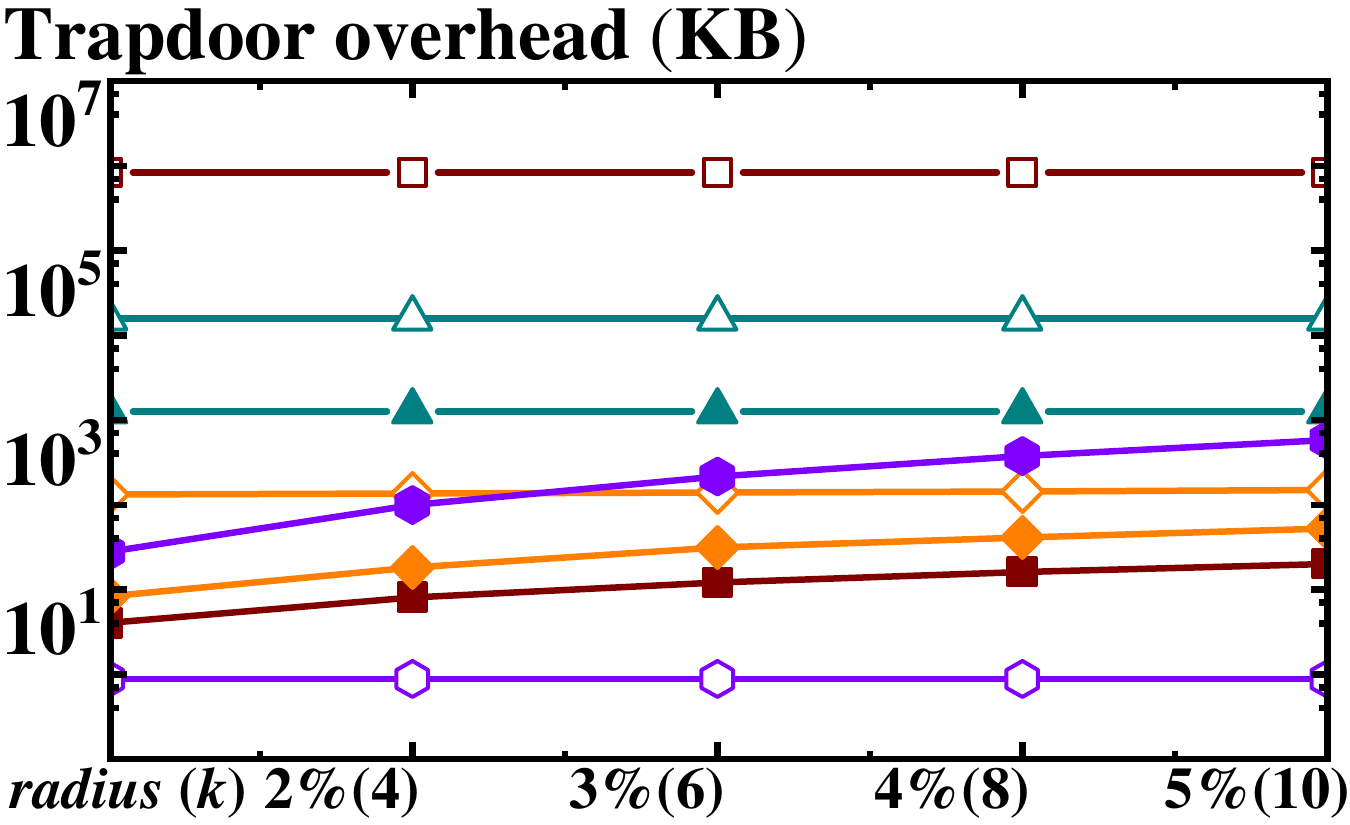}\label{RISK:fig:Comparisons:TrapOver:Twitter}}
	\subfigure[Trapdoor generation time.]{\includegraphics[width=42.5mm]{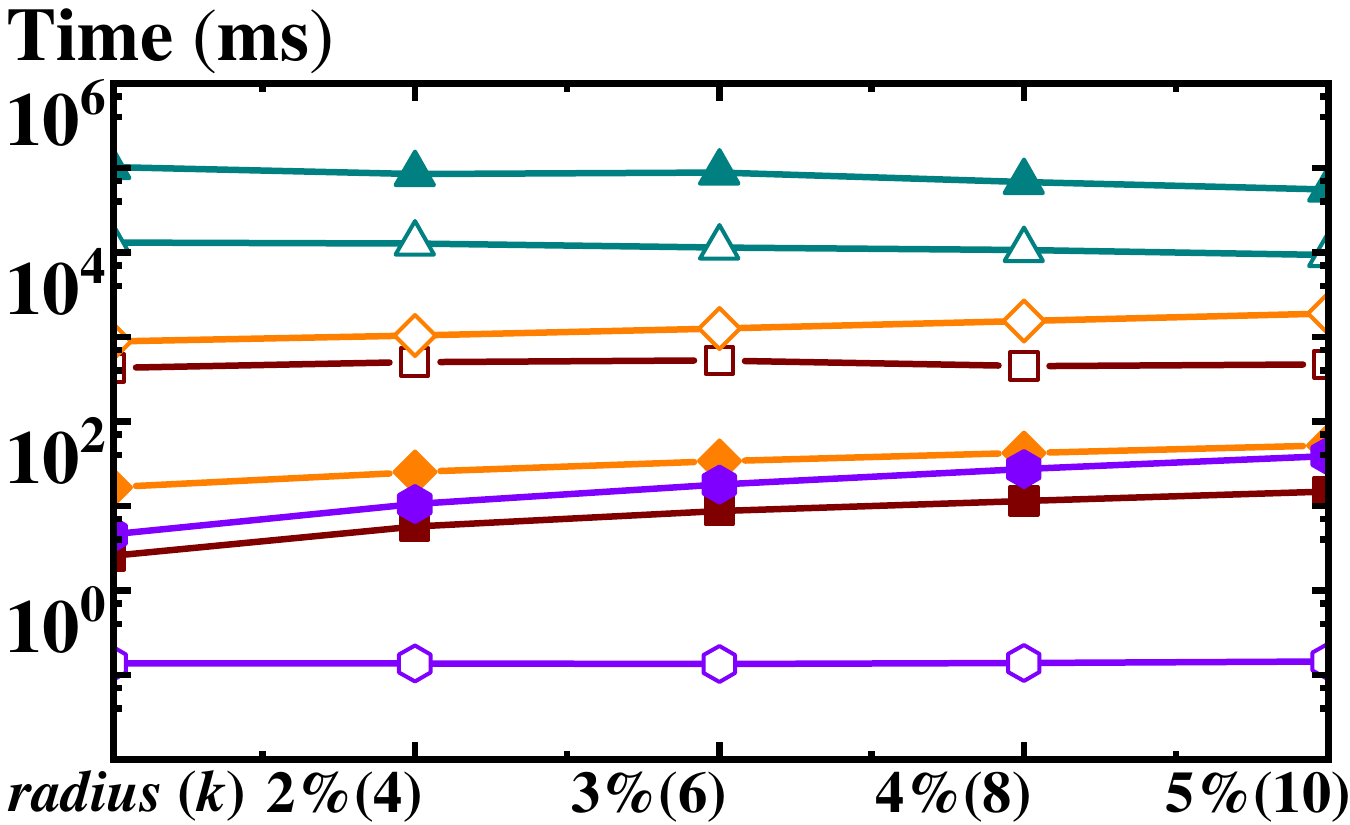}\label{RISK:fig:Comparisons:TrapTime:Twitter}}
	\caption{Experimental comparisons on \emph{Twitter} dataset.}
	\label{RISK:fig:Comparisons}
\end{figure}

\emph{Storage cost comparisons.} As illustrated in Fig.~\ref{RISK:fig:Comparisons:Storage:Twitter}, RISK reduces storage costs by $1$--$3$ orders of magnitude, primarily due to minimized encrypted objects overhead. Unlike other SOTA schemes (except RASK) that employ PPEs (i.e., ASPE and HVE) inducing substantial storage expansion. RISK adopts keyed hash function and symmetric encryption to maintain ciphertext lengths nearly identical to plaintext.

\emph{Total response time comparisons.} As shown in Fig.~\ref{RISK:fig:Comparisons:Response:Twitter}, for RSK queries, RISK reduces total response time by $0.8$--$4.4$ orders of magnitude relative to SOTA schemes. Response time scales with query radius, as the volume of returned encrypted candidates increases. For \emph{k}SK queries, RISK achieves $4$ orders of magnitude reduction compared with PBKQ~\cite{Song2024IoTsJ}. Fixed returned cell counts stabilize candidate set sizes, ensuring consistent response time despite increasing $k$, attributed to  efficient encryption and the \emph{k}Q-tree robust filtering.

\emph{Trapdoor generation performances comparisons.} As depicted in Figs.~\ref{RISK:fig:Comparisons:TrapOver:Twitter} and \ref{RISK:fig:Comparisons:TrapTime:Twitter}, RISK matches SOTA trapdoor generation performance (RASK and PBRQ) for RSK queries. By contrast, it reduces \emph{k}SK query trapdoor generation time by $4$ orders of magnitude compared to PBKQ, as the keyed hash function and symmetric encryption employed in RISK significantly outperform EASPE and HVE.

Overall, RISK delivers remarkable advantages in storage overhead, with a reduction of at least $1$ order of magnitude. It also achieves significant speedups in total response time over all SOTA methods, with no less than $0.5$ order of magnitude for RSK queries and up to $4$ orders of magnitude for \emph{k}SK queries. Notably, RISK first natively supports secure RSK and \emph{k}SK queries, whereas existing schemes are restricted to one.

\subsection{Experimental performances}
\label{RISK:Experi:Performances}
We evaluate RISK on three real-world and one synthetic datasets, using RASK (for RSK queries) and PBKQ (for \emph{k}SK queries) as baselines. PBKQ is not evaluable on the \emph{NewYork}, \emph{Paris}, and \emph{Gaussian} datasets, as its prohibitive memory demands exceed the physical memory of the test machine.

\begin{figure}
	\centering
	\includegraphics[width=85mm]{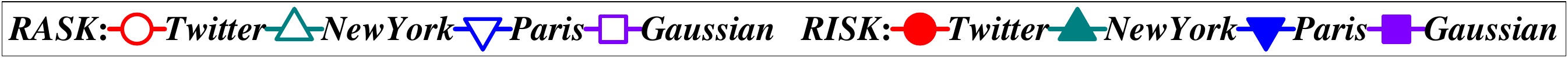}
	\subfigure[Cloud time for RSK query.]{\includegraphics[width=42.5mm]{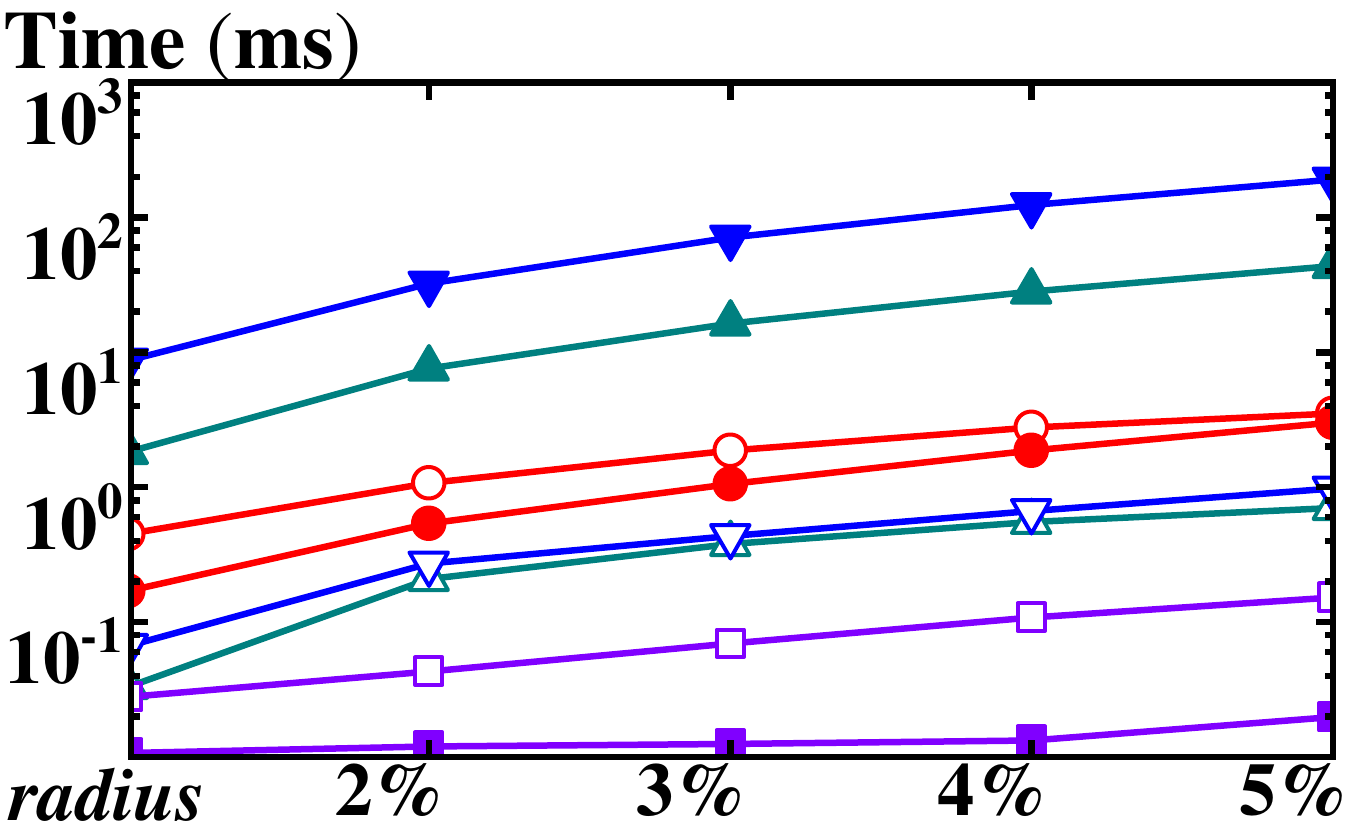}\label{RISK:fig:Performance:RSK:CloudRSK}}
	\subfigure[Client time for RSK query.]{\includegraphics[width=42.5mm]{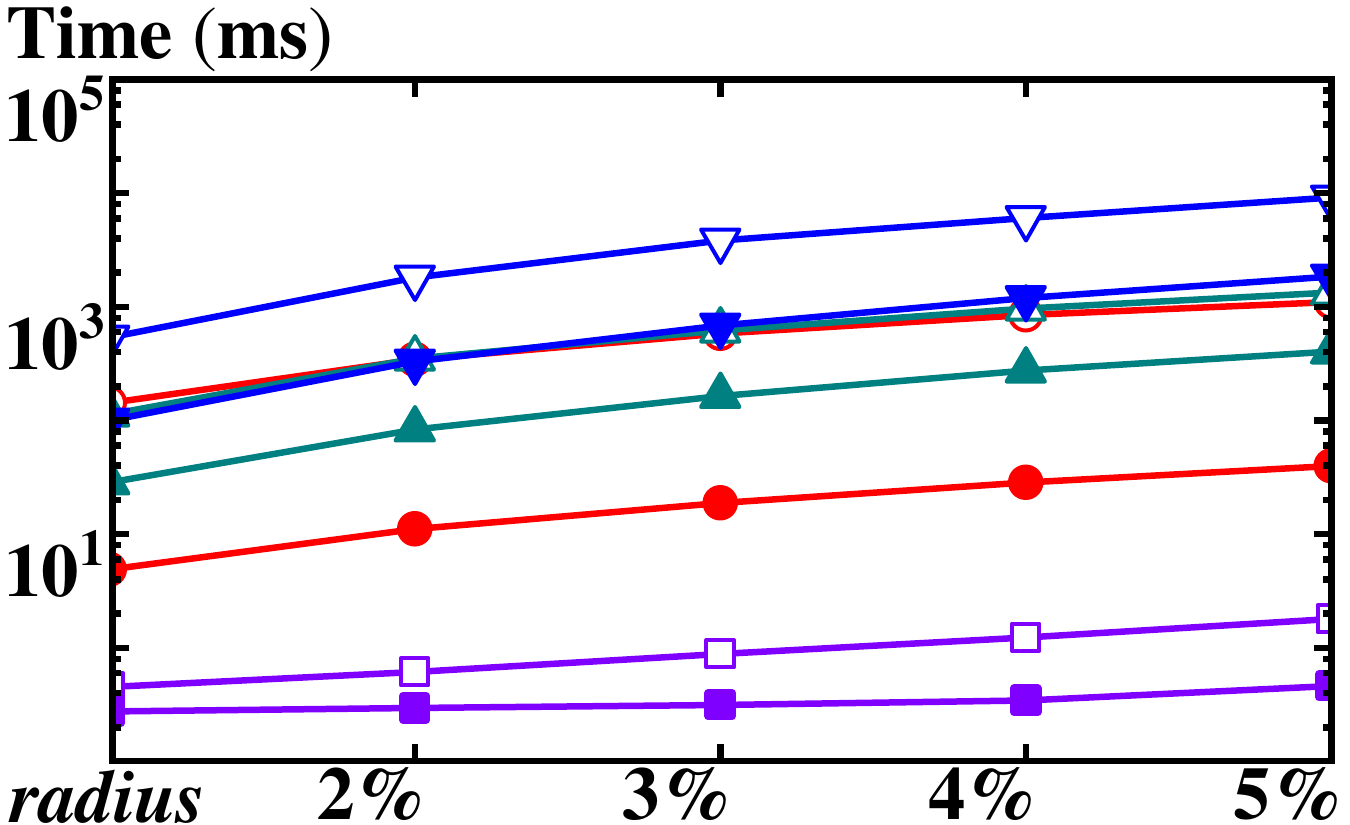}\label{RISK:fig:Performance:RSK:ClientRSK}}\\
	\subfigure[RSK's trapdoor overhead.]{\includegraphics[width=42.5mm]{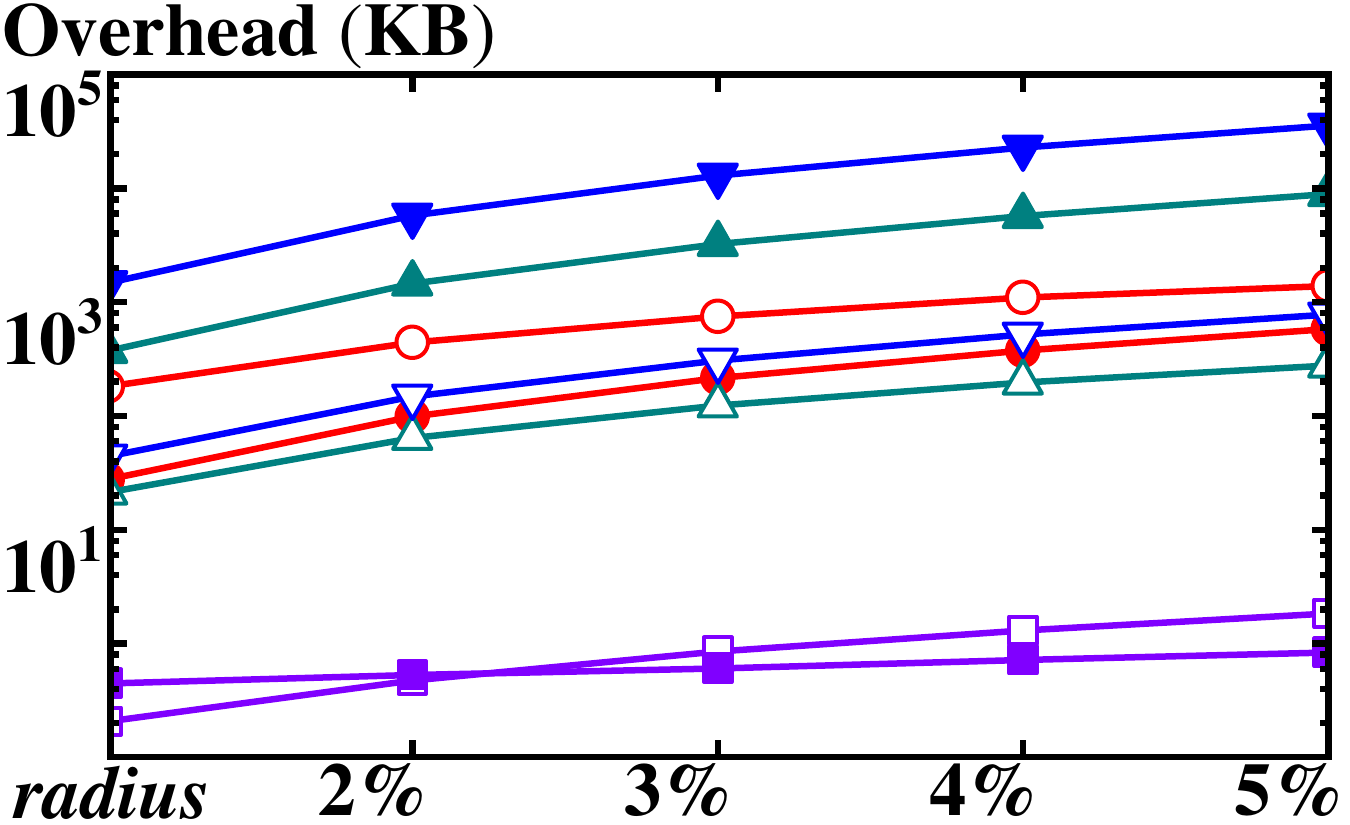}\label{RISK:fig:Performance:RSK:Trapdoor_RSK}}
	\subfigure[RSK's candidate overhead.]{\includegraphics[width=42.5mm]{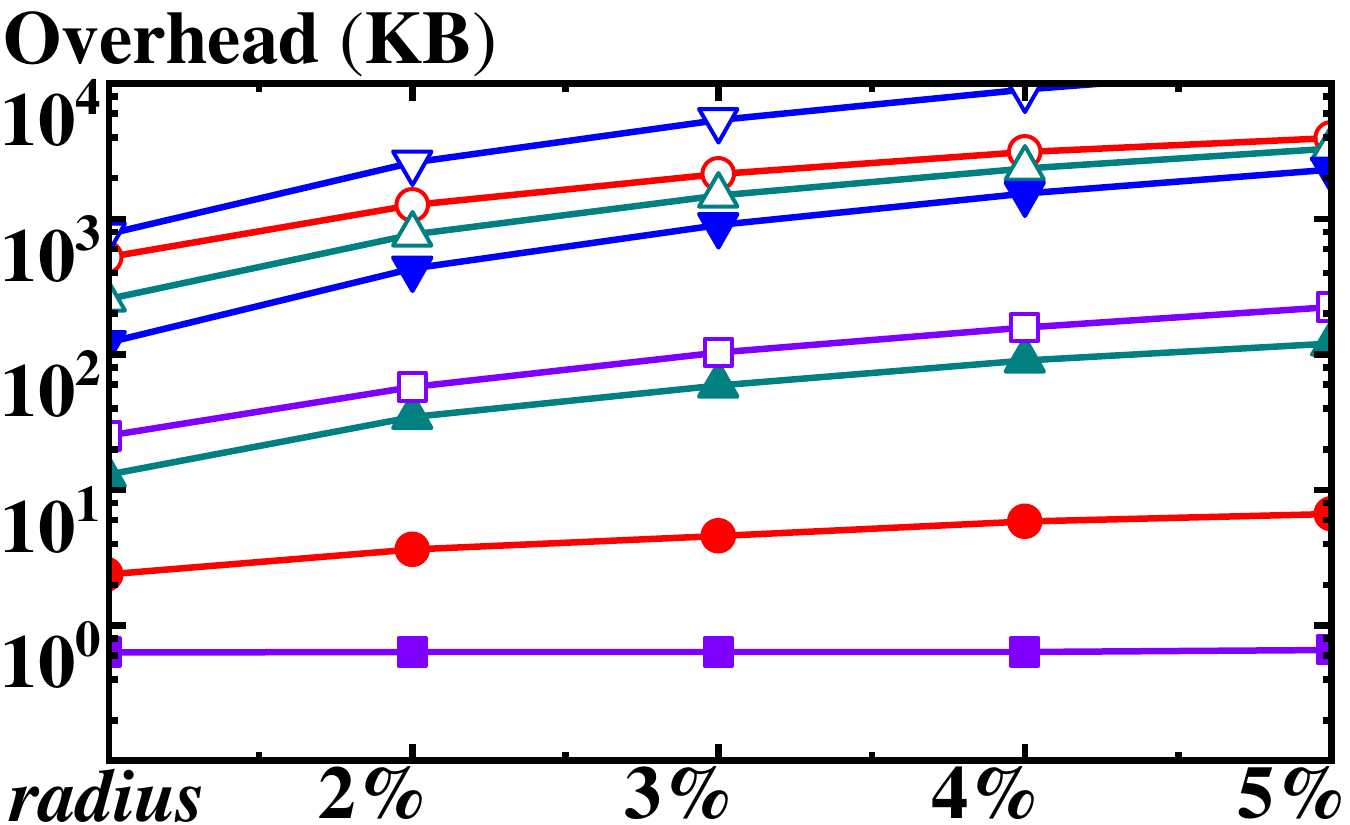}\label{RISK:fig:Performance:RSK:Candidates_RSK}}
	\caption{The performances for RSK query.}
	\label{RISK:fig:Performance:RSK}
\end{figure}

For RSK queries, total response time is the sum of cloud time and client time. RISK achieves a one-order-of-magnitude response time reduction over RASK across all these datasets, a discrepancy validated by the cloud-client processing time breakdown in Figs.~\ref{RISK:fig:Performance:RSK:CloudRSK} and \ref{RISK:fig:Performance:RSK:ClientRSK}. RISK's trapdoor cost correlates positively with dataset scale (Fig.~\ref{RISK:fig:Performance:RSK:Trapdoor_RSK}). Larger datasets yield finer S\emph{k}Q-Tree partitioning, increasing trapdoor element counts. Overall, RISK incurs higher trapdoor cost than RASK due to differing region partitioning mechanisms. RISK adopts S\emph{k}Q-Tree partitioning strategy independent of data distribution, while RASK's partitioning adapts dynamically, causing trapdoor cost volatility across dataset sizes and skew. However, S\emph{k}Q-Tree's robust textual filtering enables RISK to efficiently prune geo-textual objects not described by $q_r.\psi$. Consequently, RISK's encrypted candidates transmission cost is $0.7$ to $2.7$ orders of magnitude lower than RASK's (Fig.~\ref{RISK:fig:Performance:RSK:Candidates_RSK}). 

For \emph{k}SK queries, RISK outperforms PBKQ~\cite{Song2024IoTsJ} by $4$ orders of magnitude in response time, as depicted in Figs.~\ref{RISK:fig:Performance:KSK:CloudkSK} and \ref{RISK:fig:Performance:KSK:ClientkSK}. For both RISK and PBKQ, cloud-side response time exceeds that of the client-side. RISK generates trapdoors that are approximately $2$ orders of magnitude shorter than those of PBKQ (Fig.~\ref{RISK:fig:Performance:KSK:Trapdoor_kSK}), an advantage stemming from its avoidance of PPE. While RISK's candidate transmission overhead is significantly higher than that of PBKQ due to the introduction of false positives, these overheads are typically limited to several KBs, which is acceptable for most real-world applications (Fig.~\ref{RISK:fig:Performance:KSK:Candidates_kSK}). 

\begin{figure}[!t]
	\centering
	\includegraphics[width=85mm]{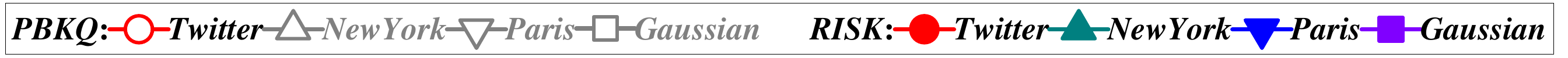}
	\subfigure[Cloud time for \emph{k}SK query.]{\includegraphics[width=42.5mm]{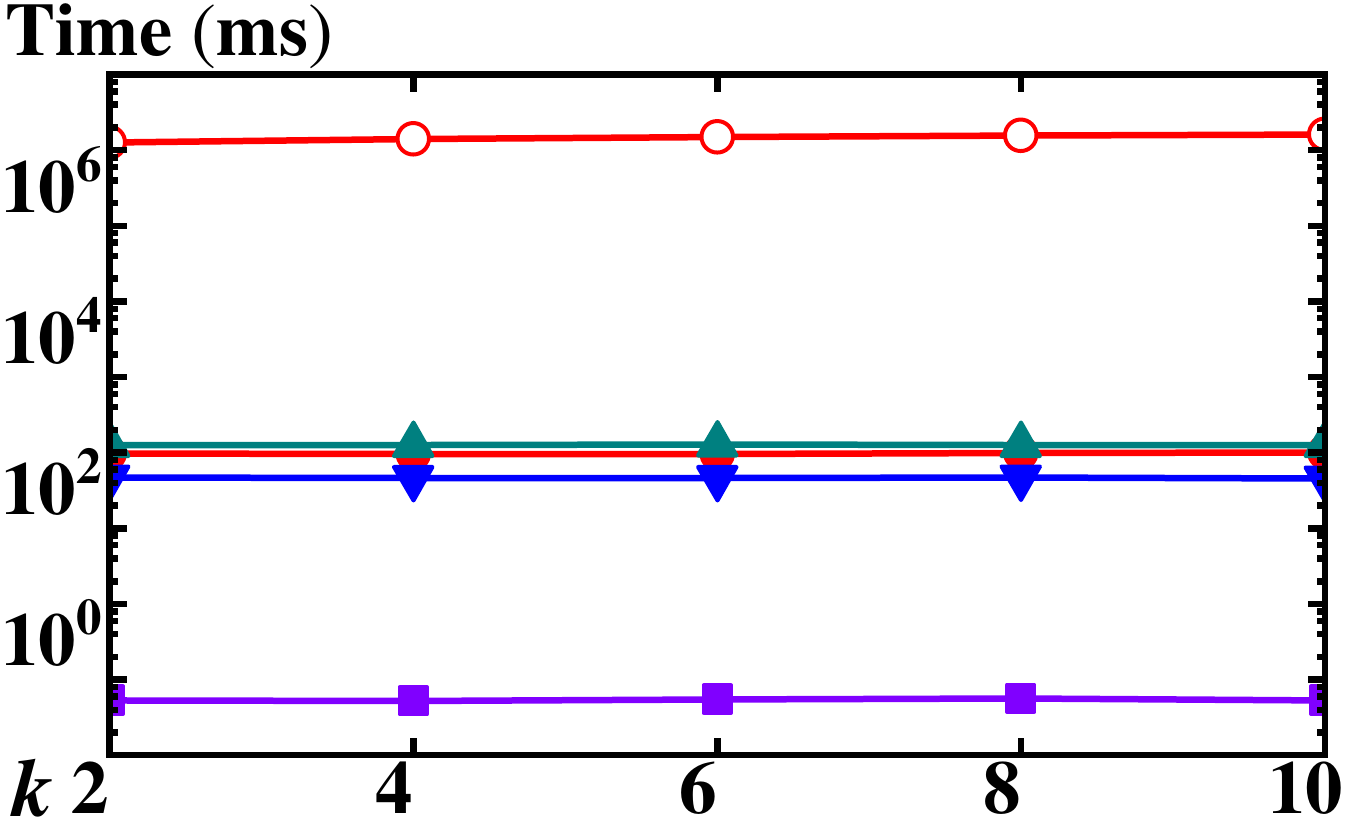}\label{RISK:fig:Performance:KSK:CloudkSK}}
	\subfigure[Client time for \emph{k}SK query.]{\includegraphics[width=42.5mm]{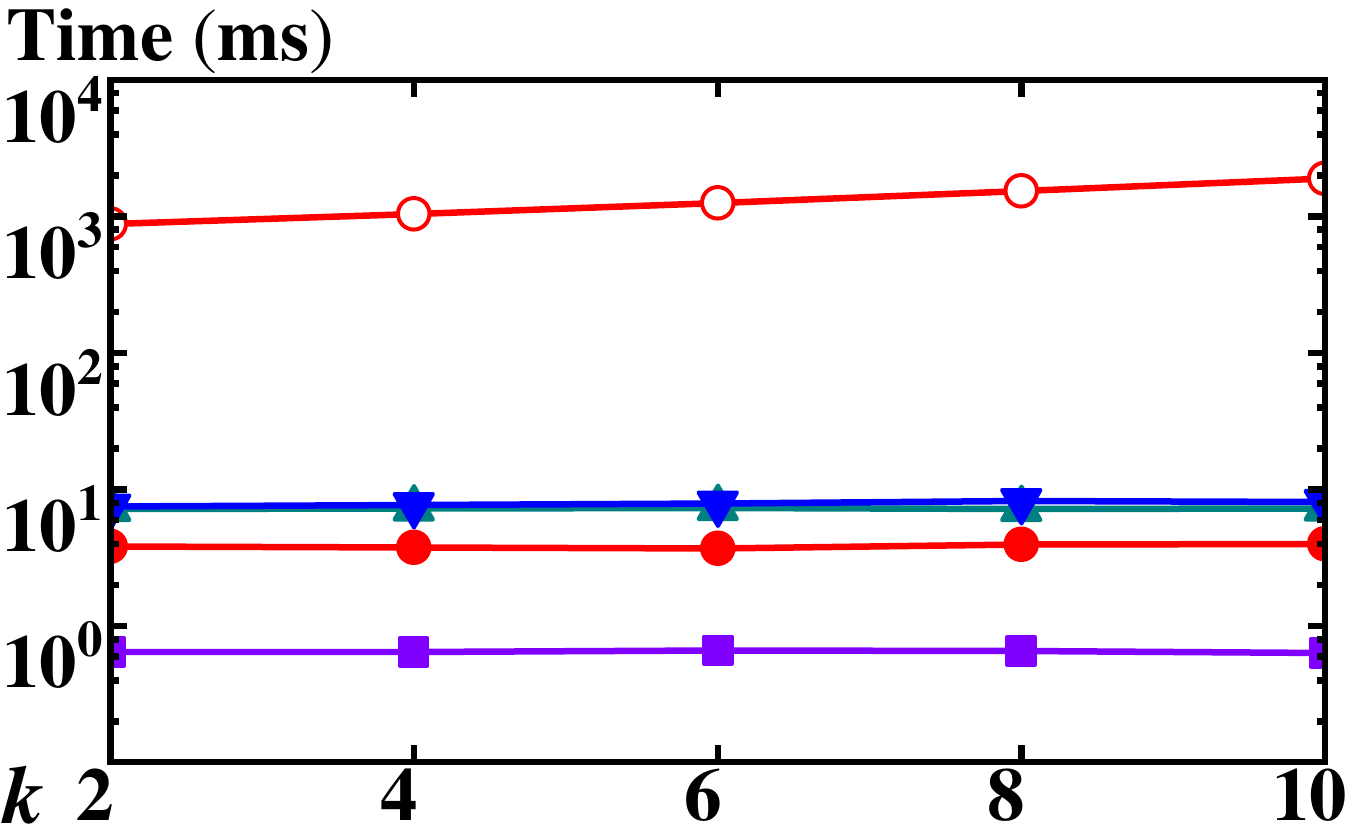}\label{RISK:fig:Performance:KSK:ClientkSK}}\\
	\subfigure[\emph{k}SK's trapdoor overhead.]{\includegraphics[width=42.5mm]{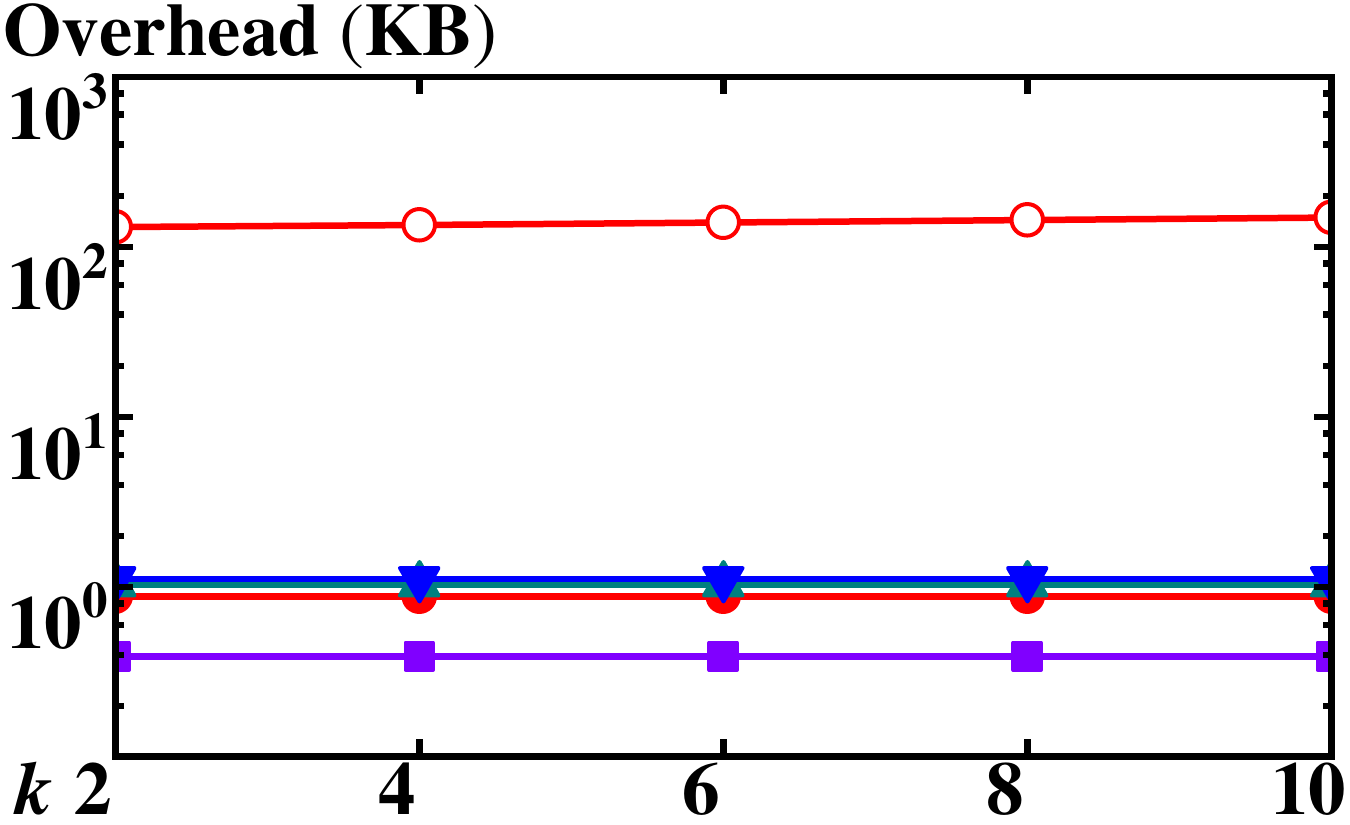}\label{RISK:fig:Performance:KSK:Trapdoor_kSK}}
	\subfigure[\emph{k}SK's candidate overhead.]{\includegraphics[width=42.5mm]{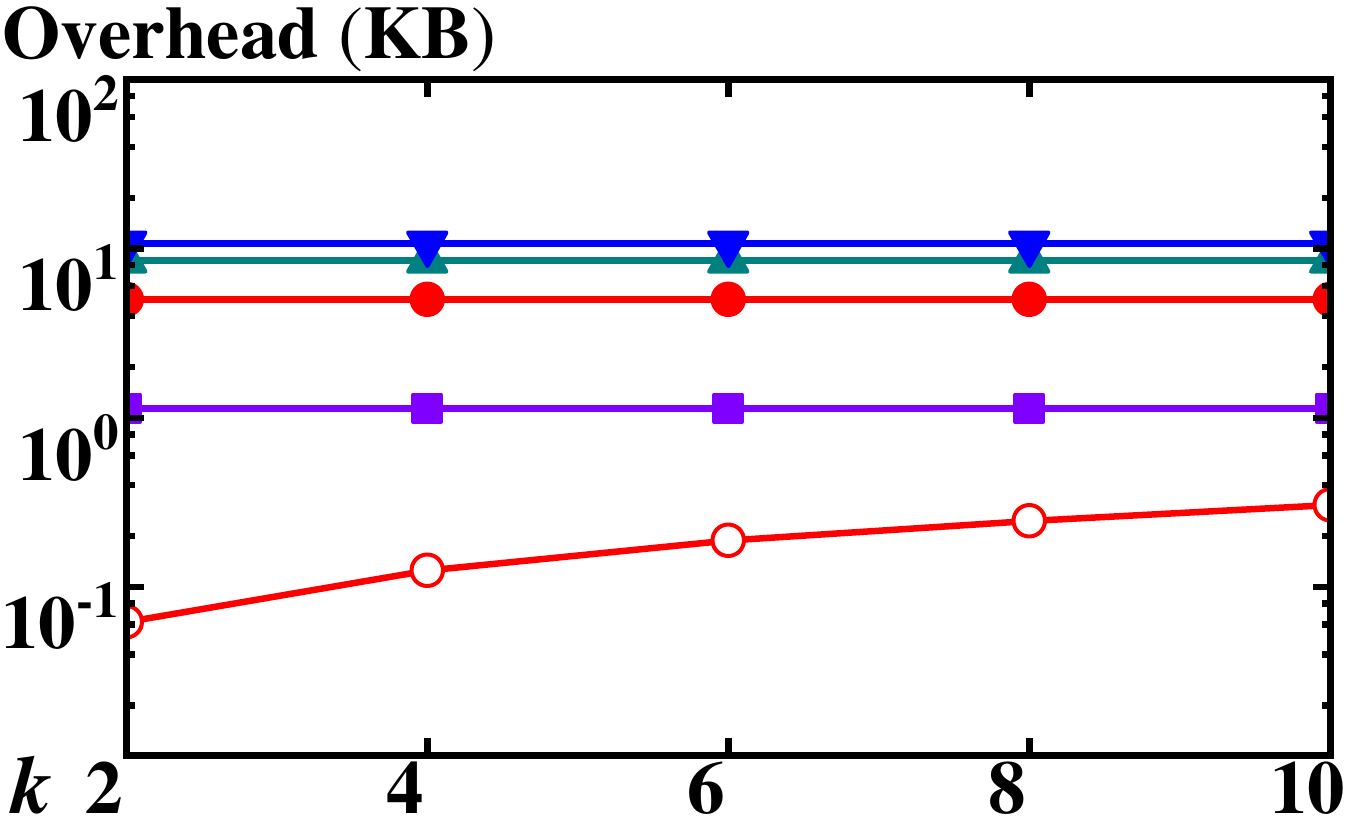}\label{RISK:fig:Performance:KSK:Candidates_kSK}}
	\caption{The performances for \emph{k}SK query.}
	\label{RISK:fig:Performance:KSK}
\end{figure}

\emph{Storage cost.} Index construction is executed at the client, with dominant overhead attributed to computing NNs to the representative objects and spatial division. As illustrated in Fig.~\ref{RISK:fig:Performance:Storage:Time}, RASK and RISK achieve nearly identical index building times, outperforming PBKQ by $0.4$ orders of magnitude. On the public cloud, storage overhead is dominated by hash values and multi-copy storage for each geo-textual object. Storage performance is summarized in Fig.~\ref{RISK:fig:Performance:Storage:Overhead}. Leveraging the high storage efficiency of the \emph{k}Q-tree, RISK achieves drastically lower storage costs than RASK and PBKQ, with storage space savings ranging from a minimum of $91\%$ to a maximum of $97\%$ relative to baseline methods.

\begin{figure}[!t]
	\centering
	\includegraphics[width=85mm]{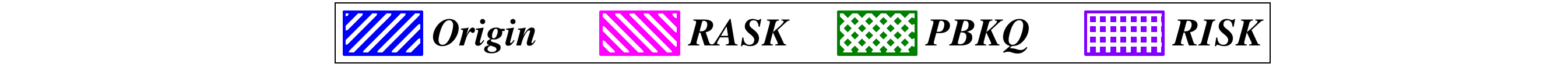}
	\subfigure[Index building time.]{\includegraphics[width=42.5mm]{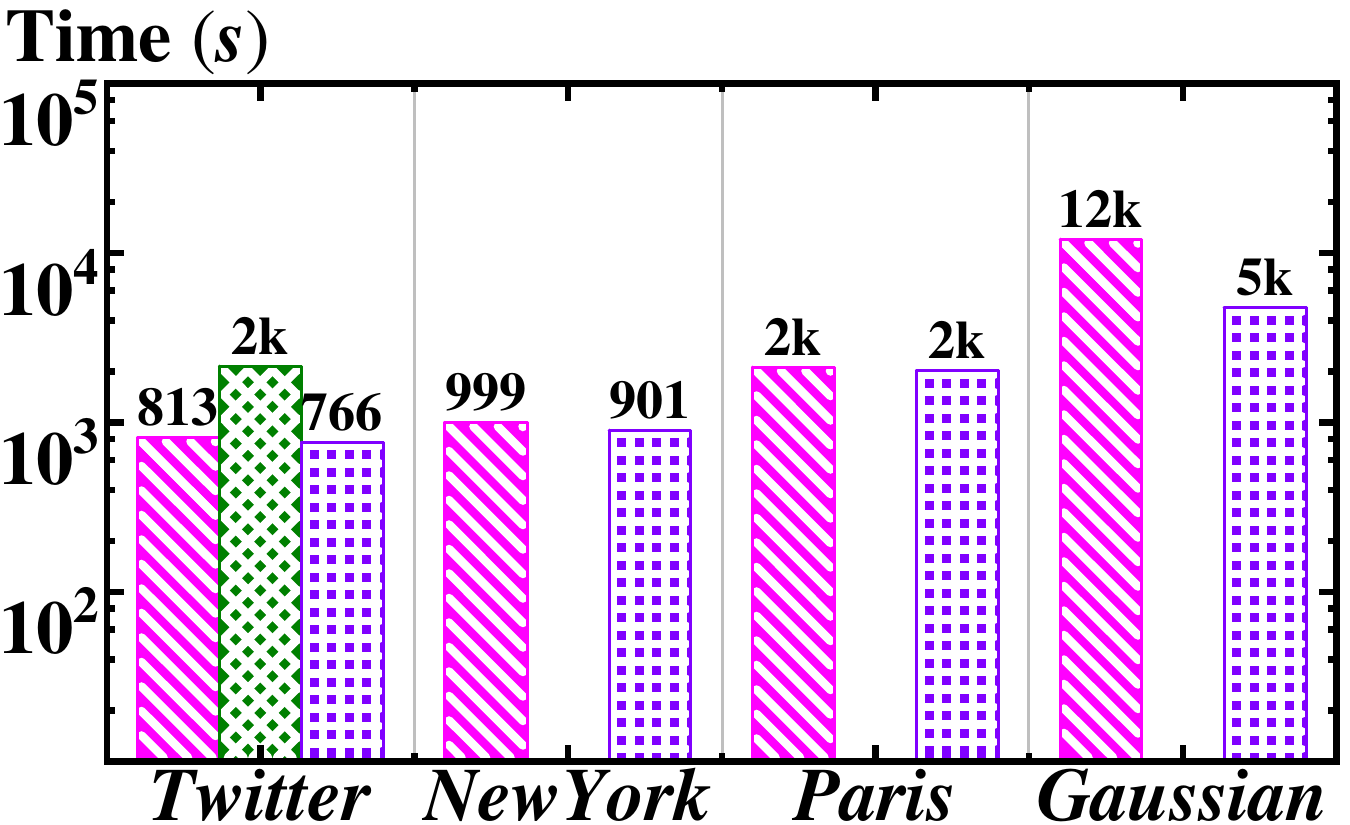}\label{RISK:fig:Performance:Storage:Time}}
	\subfigure[Storage overhead for cloud.]{\includegraphics[width=42.5mm]{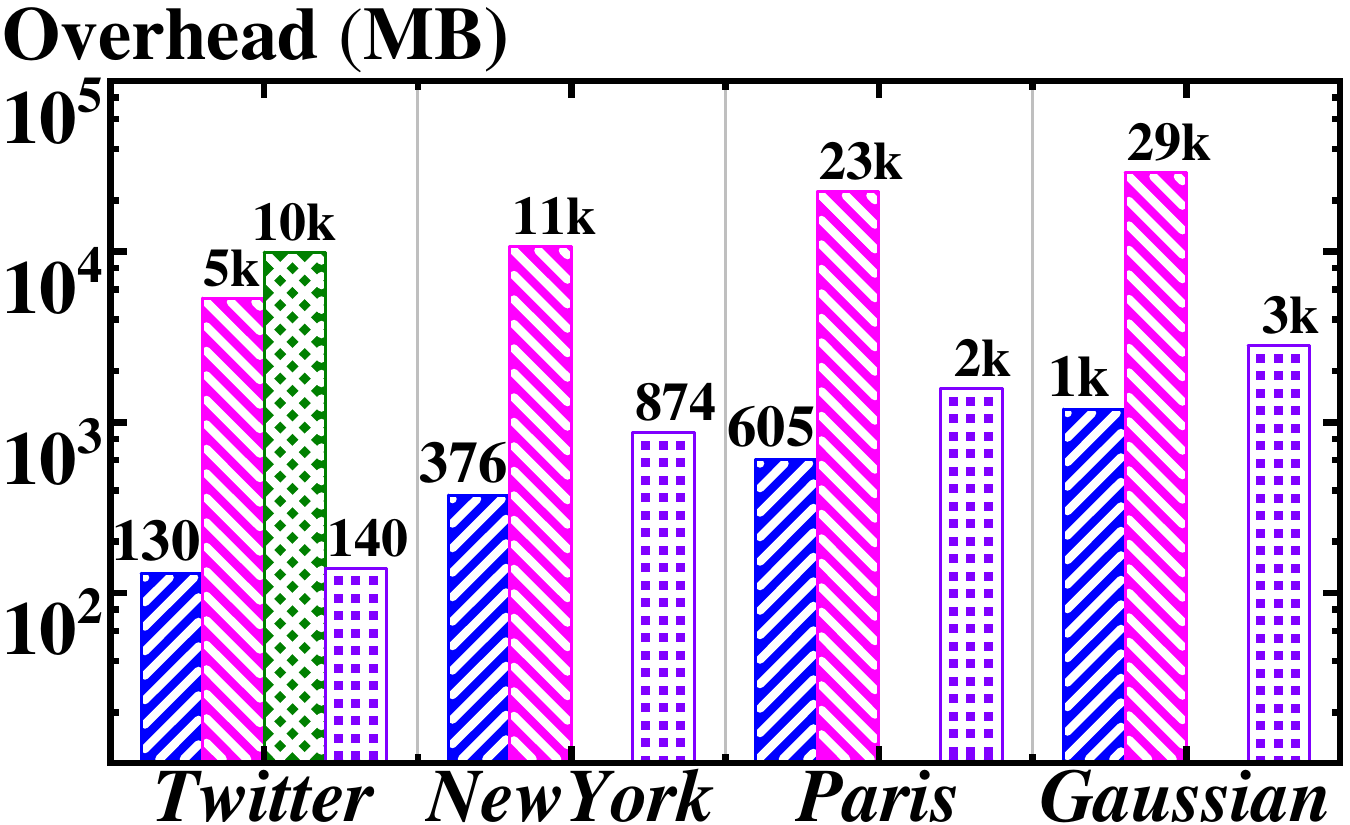}\label{RISK:fig:Performance:Storage:Overhead}}
	\caption{The performance of secure hybrid index.}
	\label{RISK:fig:Performance:Storage}
\end{figure}

\emph{Performance of the extensions to multi-party scenarios.} We vary the client count from  $20$ to $100$, where each client is assigned an independent random session key. Key negotiation between clients and the TEE employs RSA, with AES as the symmetric encryption algorithm. For all datasets, $100$ queries were executed $5$ times, and the total and average per-client per-query response time is reported. As shown in Fig.~\ref{RISK:fig:Performance:Extensions:MultiUserTotal} and~\ref{RISK:fig:Performance:Extensions:MultiUserAvg}, response time for RSK and \emph{k}SK queries remains stable across all datasets as the client count increases.
Per-client per-query response time increases marginally due to additional key negotiation, data encryption, and data decryption. For RSK and \emph{k}SK queries across all datasets, this increment is constrained to tens of milliseconds. The performance of \emph{k}SK queries is extremely close to that of RSK queries, as the majority of the response time is attributed to key agreement, communication overhead, etc. Extensions to multi-owner scenarios only prolong index construction time without degrading query performance, so detailed results are omitted.

\begin{figure}[!t]
	\centering
	\includegraphics[width=85mm]{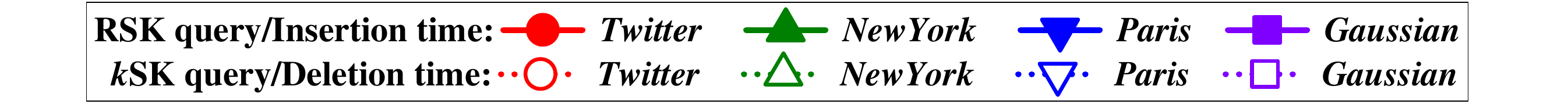}	
	\subfigure[Total time for multi-user.]{\includegraphics[width=42.5mm]{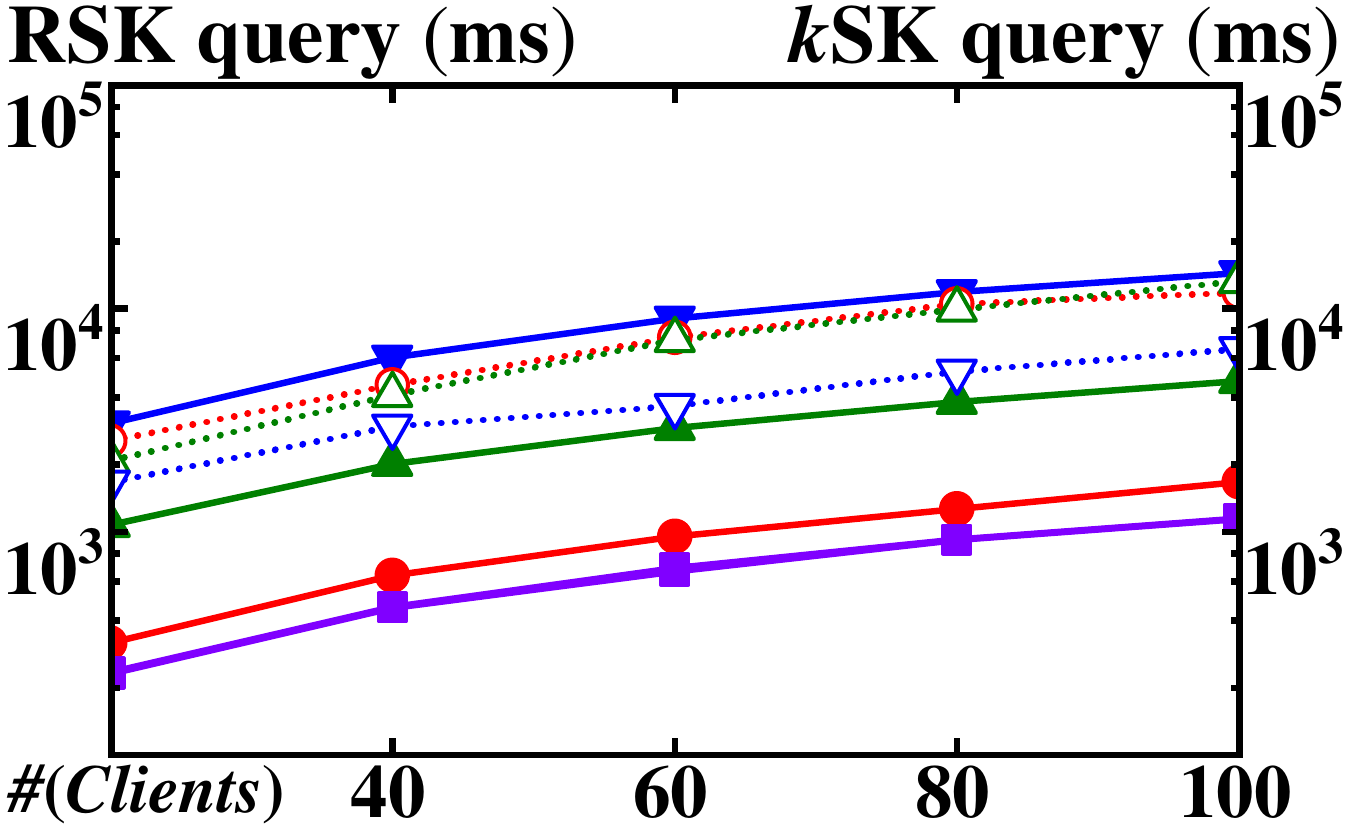}\label{RISK:fig:Performance:Extensions:MultiUserTotal}}
	\subfigure[Average time for multi-user.]{\includegraphics[width=42.5mm]{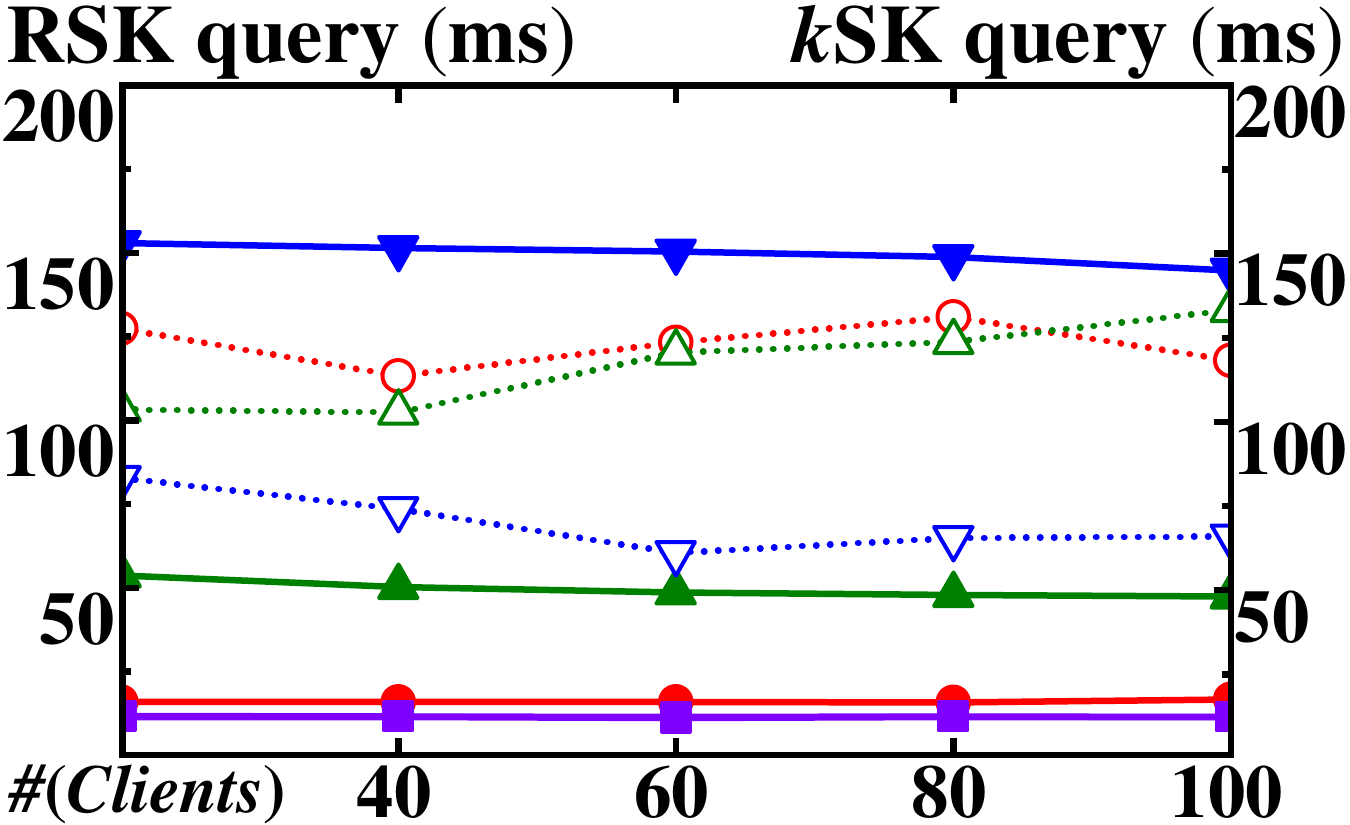}\label{RISK:fig:Performance:Extensions:MultiUserAvg}}\\
	\subfigure[Total time for data updates.]{\includegraphics[width=42.5mm]{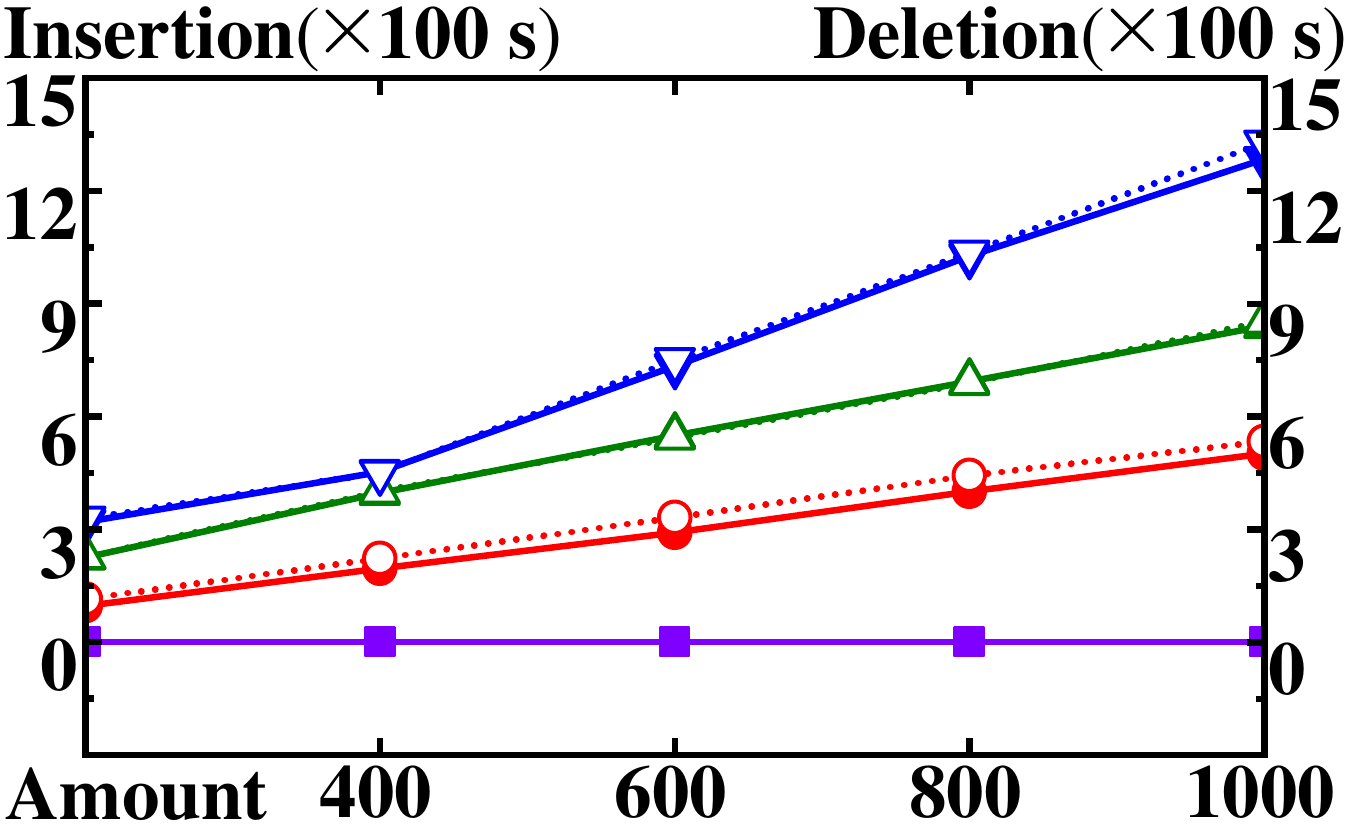}\label{RISK:fig:Performance:Extensions:DataUpdatesTotal}}
	\subfigure[Average time for data updates.]{\includegraphics[width=42.5mm]{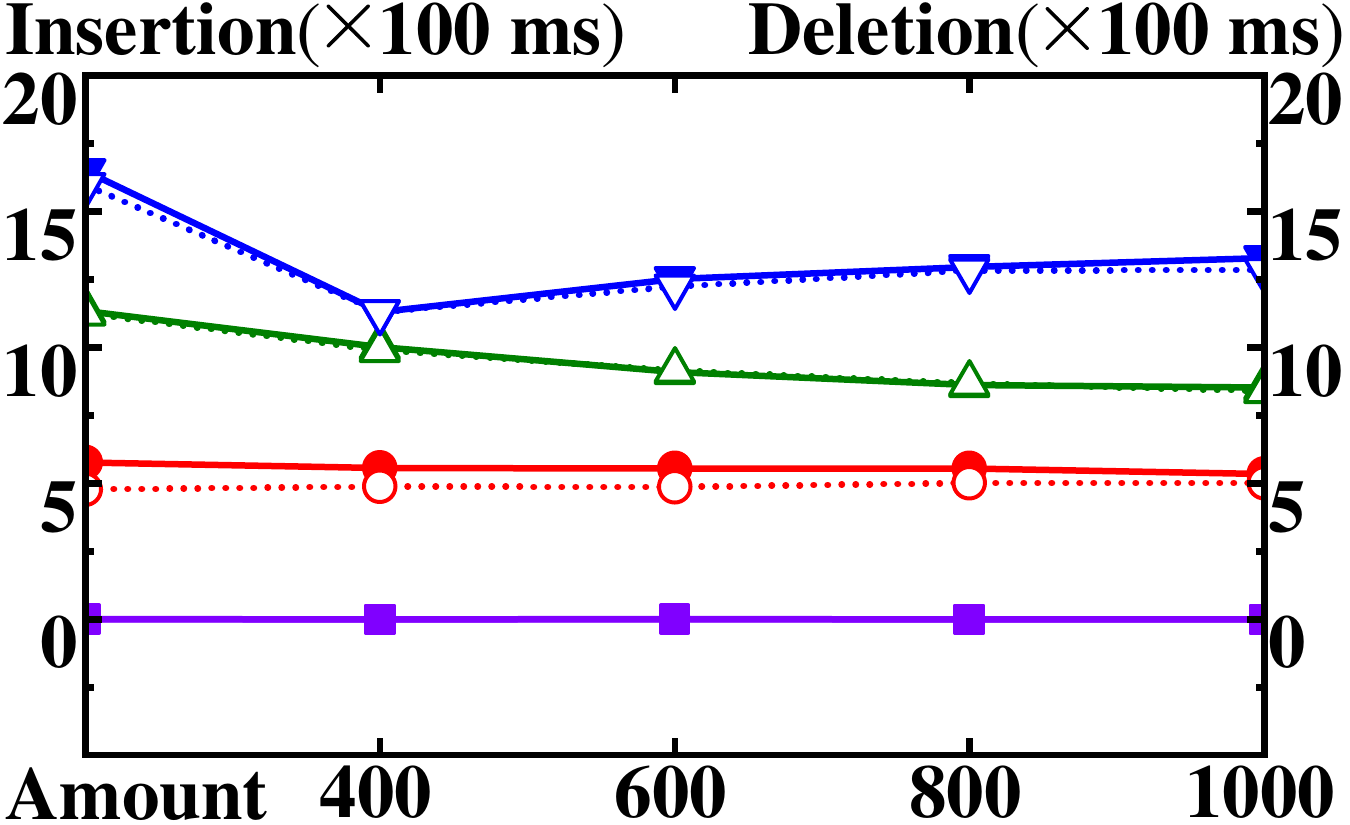}\label{RISK:fig:Performance:Extensions:DataUpdatesAvg}}
	\caption{Performances of RISK's extensions.}
	\label{RISK:fig:Performance:Extensions}
\end{figure}

\emph{Performance of the extension for data updates.} The volume of data updates is tuned from $200$ to $1000$ insertions and deletions on all datasets, with update performance summarized in Fig.~\ref{RISK:fig:Performance:Extensions:DataUpdatesTotal} and~\ref{RISK:fig:Performance:Extensions:DataUpdatesAvg}. As the update volume increases, insertion and deletion latencies grow approximately linearly across all datasets. On average, the time required to insert or delete a single object is marginally longer than that of a single \emph{k}SK query, with a maximum latency of $1,700$ ms. Data deletion performance is nearly identical to data insertion on the \emph{Gaussian} dataset, as the extremely low tree height yields negligible response times with no significant difference between the two operations. These performance metrics are indeed well within the acceptable range for practical applications.

\section{Conclusion}
\label{RISK:Conclusion}
RISK represents a promising attempt to enable secure RSK and \emph{k}SK queries in a unified framework, and it satisfies the IND-CKA2 security. Furthermore, RISK is extended to multi-party scenarios and dynamic updates. Experiments on three real-world and one synthetic datasets demonstrate that RISK outperforms SOTA schemes across multiple metrics. For millions of encrypted geo-textual objects, the response time for $10$-nearest neighbor or $1\%$ range queries is consistently under $150$ ms. RISK is the first solution that supports secure rich spatial-keyword queries based on a unified secure index.

In the future, secure spatial-keyword query will continue to confront challenges, including enabling secure Top-\emph{k} query, ensuring forward security, and supporting efficient large-scale data updates and diverse data models.

%

\section*{AI-Generated Content Acknowledgement}
We used Doubao AI for language polishing and revision in this submission to refine sentence structure, correct grammatical errors, and improve readability. All core content, including research design, algorithm design, theoretical analysis, experimental code, data analysis, and conclusions, is original to the human authors. Every suggestion from the model was critically evaluated and revised by the authors, who take full responsibility for the final content.

\bibliographystyle{IEEEtran}
\bibliography{ref_RISK}

@Article{Chen2021,
  author  = {Chen, Zhida and Chen, Lisi and Cong, Gao and Jensen, Christian},
  journal = {The VLDB Journal},
  title   = {Location- and keyword-based querying of geo-textual data: a survey},
  year    = {2021},
  month   = {07},
  pages   = {603--640},
  volume  = {30},
  doi     = {10.1007/s00778-021-00661-w},
}

@InProceedings{Wang2020INFOCOM,
  author    = {Wang, Xiangyu and Ma, Jianfeng and Liu, Ximeng and Deng, Robert H. and Miao, Yinbin and Zhu, Dan and Ma, Zhuoran},
  booktitle = {IEEE INFOCOM 2020 - IEEE Conference on Computer Communications},
  title     = {Search Me in the Dark: Privacy-preserving Boolean Range Query over Encrypted Spatial Data},
  year      = {2020},
  pages     = {2253--2262},
  doi       = {10.1109/INFOCOM41043.2020.9155505},
}

@Article{Wang2021TIFS,
  author  = {Wang, Xiangyu and Ma, Jianfeng and Li, Feng and Liu, Ximeng and Miao, Yinbin and Deng, Robert H.},
  journal = {IEEE Transactions on Information Forensics and Security},
  title   = {Enabling Efficient Spatial Keyword Queries on Encrypted Data With Strong Security Guarantees},
  year    = {2021},
  pages   = {4909--4923},
  volume  = {16},
  doi     = {10.1109/TIFS.2021.3118880},
}

@InProceedings{Cui2019ICDE,
  author       = {Cui, Ningning and Li, Jianxin and Yang, Xiaochun and Wang, Bin and Reynolds, Mark and Xiang, Yong},
  booktitle    = {2019 IEEE 35th International Conference on Data Engineering (ICDE)},
  title        = {When geo-text meets security: privacy-preserving boolean spatial keyword queries},
  year         = {2019},
  organization = {IEEE},
  pages        = {1046--1057},
  doi          = {10.1109/ICDE.2019.00097},
}

@InProceedings{Yang2022ICDCS,
  author    = {Yang, Yutao and Miao, Yinbin and Choo, Kim-Kwang Raymond and Deng, Robert H.},
  booktitle = {2022 IEEE 42nd International Conference on Distributed Computing Systems (ICDCS)},
  title     = {Lightweight Privacy-Preserving Spatial Keyword Query over Encrypted Cloud Data},
  year      = {2022},
  pages     = {392--402},
  doi       = {10.1109/ICDCS54860.2022.00045},
}

@Article{Lv2023,
  author  = {Lv, Zhen and Shang, Kaiyu and Huo, Hongwei and Liu, Ximeng and Peng, Yanguo and Wang, Xiangyu and Tan, Yaorong},
  journal = {IEEE Transactions on Services Computing},
  title   = {{RASK}: Range spatial keyword queries on massive encrypted geo-textual data},
  year    = {2023},
  pages   = {1-14},
  doi     = {10.1109/TSC.2023.3289654},
}

@Article{Chen2020,
  author   = {Chen, Lisi and Shang, Shuo and Yang, Chengcheng and Li, Jing},
  journal  = {GeoInformatica},
  title    = {Spatial keyword search: a survey},
  year     = {2020},
  issn     = {1573-7624},
  number   = {1},
  pages    = {85--106},
  volume   = {24},
  doi      = {10.1007/s10707-019-00373-y},
  refid    = {Chen2020},
}

@InProceedings{Zhou2005,
  author    = {Zhou, Yinghua and Xie, Xing and Wang, Chuang and Gong, Yuchang and Ma, Wei-Ying},
  booktitle = {Proceedings of the 14th ACM International Conference on Information and Knowledge Management},
  title     = {Hybrid index structures for location-based web search},
  year      = {2005},
  address   = {New York, NY, USA},
  pages     = {155–162},
  publisher = {Association for Computing Machinery},
  series    = {CIKM '05},
  doi       = {10.1145/1099554.1099584},
  isbn      = {1595931406},
  location  = {Bremen, Germany},
  numpages  = {8},
}

@InProceedings{Vaid2005STD,
  author    = {Vaid, Subodh and Jones, Christopher B. and Joho, Hideo and Sanderson, Mark},
  booktitle = {Advances in Spatial and Temporal Databases},
  title     = {Spatio-textual Indexing for Geographical Search on the Web},
  year      = {2005},
  address   = {Berlin, Heidelberg},
  editor    = {Bauzer Medeiros, Claudia and Egenhofer, Max J. and Bertino, Elisa},
  pages     = {218--235},
  publisher = {Springer Berlin Heidelberg},
  doi       = {10.1007/11535331_13},
  isbn      = {978-3-540-31904-7},
}

@InProceedings{Gobel2009CIKM,
  author    = {G\"{o}bel, Richard and Henrich, Andreas and Niemann, Raik and Blank, Daniel},
  booktitle = {Proceedings of the 18th ACM Conference on Information and Knowledge Management},
  title     = {A hybrid index structure for geo-textual searches},
  year      = {2009},
  address   = {New York, NY, USA},
  pages     = {1625–1628},
  publisher = {Association for Computing Machinery},
  series    = {CIKM '09},
  doi       = {10.1145/1645953.1646188},
  isbn      = {9781605585123},
  location  = {Hong Kong, China},
  numpages  = {4},
}

@InProceedings{Chen2006SIGMOD,
  author    = {Chen, Yen-Yu and Suel, Torsten and Markowetz, Alexander},
  booktitle = {Proceedings of the 2006 ACM SIGMOD International Conference on Management of Data},
  title     = {Efficient query processing in geographic web search engines},
  year      = {2006},
  address   = {New York, NY, USA},
  pages     = {277–288},
  publisher = {Association for Computing Machinery},
  series    = {SIGMOD '06},
  doi       = {10.1145/1142473.1142505},
  isbn      = {1595934340},
  location  = {Chicago, IL, USA},
  numpages  = {12},
}

@InProceedings{Christoforaki2011CIKM,
  author    = {Christoforaki, Maria and He, Jinru and Dimopoulos, Constantinos and Markowetz, Alexander and Suel, Torsten},
  booktitle = {Proceedings of the 20th ACM International Conference on Information and Knowledge Management},
  title     = {Text vs. space: efficient geo-search query processing},
  year      = {2011},
  address   = {New York, NY, USA},
  pages     = {423–432},
  publisher = {Association for Computing Machinery},
  series    = {CIKM '11},
  doi       = {10.1145/2063576.2063641},
  isbn      = {9781450307178},
  location  = {Glasgow, Scotland, UK},
  numpages  = {10},
}

@InProceedings{Ma2013WAIM,
  author    = {Ma, Youzhong and Zhang, Yu and Meng, Xiaofeng},
  booktitle = {Web-Age Information Management},
  title     = {{ST-HBase}: A Scalable Data Management System for Massive Geo-tagged Objects},
  year      = {2013},
  address   = {Berlin, Heidelberg},
  editor    = {Wang, Jianyong and Xiong, Hui and Ishikawa, Yoshiharu and Xu, Jianliang and Zhou, Junfeng},
  pages     = {155--166},
  publisher = {Springer Berlin Heidelberg},
  doi       = {10.1007/978-3-642-38562-9_16},
  isbn      = {978-3-642-38562-9},
}

@InProceedings{Li2019ICDE,
  author    = {Li, Rui and Liu, Alex X. and Liu, Ying and Xu, Huanle and Yuan, Huaqiang},
  booktitle = {2019 IEEE 35th International Conference on Data Engineering (ICDE)},
  title     = {Insecurity and Hardness of Nearest Neighbor Queries Over Encrypted Data},
  year      = {2019},
  pages     = {1614-1617},
  doi       = {10.1109/ICDE.2019.00155},
}

@InProceedings{DeFelipe2008ICDE,
  author    = {De Felipe, Ian and Hristidis, Vagelis and Rishe, Naphtali},
  booktitle = {2008 IEEE 24th International Conference on Data Engineering},
  title     = {Keyword Search on Spatial Databases},
  year      = {2008},
  pages     = {656-665},
  doi       = {10.1109/ICDE.2008.4497474},
}

@Article{Tao2014TKDE,
  author   = {Tao, Yufei and Sheng, Cheng},
  journal  = {IEEE Transactions on Knowledge and Data Engineering},
  title    = {Fast Nearest Neighbor Search with Keywords},
  year     = {2014},
  number   = {4},
  pages    = {878-888},
  volume   = {26},
  doi      = {10.1109/TKDE.2013.66},
}

@Article{Wu2012TKDE,
  author   = {Wu, Dingming and Yiu, Man Lung and Cong, Gao and Jensen, Christian S.},
  journal  = {IEEE Transactions on Knowledge and Data Engineering},
  title    = {Joint Top-K Spatial Keyword Query Processing},
  year     = {2012},
  number   = {10},
  pages    = {1889-1903},
  volume   = {24},
  doi      = {10.1109/TKDE.2011.172},
}

@Article{Zhang2016TKDE,
  author   = {Zhang, Chengyuan and Zhang, Ying and Zhang, Wenjie and Lin, Xuemin},
  journal  = {IEEE Transactions on Knowledge and Data Engineering},
  title    = {Inverted Linear Quadtree: Efficient Top K Spatial Keyword Search},
  year     = {2016},
  number   = {7},
  pages    = {1706-1721},
  volume   = {28},
  doi      = {10.1109/TKDE.2016.2530060},
}

@InProceedings{Zhang2013EDBT,
  author    = {Zhang, Dongxiang and Tan, Kian-Lee and Tung, Anthony K. H.},
  booktitle = {Proceedings of the 16th International Conference on Extending Database Technology},
  title     = {Scalable top-k spatial keyword search},
  year      = {2013},
  address   = {New York, NY, USA},
  pages     = {359–370},
  publisher = {Association for Computing Machinery},
  series    = {EDBT '13},
  doi       = {10.1145/2452376.2452419},
  isbn      = {9781450315975},
  location  = {Genoa, Italy},
  numpages  = {12},
}

@Article{Song2024IoTsJ,
  author   = {Song, Yu and Yu, Jia and Ge, Xinrui and Hao, Rong},
  journal  = {IEEE Internet of Things Journal},
  title    = {Enabling Privacy-Preserving Boolean {kNN} Query Over Cloud-Based Spatial Data},
  year     = {2024},
  number   = {23},
  pages    = {38262-38272},
  volume   = {11},
  doi      = {10.1109/JIOT.2024.3445170},
}

@InProceedings{Wong2009SIGMOD,
  author    = {Wong, Wai Kit and Cheung, David Wai-lok and Kao, Ben and Mamoulis, Nikos},
  booktitle = {Proceedings of the 2009 ACM SIGMOD International Conference on Management of Data},
  title     = {Secure {kNN} computation on encrypted databases},
  year      = {2009},
  address   = {New York, NY, USA},
  pages     = {139–152},
  publisher = {Association for Computing Machinery},
  series    = {SIGMOD '09},
  doi       = {10.1145/1559845.1559862},
  isbn      = {9781605585512},
  location  = {Providence, Rhode Island, USA},
  numpages  = {14},
}

@Article{Tong2023TKDE,
  author   = {Tong, Qiuyun and Miao, Yinbin and Weng, Jian and Liu, Ximeng and Choo, Kim-Kwang Raymond and Deng, Robert H.},
  journal  = {IEEE Transactions on Knowledge and Data Engineering},
  title    = {Verifiable Fuzzy Multi-Keyword Search Over Encrypted Data With Adaptive Security},
  year     = {2023},
  number   = {5},
  pages    = {5386-5399},
  volume   = {35},
  doi      = {10.1109/TKDE.2022.3152033},
}

@InProceedings{Curtmola2006CCS,
  author    = {Curtmola, Reza and Garay, Juan and Kamara, Seny and Ostrovsky, Rafail},
  booktitle = {Proceedings of the 13th ACM Conference on Computer and Communications Security},
  title     = {Searchable symmetric encryption: improved definitions and efficient constructions},
  year      = {2006},
  address   = {New York, NY, USA},
  pages     = {79–88},
  publisher = {Association for Computing Machinery},
  series    = {CCS '06},
  doi       = {10.1145/1180405.1180417},
  isbn      = {1595935185},
  location  = {Alexandria, Virginia, USA},
  numpages  = {10},
}

@Article{Jiang2022TIFS,
  author   = {Jiang, Qin and Chang, Ee-Chien and Qi, Yong and Qi, Saiyu and Wu, Pengfei and Wang, Jianfeng},
  journal  = {IEEE Transactions on Information Forensics and Security},
  title    = {Rphx: Result Pattern Hiding Conjunctive Query Over Private Compressed Index Using Intel SGX},
  year     = {2022},
  pages    = {1053-1068},
  volume   = {17},
  doi      = {10.1109/TIFS.2022.3144877},
}

@Article{Finkel1974Acta,
  author   = {Finkel, R. A. and Bentley, J. L.},
  journal  = {Acta Informatica},
  title    = {Quad trees a data structure for retrieval on composite keys},
  year     = {1974},
  issn     = {1432-0525},
  number   = {1},
  pages    = {1--9},
  volume   = {4},
  doi      = {10.1007/BF00288933},
  refid    = {Finkel1974},
}

@Article{Li2024TKDE,
  author   = {Li, Mingxin and Wang, Hancheng and Dai, Haipeng and Li, Meng and Chai, Chengliang and Gu, Rong and Chen, Feng and Chen, Zhiyuan and Li, Shuaituan and Liu, Qizhi and Chen, Guihai},
  journal  = {IEEE Transactions on Knowledge and Data Engineering},
  title    = {A Survey of Multi-Dimensional Indexes: Past and Future Trends},
  year     = {2024},
  number   = {8},
  pages    = {3635-3655},
  volume   = {36},
  doi      = {10.1109/TKDE.2024.3364183},
}

@Article{Samet1984CSUR,
  author     = {Samet, Hanan},
  journal    = {ACM Comput. Surv.},
  title      = {The Quadtree and Related Hierarchical Data Structures},
  year       = {1984},
  issn       = {0360-0300},
  month      = jun,
  number     = {2},
  pages      = {187–260},
  volume     = {16},
  address    = {New York, NY, USA},
  doi        = {10.1145/356924.356930},
  issue_date = {June 1984},
  numpages   = {74},
  publisher  = {Association for Computing Machinery},
}

@InProceedings{Hjaltason1995ASD,
  author    = {Hjaltason, G{\'i}sli R. and Samet, Hanan},
  booktitle = {Advances in Spatial Databases},
  title     = {Ranking in spatial databases},
  year      = {1995},
  address   = {Berlin, Heidelberg},
  editor    = {Egenhofer, Max J. and Herring, John R.},
  pages     = {83--95},
  publisher = {Springer Berlin Heidelberg},
  doi       = {doi.org/10.1007/3-540-60159-7_6},
  isbn      = {978-3-540-49536-9},
}

@InProceedings{Fox1992SIGIR,
  author    = {Fox, Edward A. and Chen, Qi Fan and Heath, Lenwood S.},
  booktitle = {Proceedings of the 15th Annual International ACM SIGIR Conference on Research and Development in Information Retrieval},
  title     = {A faster algorithm for constructing minimal perfect hash functions},
  year      = {1992},
  address   = {New York, NY, USA},
  pages     = {266–273},
  publisher = {Association for Computing Machinery},
  series    = {SIGIR '92},
  doi       = {10.1145/133160.133209},
  isbn      = {0897915232},
  location  = {Copenhagen, Denmark},
  numpages  = {8},
}

@InProceedings{Pibiri2021SIGIR,
  author    = {Pibiri, Giulio Ermanno and Trani, Roberto},
  booktitle = {Proceedings of the 44th International ACM SIGIR Conference on Research and Development in Information Retrieval},
  title     = {{PTHash}: Revisiting FCH Minimal Perfect Hashing},
  year      = {2021},
  address   = {New York, NY, USA},
  pages     = {1339–1348},
  publisher = {Association for Computing Machinery},
  series    = {SIGIR '21},
  doi       = {10.1145/3404835.3462849},
  isbn      = {9781450380379},
  location  = {Virtual Event, Canada},
  numpages  = {10},
}

@Article{Zhou2024CSUR,
  author     = {Zhou, Sheng and Xu, Hongjia and Zheng, Zhuonan and Chen, Jiawei and Li, Zhao and Bu, Jiajun and Wu, Jia and Wang, Xin and Zhu, Wenwu and Ester, Martin},
  journal    = {ACM Comput. Surv.},
  title      = {A Comprehensive Survey on Deep Clustering: Taxonomy, Challenges, and Future Directions},
  year       = {2024},
  issn       = {0360-0300},
  month      = nov,
  number     = {3},
  volume     = {57},
  address    = {New York, NY, USA},
  articleno  = {69},
  doi        = {10.1145/3689036},
  issue_date = {March 2025},
  numpages   = {38},
  publisher  = {Association for Computing Machinery},
}

@Article{Xie2022CSUR,
  author     = {Xie, Yiqun and Shekhar, Shashi and Li, Yan},
  journal    = {ACM Comput. Surv.},
  title      = {Statistically-Robust Clustering Techniques for Mapping Spatial Hotspots: A Survey},
  year       = {2022},
  issn       = {0360-0300},
  month      = jan,
  number     = {2},
  volume     = {55},
  address    = {New York, NY, USA},
  articleno  = {36},
  doi        = {10.1145/3487893},
  issue_date = {February 2023},
  numpages   = {38},
  publisher  = {Association for Computing Machinery},
}

@Article{Canetti2004JACM,
  author     = {Canetti, Ran and Goldreich, Oded and Halevi, Shai},
  journal    = {J. ACM},
  title      = {The random oracle methodology, revisited},
  year       = {2004},
  issn       = {0004-5411},
  month      = jul,
  number     = {4},
  pages      = {557–594},
  volume     = {51},
  address    = {New York, NY, USA},
  doi        = {10.1145/1008731.1008734},
  issue_date = {July 2004},
  numpages   = {38},
  publisher  = {Association for Computing Machinery},
}

@InProceedings{Zhang2016USENIXSecurity,
  author    = {Yupeng Zhang and Jonathan Katz and Charalampos Papamanthou},
  booktitle = {25th USENIX Security Symposium (USENIX Security 16)},
  title     = {All Your Queries Are Belong to Us: The Power of {File-Injection} Attacks on Searchable Encryption},
  year      = {2016},
  address   = {Austin, TX},
  month     = aug,
  pages     = {707--720},
  publisher = {USENIX Association},
  isbn      = {978-1-931971-32-4},
}

@InProceedings{Lacharite2018SP,
  author    = {Lacharit{\'{e}}, Marie-Sarah and Minaud, Brice and Paterson, Kenneth G.},
  booktitle = {2018 IEEE Symposium on Security and Privacy (SP)},
  title     = {Improved Reconstruction Attacks on Encrypted Data Using Range Query Leakage},
  year      = {2018},
  pages     = {297-314},
  doi       = {10.1109/SP.2018.00002},
}

@Article{Li2023SCUR,
  author     = {Li, Feng and Ma, Jianfeng and Miao, Yinbin and Liu, Ximeng and Ning, Jianting and Deng, Robert H.},
  journal    = {ACM Comput. Surv.},
  title      = {A Survey on Searchable Symmetric Encryption},
  year       = {2023},
  issn       = {0360-0300},
  month      = nov,
  number     = {5},
  volume     = {56},
  address    = {New York, NY, USA},
  articleno  = {119},
  doi        = {10.1145/3617991},
  issue_date = {May 2024},
  numpages   = {42},
  publisher  = {Association for Computing Machinery},
}

@InProceedings{Lai2018CCS,
  author    = {Lai, Shangqi and Patranabis, Sikhar and Sakzad, Amin and Liu, Joseph K. and Mukhopadhyay, Debdeep and Steinfeld, Ron and Sun, Shi-Feng and Liu, Dongxi and Zuo, Cong},
  booktitle = {Proceedings of the 2018 ACM SIGSAC Conference on Computer and Communications Security},
  title     = {Result Pattern Hiding Searchable Encryption for Conjunctive Queries},
  year      = {2018},
  address   = {New York, NY, USA},
  pages     = {745--762},
  publisher = {Association for Computing Machinery},
  series    = {CCS '18},
  doi       = {10.1145/3243734.3243753},
  isbn      = {9781450356930},
  location  = {Toronto, Canada},
  numpages  = {18},
}

@Article{Chen2013VLDB,
  author     = {Chen, Lisi and Cong, Gao and Jensen, Christian S. and Wu, Dingming},
  journal    = {Proc. VLDB Endow.},
  title      = {Spatial keyword query processing: an experimental evaluation},
  year       = {2013},
  issn       = {2150-8097},
  month      = jan,
  number     = {3},
  pages      = {217–228},
  volume     = {6},
  doi        = {10.14778/2535569.2448955},
  issue_date = {January 2013},
  numpages   = {12},
  publisher  = {VLDB Endowment},
}

@TechReport{Barker2020NIST,
  author      = {Barker, Elaine and Dang, Quynh},
  institution = {National Institute of Standards and Technology},
  title       = {{NIST} special publication 800-57 part 1, revision 5, recommendation for key management: Part 1--general},
  year        = {2020},
  journal     = {NIST, Tech. Rep},
  pages       = {171},
}

@Article{Wu2025TDSC,
  author   = {Wu, Haotian and Peng, Zhe and Xiao, Jiang and Xue, Lei and Lin, Chenhao and Chung, Sai-Ho},
  journal  = {IEEE Transactions on Dependable and Secure Computing},
  title    = {{HeX}: Encrypted Rich Queries With Forward and Backward Privacy Using Trusted Hardware},
  year     = {2025},
  number   = {4},
  pages    = {3751-3765},
  volume   = {22},
  doi      = {10.1109/TDSC.2025.3540958},
}

@Article{Lv2025TKDE,
  author   = {Lv, Siyi and Huang, Yanyu and Li, Xinhao and Li, Tong and Guo, Liang and Chen, Xiaofeng and Liu, Zheli},
  journal  = {IEEE Transactions on Knowledge and Data Engineering},
  title    = {{LUNA}: Efficient Backward-Private Dynamic Symmetric Searchable Encryption Scheme With Secure Deletion in Encrypted Database},
  year     = {2025},
  number   = {4},
  pages    = {1961-1974},
  volume   = {37},
  doi      = {10.1109/TKDE.2023.3329234},
}

@Article{Wang2025TIFS,
  author   = {Wang, Haoyang and Fan, Kai and Yu, Chong and Zhang, Kuan and Li, Fenghua and Zhu, Haojin},
  journal  = {IEEE Transactions on Information Forensics and Security},
  title    = {Hide Yourself: Multi-Dimensional Range Queries for Responses-Hiding Over Outsourced Data},
  year     = {2025},
  pages    = {6923-6936},
  volume   = {20},
  doi      = {10.1109/TIFS.2025.3583252},
}

@Article{Chen2025SIGMOD,
  author     = {Chen, Yahong and Pang, Xiaoyi and Li, Xiaoguang and Wang, Hanyi and Niu, Ben and Hu, Shengnan},
  journal    = {Proc. ACM Manag. Data},
  title      = {{U-DPAP}: Utility-aware Efficient Range Counting on Privacy-preserving Spatial Data Federation},
  year       = {2025},
  month      = feb,
  number     = {1},
  volume     = {3},
  address    = {New York, NY, USA},
  articleno  = {83a},
  doi        = {10.1145/3714333},
  issue_date = {February 2025},
  numpages   = {25},
  publisher  = {Association for Computing Machinery},
}

@Article{Ahmed2025VLDB,
  author     = {Ahmed, Haseeb and Rao, Nachiket and Kati, Abdelkarim and Kerschbaum, Florian and Maiyya, Sujaya},
  journal    = {Proc. VLDB Endow.},
  title      = {{OasisDB}: An Oblivious and Scalable System for Relational Data},
  year       = {2025},
  issn       = {2150-8097},
  month      = jul,
  number     = {11},
  pages      = {4478–4491},
  volume     = {18},
  doi        = {10.14778/3749646.3749707},
  issue_date = {July 2025},
  numpages   = {14},
  publisher  = {VLDB Endowment},
}

@InProceedings{perez2025USENIX,
  author    = {P{\'e}rez, Carolina Ortega and Daffalla, Alaa and Ristenpart, Thomas and Tech, Cornell},
  booktitle = {USENIX Security Symposium},
  title     = {Encrypted Access Logging for Online Accounts: Device Attributions without Device Tracking},
  year      = {2025},
}

@InProceedings{Sen2025CCS,
  author    = {Sen, Pritam and Ma, Yao and Borcea, Cristian},
  booktitle = {Proceedings of the 2025 ACM SIGSAC Conference on Computer and Communications Security},
  title     = {CryptGNN: Enabling Secure Inference for Graph Neural Networks},
  year      = {2025},
  address   = {New York, NY, USA},
  pages     = {291–305},
  publisher = {Association for Computing Machinery},
  series    = {CCS '25},
  doi       = {10.1145/3719027.3765232},
  isbn      = {9798400715259},
  location  = {Taipei, Taiwan},
  numpages  = {15},
}

@Article{Yao2025TIFS,
  author   = {Yao, Lisha and Weng, Jian and Wu, Pengfei and Chen, Shixin and Sun, Jianfei and Yang, Guomin and Deng, Robert H.},
  journal  = {IEEE Transactions on Information Forensics and Security},
  title    = {Breaking the Trilemma: Toward Efficient, Privacy-Preserving, and Forward-Secure Data Sharing in the Post-Quantum Era},
  year     = {2025},
  pages    = {13370-13385},
  volume   = {20},
  doi      = {10.1109/TIFS.2025.3637717},
}

@Misc{URLIBM,
  title = {{IBM} Reduces the Co-Selling Lifecycle by 90\% and Boosts Sales Opportunities with {AWS} by 117\% with {ACE} {CRM} Integration Using {Labra} Platform},
  year  = {2024},
  url   = {https://aws.amazon.com/partners/success/ibm-labra/},
}

@Misc{URLUber,
  title = {Modernizing {Uber}’s Batch Data Infrastructure with Google Cloud Platform},
  year  = {2024},
  url   = {https://www.uber.com/en-HK/blog/modernizing-ubers-data-infrastructure-with-gcp/},
}

@Misc{URLTwitter,
  title = {{US} Election 2020 Tweets},
  year  = {2020},
  url   = {https://www.kaggle.com/datasets/manchunhui/us-election-2020-tweets},
}

@Misc{URLNewYork,
  title = {{OSM} extracts for {Paris}},
  year  = {2024},
  url   = {https://download.bbbike.org/osm/bbbike/NewYork/},
}

@Misc{URLParis,
  title = {{OSM} extracts for {New York}},
  year  = {2024},
  url   = {https://download.bbbike.org/osm/bbbike/Paris/},
}

@Article{Fang2024VLDB,
  author     = {Fang, Wenjing and Cao, Shunde and Hua, Guojin and Ma, Junming and Yu, Yongqiang and Huang, Qunshan and Feng, Jun and Tan, Jin and Zan, Xiaopeng and Duan, Pu and Yang, Yang and Wang, Li and Zhang, Ke and Wang, Lei},
  journal    = {Proc. VLDB Endow.},
  title      = {{SecretFlow-SCQL}: A Secure Collaborative Query Platform},
  year       = {2024},
  issn       = {2150-8097},
  month      = aug,
  number     = {12},
  pages      = {3987–4000},
  volume     = {17},
  doi        = {10.14778/3685800.3685821},
  issue_date = {August 2024},
  numpages   = {14},
  publisher  = {VLDB Endowment},
  url        = {https://doi.org/10.14778/3685800.3685821},
}

@Misc{Bian2024mayfly,
  author        = {Christopher Bian and Albert Cheu and Stanislav Chiknavaryan and Zoe Gong and Marco Gruteser and Oliver Guinan and Yannis Guzman and Peter Kairouz and Artem Lagzdin and Ryan McKenna and Grace Ni and Edo Roth and Maya Spivak and Timon Van Overveldt and Ren Yi},
  title         = {Mayfly: Private Aggregate Insights from Ephemeral Streams of On-Device User Data},
  year          = {2024},
  archiveprefix = {arXiv},
  eprint        = {2412.07962},
  primaryclass  = {cs.CR},
  url           = {https://arxiv.org/abs/2412.07962},
}

\end{document}